\documentclass[11pt, letterpaper]{article}
\usepackage{fullpage}
\usepackage[T1]{fontenc}
\usepackage{color}
\usepackage{babel}
\usepackage{verbatim}
\usepackage{float}
\usepackage{amsmath}
\usepackage{amsthm}
\usepackage{amssymb}
\usepackage{wasysym}
\usepackage[unicode=true,pdfusetitle,
 bookmarks=true,bookmarksnumbered=false,bookmarksopen=false,
 breaklinks=false,pdfborder={0 0 0},pdfborderstyle={},backref=page,colorlinks=true]
 {hyperref}
\hypersetup{
 linkcolor=blue, urlcolor=marineblue, citecolor=red}

\makeatletter

\newcommand{\lyxmathsym}[1]{\ifmmode\begingroup\def\b@ld{bold}
  \text{\ifx\math@version\b@ld\bfseries\fi#1}\endgroup\else#1\fi}

\providecommand{\tabularnewline}{\\}
\floatstyle{ruled}
\newfloat{algorithm}{tbp}{loa}
\providecommand{\algorithmname}{Algorithm}
\floatname{algorithm}{\protect\algorithmname}

\theoremstyle{plain}
\newtheorem{thm}{\protect\theoremname}
\theoremstyle{definition}
\newtheorem{defn}[thm]{\protect\definitionname}
\theoremstyle{plain}
\newtheorem{lem}[thm]{\protect\lemmaname}
\theoremstyle{plain}
\newtheorem{cor}[thm]{\protect\corollaryname}

\@ifundefined{date}{}{\date{}}
\usepackage{algorithm,algorithmicx,algpseudocode}

\makeatother

\providecommand{\corollaryname}{Corollary}
\providecommand{\definitionname}{Definition}
\providecommand{\lemmaname}{Lemma}
\providecommand{\theoremname}{Theorem}

\begin{document}
\title{Interior Point Methods with a Gradient Oracle}
\author{Adrian Vladu\thanks{CNRS \& IRIF, Université Paris Cité, \texttt{vladu@irif.fr}}}
\maketitle
\begin{abstract}\sloppy
We provide an interior point method based on quasi-Newton iterations,
which only requires first-order access to a strongly self-concordant
barrier function. To achieve this, we extend the techniques of Dunagan-Harvey
[STOC '07] to maintain a preconditioner, while using only first-order
information. We measure the quality of this preconditioner in terms
of its relative excentricity to the unknown Hessian matrix, and we
generalize these techniques to convex functions with a slowly-changing
Hessian. We combine this with an interior point method to show that,
given first-order access to an appropriate barrier function for a
convex set $K$, we can solve well-conditioned linear optimization
problems over $K$ to $\varepsilon$ precision in time $\widetilde{O}\left(\left(\mathcal{T}+n^{2}\right)\sqrt{n\nu}\log\left(1/\varepsilon\right)\right)$,
where $\nu$ is the self-concordance parameter of the barrier function,
and $\mathcal{T}$ is the time required to make a gradient query.

As a consequence we show that:
\begin{itemize}
\item Linear optimization over $n$-dimensional convex sets can be solved
in time $\widetilde{O}\left(\left(\mathcal{T}n+n^{3}\right)\log\left(1/\varepsilon\right)\right)$.
This parallels the running time achieved by state of the art algorithms
for cutting plane methods, when replacing separation oracles with
first-order oracles for an appropriate barrier function. 
\item We can solve semidefinite programs involving $m\geq n$ matrices in
$\mathbb{R}^{n\times n}$ in time $\widetilde{O}\left(mn^{4}+m^{1.25}n^{3.5}\log\left(1/\varepsilon\right)\right)$,
improving over the state of the art algorithms, in the case where
$m=\Omega\left(n^{\frac{3.5}{\omega-1.25}}\right)$.
\end{itemize}
Along the way we develop a host of tools allowing us to control the
evolution of our potential functions, using techniques from matrix
analysis and Schur convexity.
\end{abstract}

\newpage 
\sloppy
\global\long\def\vzero{\boldsymbol{\mathit{0}}}\global\long\def\vx{\boldsymbol{\mathit{x}}}\global\long\def\vb{\boldsymbol{\mathit{b}}}\global\long\def\vv{\boldsymbol{\mathit{v}}}\global\long\def\vu{\boldsymbol{\mathit{u}}}\global\long\def\vr{\boldsymbol{\mathit{r}}}\global\long\def\vDelta{\boldsymbol{\mathit{\Delta}}}\global\long\def\vz{\boldsymbol{\mathit{z}}}\global\long\def\vs{\boldsymbol{\mathit{s}}}\global\long\def\vrho{\boldsymbol{\mathit{\rho}}}\global\long\def\vdelta{\boldsymbol{\mathit{\delta}}}\global\long\def\vDeltat{\boldsymbol{\widetilde{\mathit{\Delta}}}}\global\long\def\vc{\boldsymbol{\mathit{c}}}\global\long\def\vh{\boldsymbol{\mathit{h}}}\global\long\def\vup{\boldsymbol{\mathit{u'}}}\global\long\def\vrp{\boldsymbol{\mathit{r'}}}

\global\long\def\vxh{\boldsymbol{\mathit{\widehat{x}}}}\global\long\def\vbh{\boldsymbol{\mathit{\widehat{b}}}}\global\long\def\vxhp{\boldsymbol{\mathit{\widehat{x}'}}}\global\long\def\vxp{\boldsymbol{\mathit{x'}}}\global\long\def\vxs{\boldsymbol{\mathit{x^{\star}}}}\global\long\def\vsp{\boldsymbol{\mathit{s'}}}\global\long\def\vy{\boldsymbol{\mathit{y}}}\global\long\def\vyp{\boldsymbol{\mathit{y'}}}\global\long\def\vg{\boldsymbol{\mathit{g}}}\global\long\def\vrt{\boldsymbol{\mathit{\widetilde{r}}}}\global\long\def\vrtp{\boldsymbol{\mathit{\widetilde{r}'}}}\global\long\def\vys{\boldsymbol{\mathit{y^{\star}}}}

\global\long\def\vrh{\boldsymbol{\mathit{\widehat{r}}}}\global\long\def\vrhp{\boldsymbol{\mathit{\widehat{r}'}}}

\global\long\def\exc{\mathcal{E}}

\global\long\def\HH{\boldsymbol{\mathit{H}}}\global\long\def\HHtil{\boldsymbol{\mathit{\widetilde{H}}}}\global\long\def\XX{\boldsymbol{\mathit{X}}}\global\long\def\XXp{\boldsymbol{\mathit{X'}}}\global\long\def\Id{\boldsymbol{\mathit{I}}}\global\long\def\PP{\boldsymbol{\mathit{P}}}\global\long\def\YY{\boldsymbol{\mathit{Y}}}

\global\long\def\HHh{\boldsymbol{\mathit{\widehat{H}}}}\global\long\def\HHb{\boldsymbol{\mathit{\overline{H}}}}\global\long\def\HHt{\boldsymbol{\mathit{\widetilde{H}}}}\global\long\def\HHtp{\boldsymbol{\mathit{\widetilde{H}'}}}\global\long\def\HHp{\boldsymbol{\mathit{H'}}}

\global\long\def\AA{\boldsymbol{A}}\global\long\def\DD{\boldsymbol{D}}\global\long\def\MM{\boldsymbol{M}}\global\long\def\RR{\boldsymbol{R}}\global\long\def\SS{\mathit{\boldsymbol{S}}}\global\long\def\SSp{{\it \mathit{\boldsymbol{S'}}}}\global\long\def\BB{\boldsymbol{B}}\global\long\def\CC{\boldsymbol{C}}\global\long\def\XXs{\boldsymbol{X^{\star}}}

\global\long\def\PProj{\mathit{\boldsymbol{\Pi}}}

\global\long\def\ks{\kappa_{\star}}

\global\long\def\diag#1{\mathbb{D}\left(#1\right)}\global\long\def\epsilon{\varepsilon}\global\long\def\ln{\log}

 \section{Introduction}

Gradient methods are the most basic elements of continuous optimization
\cite{ben2001lectures,nesterov2018lectures,bubeck2015convex}. While
they have been central to many important results in the theory of
minimizing convex functions, especially due to their simplicity, they
are typically unable to provide high precision solutions unless the
functions they are used are extremely well behaved (smooth and strongly-convex).
In addition to this, their behavior can differ wildly when performing
a change of basis, so finding the appropriate basis turns out to be
a challenging task that may have drastic effects on their convergence
\cite{kelner2013simple}. However, for many important scenarios a
single good basis can simply not exist. 

In particular, when considering linear optimization objectives over
general convex sets, methods that are able to provide high precision
solutions are cutting plane methods (CPM) or interior point methods
(IPM). These overcome the limitations of standard gradient methods,
as they all implicitly maintain a changing basis, which permit very
fast convergence to a near minimizer of the objective. Doing so is,
however, very costly. In the case of cutting plane methods, one needs
to carefully maintain an appropriate center of the current domain,
which needs to be updated fast when adding new constraints, while
in the case of interior point methods one needs access to the Hessian
matrix of an appropriate barrier function, whose mere evaluation may
be extremely costly \cite{nesterov1992conic,jiang2020faster}. 

In the context of interior point methods, several works have tried
to bridge the gap between the standard version, which requires access
to the Hessian matrix, and that where only gradients are available,
by resorting to quasi-Newton methods. These attempt to mimic the classical
Newton iterations performed in IPM's (which can be thought of as gradient
steps, preconditioned with the Hessian matrix) by maintaining a ``fake''
Hessian matrix which gets corrected whenever the step obtained using
it certifiably points towards a wrong direction \cite{gondzio2019quasi,gondzio2022polynomial,tunccel2001generalization}.
A related line of remarkable work has produced efficient quasi-Newton
methods \cite{broyden1970convergence,fletcher1970new,goldfarb1970family,shanno1970conditioning,byrd1994representations},
which aim to match the performance of second order methods only by
using a gradient oracle. While extremely promising, it is unclear
to what extent these methods yield good convergence bounds when applied
to standard tasks, such as those modeled by linear or semidefinite
programs. Addressing these questions will yield new and improved algorithms,
opening a fresh research direction for efficient optimization. Relatedly,
Dunagan and Harvey \cite{dunagan2007iteratively} provided a beautiful
method for solving linear systems, partially inspired from the conjugate
gradient method, which was based on maintaining a dynamic preconditioner.
While they explicitly stated the possibility of using their algorithm
inside quasi-Newton methods, for strange reasons this direction has
not been pursued until now.

As an important benchmark for optimization algorithms, we consider
the semidefinite programs (SDPs), which optimize a linear objective
over the intersection between an affine space and the cone of positive
semidefinite matrices. These have broad applications in multiple scientific
fields, including theoretical computer science, operations research,
and engineering \cite{vandenberghe1996semidefinite}. 

We formally define semidefinite programming with $n\times n$ variables
and $m$ constraints.
\begin{defn}
[Semidefinite programming]Given symmetric matrices $\BB,\AA_{1},\dots,\AA_{m}\in\mathbb{R}^{n\times n}$,
and $c_{i}\in\mathbb{R}$ for all $i\in\left\{ 1,\dots,m\right\} $
we aim to solve the optimization problem:
\begin{equation}
\max\left\{ \left\langle \BB,\XX\right\rangle :\XX\succeq0,\left\langle \AA_{i},\XX\right\rangle =c_{i},\text{ for all }1\leq i\leq m\right\} \,,\label{eq:sdp-primal}
\end{equation}
where $\left\langle \AA,\BB\right\rangle :=\sum_{i,j}\AA_{ij}\BB_{ij}$
is the trace product.
\end{defn}

Standard methods for solving SDPs to high precision rely either on
cutting plane methods or interior point methods. While the early works,
as well as the more recent improvements to the running time of SDP
solvers have used cutting planes methods \cite{khachiyan1979polynomial,lee2015faster,jiang2020improved},
a recent trend has been to seek further improvements by using interior
point methods instead \cite{jiang2020faster,huang2021solving}. As
stated before, one major bottleneck in obtaining fast algorithms via
IPMs for certain optimization problems consists of the unyieldingly
large time required to evaluate the Hessian matrix, and this is precisely
one of the regimes where this obstacle occurs. While \cite{nesterov1992conic}
use a series of clever tricks to speed up the time to compute the
Hessian, \cite{jiang2020faster} develop a series of sophisticated
techniques based on rectangular matrix multiplication. It therefore
appears that quasi-Newton methods, which do not need access to the
true Hessian matrix, could possibly represent a valid path towards
obtaining faster and more practical algorithms for semidefinite programming.

Hence the question we address in this paper is:
\begin{center}
\textit{Can we obtain efficient interior point methods that rely only
on gradient information?}
\par\end{center}

\subsection{Our Results}

We present a fast interior point method for solving linear optimization
problems over convex sets, when only gradient access to a barrier
function for these sets is available. In this paper we rely on a slightly
stronger notion of barriers functions than the one used in standard
interior point literature, namely \textit{strongly self-concordant
barriers} \cite{laddha2020strong}. We provide a more extensive overview
of these in Section \ref{subsec:Strongly-Self-Concordant-Barrier}.
Strong self-concordance is a property implicitly used in many recent
works \cite{cohen2021solving}, and is shared by standard barrier
functions such as the logarithmic barrier \cite{jiang2020faster},
the universal barrier \cite{nesterov2018lectures,guler1996barrier},
or the entropic barrier \cite{bubeck2014entropic}. As shown in \cite{laddha2020strong},
for the latter two cases, strong self-concordance is a property that
follows from recent developments on the KLS conjecture \cite{chen2021almost,jambulapati2022slightly}.
The number of iterations of an interior point method depends on the
quality of the barrier function, captured by the self-concordance
parameter (see Definition \ref{def:nu-strongly-concordant}). The
self-concordance parameter of the universal and entropic barriers
for a set in $\mathbb{R}^{n}$ is $\widetilde{O}\left(n\right)$.

Our main result is the following theorem.
\begin{thm}
[informal] \label{thm:main-informal}Let $K\subseteq\mathbb{R}$
be a convex set. There is an interior point method that solves linear
optimization problems over $K$ in time $\widetilde{O}\left(\sqrt{\nu n}\left(n^{2}+\mathcal{T}_{\text{gradient}}\right)\right)$,
where $\mathcal{T}_{\text{gradient}}$ is the time required to evaluate
the gradient of a $\nu$-strongly self-concordant barrier function
for $K$. Given gradient access to the universal or entropic barriers,
the running time is $\widetilde{O}\left(n^{3}+n\mathcal{T}_{\text{gradient}}\right)$.
\end{thm}

\paragraph{The gradient complexity of interior point methods.}

Interestingly, our running time parallels the ones achieved by state-of-the-art
cutting plane methods, when replacing the gradient oracle with a separation
oracle for $K$. While these two oracles are incomparable, our result
contributes to an exciting direction towards understanding the query
complexity of optimization problems under different access models
\cite{blikstad2022nearly}. Note that naively, one would expect a
gradient query complexity of $\widetilde{O}\left(n^{3/2}\right)$,
as running $\widetilde{O}\left(\sqrt{n}\right)$ iterations of an
interior point method would naively require $O\left(n\right)$ gradient
queries per iteration, in order to approximate the Hessian matrix
of the barrier function. While recent works \cite{cohen2021solving,lee2019solving,van2020solving,van2020deterministic,jiang2020faster,huang2021solving}
have exhaustively leveraged the stability of the Hessian matrix across
iterations to obtain running time improvements for interior-point
methods, it is unclear how to use these in our setting, as their amortized
analysis crucially uses the structure of the underlying Hessian matrix.

This raises several interesting questions concerning the adaptive
query complexity of barrier optimization. Notably, it would be interesting
to know whether the parallel query complexity can be reduced below
$\widetilde{O}\left(n\right)$, and understand the trade-offs between
the number of parallel rounds and the total number of queries. In
particular, achieving a parallel round complexity of $o\left(\sqrt{n}\right)$
is likely to have broad consequences, as it would plausibly lead to
more efficient interior-point methods, whenever fast solvers for structured
linear systems are available. 

\begin{table}[h]\small
\begin{centering}
\begin{tabular}{|c|c|c|c|c|}
\hline 
Year & References & Method & \#Iters & Cost per iter\tabularnewline
\hline 
2015 & \cite{lee2015faster} & CPM & $n$ & $n^{2}+\mathcal{T}_{\text{cutting plane}}$\tabularnewline
\hline 
2020 & \cite{jiang2020improved} & CPM & $n$ & $n^{2}+\mathcal{T}_{\text{cutting plane}}$\tabularnewline
\hline 
2022 & Our Result & IPM & $n$ & $n^{2}+\mathcal{T}_{\text{gradient}}$\tabularnewline
\hline 
\end{tabular}
\par\end{centering}
\caption{Comparison of linear optimization algorithms over convex sets $K\subseteq\mathbb{R}^{n}$.
CPM stands for cutting plane method, and IPM for interior point method.
$\mathcal{T}_{\text{cutting plane}}$ is the time required to compute
a cutting plane, while $\mathcal{T}_{\text{gradient}}$ represents
the time required to evaluate the gradient of an barrier function
for $K$.}
\end{table}

\paragraph{Semidefinite programming.}

Using Theorem \ref{thm:main-informal} we obtain the following result
on solving general semidefinite programs, which follows from using
a standard logarithmic barrier on the SDP cone.
\begin{thm}
[informal] There is an interior point method that solves a general
SDP with variable size $n\times n$ and $m$ constraints in time $\widetilde{O}\left(mn^{4}+m^{1.25}n^{3.5}\right)$.
\end{thm}

Our running time can be interpreted as follows: $n$ is the number
of iterations of the method, $mn^{2}$ is the input size, $mn^{2}+n^{\omega}$
is the time to evaluate on gradient of the barrier function, which
requires computing a weighted sum of the input matrices, and computing
its inverse using fast matrix multiplication. We note that when $m\geq n$,
which holds in the case of most standard SDP applications, the running
time is $\widetilde{O}\left(mn^{4}+m^{1.25}n^{3.5}\right)$, even
when using naive matrix inversion in $O\left(n^{3}\right)$. This
may be an advantage for applications where matrix multiplication algorithms
with faster than $O\left(n^{3}\right)$ theoretical running require
extensive tuning and dedicated hardware \cite{huang2018practical}.

\begin{table}[h]\small
\begin{centering}
\begin{tabular}{|c|c|c|c|p{42mm}|c|}
\hline 
Year & References & Method & \#Iters & Cost per iter & Cost when $m=n^{4}$\tabularnewline
\hline 
1992 & \cite{nesterov1992conic} & IPM & $\sqrt{n}$ & $m^{2}n^{2}+mn^{\omega}+m^{\omega}$ & $n^{10.5}$\tabularnewline
\hline 
2015 & \cite{lee2015faster} & CPM & $m$ & $mn^{2}+m^{2}+n^{\omega}$ & $n^{10}$\tabularnewline
\hline 
2020 & \cite{jiang2020improved} & CPM & $m$ & $mn^{2}+m^{2}+n^{\omega}$ & $n^{10}$\tabularnewline
\hline 
2020 & \cite{jiang2020faster} & IPM & $\sqrt{n}$ & $mn^{2}+m^{\omega}+n^{\omega}$ & $n^{4\omega+1/2}$\tabularnewline
\hline 
2021 & \cite{huang2021solving} & IPM & $\sqrt{n}$ & $m^{2}+n^{4}+\left(m^{\omega}+n^{2\omega}\right)n^{-1/2}$ \newline (amortized) & $n^{4\omega}$\tabularnewline
\hline 
2022 & Our Result & IPM & $n^{2}+n^{3/2}m^{1/4}$ & $mn^{2}+n^{\omega}$  & $n^{8.5}$\tabularnewline
\hline 
\end{tabular}
\par\end{centering}
\caption{Summary of key SDP algorithms. CPM stands for cutting plane method,
and IPM for interior point method. $n$ is the size of the variable
matrix, and $m$ is the number of constraints. Runtimes hide $n^{o\left(1\right)},m^{o\left(1\right)}$
and $\text{polylog}\left(1/\epsilon\right)$ factors, where $\epsilon$
is the accuracy parameter, as well as factors depending on the bit
complexity.}
\end{table}

\subsection{Our Techniques}

To achieve our results we build on the ideas of Dunagan and Harvey
\cite{dunagan2007iteratively} to maintain a dynamic preconditioner
for the linear systems that show up when running an interior point
method. In the case of strongly self-concordant barriers these systems
are strongly related to each other, so we can recycle the preconditioner
across multiple iterations of the interior point method.

\paragraph{Linear system solving with an adaptive preconditioner.}

In Section \ref{sec:linear} we review the main ideas from \cite{dunagan2007iteratively},
as we will build on this framework to develop our gradient-based interior
point method. The basic idea is as follows: when one attempts to solve
a linear system $\HH\vx=\vb$ using the Richardson iteration, a standard
measure of progress is the residual norm $\left\Vert \vb-\HH\vx\right\Vert _{2}$.
After running the Richardson iteration for a single step, the change
in the norm of the residual faces two possibilities. In the first
case, it decreases by a constant factor, which is the ideal scenario,
since this allows a very fast convergence towards the optimum. In
the second case, the decrease is small. This, however, can only happen
because the underlying matrix has some very large or very small eigenvalues,
and our attempted descent direction happens to align well with these.
Therefore, this situation constitutes a certificate for the existence
of very large or very small eigenvalues. We can quantify this, by
showing that we can use these directions to build a preconditioner,
which will prevent such occurrences in the future.

The main idea developed by Dunagan and Harvey is that this preconditioner
can be maintained by performing rank-$1$ updates involving the direction
that certifies the existence of large/small eigenvalues. While they
can not directly show that this permanently reduces the impact of
the extreme eigenvalues certified by this direction, they make progress
in the sense that a certain potential function called \textit{excentricity},
which measures the quality of our current preconditioner $\HHt$,
gets reduced. Hence, in each step one either manages to reduce the
gradient norm or excentricity by a constant factor. This automatically
yields an upper bound on the number of iterations which depends on
the quality of the preconditioner when the method is initialized.

We formally recover this argument in Section \ref{sec:linear}, and
obtain a simple algorithm entirely based on preconditioned Richardson.
For completeness, we recover the convergence guarantees from \cite{dunagan2007iteratively}
using our simpler algorithm. 

We stress the importance of our contribution, since \cite{dunagan2007iteratively}
uses a slightly different set of updates, for which theoretical guarantees
are difficult to prove in the situation where quantities involving
the matrix can only be approximated (as we do in Section \ref{sec:nonlinear}). 

\paragraph{Newton steps on convex functions with a slowly moving Hessian.}

In Section \ref{sec:nonlinear} we extend the analysis from Section
\ref{sec:linear} to the setting where access to the matrix $\HH$
is restricted. While the analysis from Section \ref{sec:nonlinear}
only uses matrix-vector products, in the setting we care about we
are even more constrained, as this matrix represents the Hessian of
a convex function $g$, which can only be accessed through gradient
queries. To handle this difficulty, we approximate matrix-vector products
using two gradient queries, since intuitively:
\[
\HH_{\vy}\vv\approx\frac{1}{\tau}\left(\nabla g\left(\vy+\tau\vv\right)-\nabla g\left(\vy\right)\right)\,,
\]
for some appropriate step size $\tau$. To do so it is important to
consider the conditioning of $\HH_{\vy}$, as obtaining an accurate
estimate requires setting $\tau$ sufficiently large. However, we
do not want $\tau$ to be extremely large, as its magnitude determines
among others the number of bits of precision required to evaluate
this approximation. We therefore need to provide robust versions of
the algorithms from Section \ref{sec:linear}, and show that they
satisfy similar guarantees, even though we only access the function
$g$ through gradient queries. 

While the proof roughly follows the same lines as before, it is significantly
more involved, as even picking the right step size for preconditioned
Richardson requires estimating quantities that are not directly available.
Furthermore, since we aim to be mindful about the size of the bit
representation of our numbers, we always control the range of the
preconditioner's eigenvalues. These technical difficulties require
us to very carefully analyze the produced excentricity certificates.
The main result of the section is given in the Robust Step-or-Update
Lemma (Lemma \ref{lem:robust-step-or-update}), and can be thought
of as a Newton step with an adaptive preconditioner $\HHt$, which
either performs a step that reduces the ``fake'' dual local norm
$\left\Vert \nabla g\left(\vy\right)\right\Vert _{\HHt^{-1}}$ , or
exhibits an excentricity certificate which we use to reduce the excentricity
potential.

\paragraph{Interior point method with an adaptive preconditioner.}

The main contribution of our paper lies in using adaptive preconditioning
inside interior point methods. In Section \ref{sec:path-following}
we describe the main path-following method. To do so, we approximately
follow the central path constituted by minimizers for the family of
functions
\[
g_{\mu}\left(\vy\right)=\frac{\left\langle \vc,\vy\right\rangle }{\mu}+\phi\left(\vy\right)\,,
\]
where $\phi$ is a $\nu$-strongly self-concordant barrier function
for $K\subseteq\mathbb{R}^{n}$ (see Section \ref{subsec:Strongly-Self-Concordant-Barrier}).
In classical interior point methods, path following consists of generating
a sequence of iterates $\vy$ which are close to the \textit{central
path}, which represents the set of unique minimizers of $g_{\mu}$
for all $\mu\in\left(0,\infty\right)$. Closeness to central path
is measured in terms of the dual local norm of the gradients $\left\Vert \nabla g_{\mu}\left(\vy\right)\right\Vert _{\HH_{\mu}^{-1}}$.
Classically, one generates these iterates by maintaining the invariant
that $\left\Vert \nabla g_{\mu}\left(\vy\right)\right\Vert _{\HH_{\mu}^{-1}}\leq O\left(1\right)$,
dialing down $\mu$ by some factor, and restoring closeness to the
central path by executing one (or a few) Newton steps, corresponding
to a sequence of linear system solves. The amount by which $\mu$
gets dialed down is determined by the self-concordance parameter $\nu$.
Typically, one moves from $\mu$ to $\mu'=\mu/\left(1+O\left(1/\sqrt{\nu}\right)\right)$,
which determines the total number of iterations of the path-following
method to be $\widetilde{O}\left(\sqrt{\nu}\right)$. 

In this work, we slow down the interior point method, by performing
significantly shorter steps. Specifically we dial down the centrality
parameter only by a factor of $1+O\left(1/\sqrt{n\nu}\right)$ and
we maintain the invariant that $\left\Vert \nabla g_{\mu}\left(\vy\right)\right\Vert _{\HH_{\mu}^{-1}}\leq O\left(1/\sqrt{n}\right)$.
The reason for doing this is that between two subsequent iterates
$\vy$ and $\vyp$ produced by path following methods, the Hessian
matrix changes significantly. Unfortunately, too large a change, which
is likely to occur when using standard step sizes, makes the preconditioner
significantly worse, so the excentricity between the preconditioner
$\HHt$ and the two Hessians $\HH_{\vy}$, and $\HH_{\vyp}$, respectively,
can increase by a lot. Using very short steps guarantees that this
is not the case, and allows us to match the increases in excentricity
that occur when changing the iterate with the decreases in excentricity
caused by rank-1 updates of the preconditioner.

An important feature of the excentricity potential is that while changes
caused by rank-1 updates of the preconditioner are easy to control,
the changes caused by modifying the preconditioned matrix are generally
not. In particular using the notation defined in Section \ref{sec:prelim}
we aim to compare $\exc\left(\HHt^{-1}\HH_{\vy}\right)$ against $\exc\left(\HHt^{-1}\HH_{\vyp}\right)$. 

We are able to prove upper bounds on the increase in excentricity
by leveraging properties of strongly self-concordant barriers. To
do so we show, using majorization techniques and Schur convexity,
that the change in excentricity depends on the eigenvalues of $\HH_{\vy}^{-1}\HH_{\vyp}$.
These can be tightly controlled using strong self-concordance. We
prove this in Section \ref{subsec:excent-proofs}.

We conclude that with each step along the central path excentricity
increases by at most a constant factor, while each rank-1 update caused
by an excentricity certificate found by our preconditioned Richardson
iteration causes decrease in excentricity by at least a constant factor.
Therefore, when properly initialized, our interior point method only
requires $\widetilde{O}\left(\sqrt{n\nu}\ln\frac{1}{\epsilon}\right)$
rank-1 updates to the preconditioner $\HHt$. Since all the rank-1
updates on $\HHt$ and $\HHt^{-1}$, which we explicitly maintain
using the Sherman-Morrison formula, take $O\left(n^{2}\right)$ time,
we bound the total running time by 
\[
\widetilde{O}\left(\left(n^{2}+\mathcal{T}_{\text{gradient}}\right)\sqrt{n\nu}\ln\frac{1}{\epsilon}\right)\,,
\]
where $\mathcal{T}_{\text{gradient}}$ is the time required to perform
a query on the gradient oracle which on input $\vy$ returns $\nabla\phi\left(\vy\right)$.

The formal statement of the theorem is given in Theorem \ref{thm:main-opt}. 

As a corollary we see that when given access to a gradient oracle
for an $\widetilde{O}\left(n\right)$-strongly self-concordant barrier
for the domain $K$, we can perform linear optimization in time $\widetilde{O}\left(n^{3}+n\mathcal{T}_{\text{gradient}}\ln\frac{1}{\epsilon}\right)$
time. We give the formal statement in Corollary \ref{cor:main-cor}.

\paragraph{SDP via interior point methods.}

For our SDP application, instead of directly solving the objective
(\ref{eq:sdp-primal}), just like in \cite{jiang2020faster}, we consider
the dual problem
\begin{equation}
\min\left\{ \left\langle \vc,\vy\right\rangle :\BB-\sum_{i=1}^{m}\vy_{i}\AA_{i}\succeq0\right\} \,,\label{eq:sdp-dual-formulation}
\end{equation}
which we optimize by solving the barrier formulation 
\begin{equation}
g_{\mu}\left(\vy\right)=\frac{\left\langle \vc,\vy\right\rangle }{\mu}+\phi\left(\vy\right)\,,\label{eq:sdp-barrier-obj}
\end{equation}
where $\phi$ is a strongly self-concordant barrier function for the
set $\left\{ \vy:\BB-\sum_{i=1}^{m}\vy_{i}\AA_{i}\succeq0\right\} $.
Unfortunately, the standard $\log\det$ barrier does not satisfy the
required strong self-concordance property required by our potential
function analysis. We fix this by scaling it by a factor of $\sqrt{m}$,
which in exchange further slows down the interior point method. Specifically,
in our case we use 
\[
\phi\left(\vy\right)=-\sqrt{m}\cdot\log\det\left(\BB-\sum_{i=1}^{m}\vy_{i}\AA_{i}\right)\,,
\]
for which we can show that it satisfies the required strong self-concordance
property with $\nu=n\sqrt{m}$. To evaluate gradients of the barrier
it suffices to compute
\[
\left[\nabla\phi\left(\vy\right)\right]_{i}=\left\langle \AA_{i},\left(\BB-\sum_{i=1}^{m}\vy_{i}\AA_{i}\right)^{-1}\right\rangle 
\]
which can be done in time $O\left(mn^{2}+n^{\omega}\right)$. Compared
to previous works on fast SDP solving our crucial advantage is that
we never need to compute the Hessian of $\phi$, which involves very
costly matrix multiplication, and hence additional running time dependence
in $m$. In fact \cite{nesterov1992conic} use a series of clever
tricks to compute the Hessian matrix in time $O\left(mn^{\omega}+m^{2}n^{2}+m^{\omega}\right)$,
while \cite{jiang2020faster} develop a series of techniques based
on rectangular matrix multiplication to reduce this running time to
the time required to multiply an $m\times n^{2}$ matrix by an $n^{2}\times m$
matrix. Fortunately, we can completely avoid these bottlenecks by
using our dynamic preconditioner.

Together with our analysis on gradient-based interior point methods
we obtain our final running time of $\widetilde{O}\left(mn^{4}+m^{1.25}n^{3.5}\right)$.
We provide a formal statement in Theorem \ref{thm:sdp}.

\subsection{Related Work}

The literature on optimization methods is truly extensive. Here we
summarize a few relevant results in literature.

\paragraph{Cutting plane methods.}

A class of optimization techniques known as cutting plane methods
repeatedly refine a convex set that contains the sought solution,
via queries to a separation oracle. Designing efficient cutting plane
algorithms has been a long-running effort since its introduction in
the 1950s \cite{shor1977cut,yn76,k80,kte88,nn89,v89,av95,bv02,lee2015faster,jiang2020improved}.

\paragraph{Interior point methods.}

Interior point methods are fundamental optimization techniques for
minimizing linear functions over convex sets, provided first and second
order access to a barrier function for the set. Extensive early work
\cite{frisch1955logarithmic,kojima1989primal,megiddo1989pathways,fiacco1990nonlinear}
has been beautifully explained by Nesterov and Nemirovski \cite{nesterov1994interior},
who developed the theory of self-concordant barriers, and expanded
the area of IPM applications to more general settings, including semidefinite
programming and more broadly conic optimization. Recent developments
have used IPM theory to provide important progress towards reducing
the theoretical running time of linear and semidefinite programming
\cite{lee2014path,cohen2021solving,lee2019solving,van2020solving,van2020deterministic,jiang2020faster,huang2021solving}.

\paragraph{Quasi-Newton methods.}

An important series of developments has focused on quasi-Newton methods,
which attempt to obtain convergence guarantees comparable to those
of Newton's method using only gradient access. A famous set of results
is given by the BFGS and L-BFGS methods \cite{broyden1970convergence,fletcher1970new,goldfarb1970family,shanno1970conditioning,byrd1994representations},
which are used in practical settings. Also see \cite{gondzio2019quasi,gondzio2022polynomial,tunccel2001generalization}
for applications of quasi-Newton methods to interior point solvers.
We note that \cite{gondzio2019quasi,gondzio2022polynomial} use quasi-Newton
steps interspersed with standard Newton steps. Although their results
do not rely solely on gradient information, they provide important
practical speedups since they reduce the total number of linear system
solves. In \cite{tunccel2001generalization}, both zeroth and first
order methods for linear programming are provided, but they are only
shown to converge in finite time, without an explicit bound on the
iteration complexity.

\paragraph{SDP solvers.}

Early SDP solvers have relied on cutting plane methods. Starting with
the works of Nesterov-Nemirovski \cite{nesterov1992conic} and Anstreicher
\cite{anstreicher2000volumetric}, several SDP solvers based on IPMs
have been developed. Recently, new techniques based on inverse matrix
maintenance have provided fast algorithms for SDPs~\cite{jiang2020faster,huang2021solving}.
 
\section{Preliminaries\label{sec:prelim}}

\subsection{Notation}

We write matrices and vectors in bold. We use $\left\langle \cdot,\cdot\right\rangle $
to denote inner products. Given a symmetric matrix $\AA$, we use
$\left\Vert \AA\right\Vert $ to represent the spectral norm of $\AA$,
that is $\left\Vert \AA\right\Vert =\max\left\{ -\lambda_{\min}\left(\AA\right),\lambda_{\max}\left(\AA\right)\right\} $,
where $\lambda_{\min}\left(\AA\right)$ and $\lambda_{\max}\left(\AA\right)$
denote the smallest and largest eigenvalues of $\AA$, respectively.
We use $\left\Vert \AA\right\Vert _{F}$ to represent its Frobenius
norm, that is the $\ell_{2}$ norm of its eigenvalues. Similarly we
use $\left\Vert \AA\right\Vert _{1}$ to denote the $\ell_{1}$ norm
of its eigenvalues. For an arbitrary symmetric matrix $\AA$, we use
$\AA_{\geq1}$ to denote the matrix obtained from $\AA$ by increasing
all of its sub-unitary eigenvalues to $1$, and we define $\AA_{<1}$
analogously.

Given a function $g:K\rightarrow\mathbb{R}$, where $K\subseteq\mathbb{R}^{n}$
is a convex set, we use $\nabla g\left(\vy\right)$ to represent the
gradient of $g$ and $\nabla^{2}g\left(\vy\right)$ for the Hessian
matrix at $\vy$. To simplify notation, we use $\HH_{\vy}:=\nabla^{2}g\left(\vy\right)$,
whenever the meaning is clear from the context. When solving linear
systems, we maintain a preconditioner which typically denote by $\HHt$
.

\subsection{The Excentricity Potential}

One of our main tools for analyzing the method is the excentricity
potential function, which has been introduced by Dunagan and Harvey
\cite{dunagan2007iteratively} in the context of analyzing a certain
version of the conjugate gradient method, without resorting to arguments
based on the best interpolating polynomial. In broad terms, excentricity
measures how good of a preconditioner a matrix $\HHt$ is for some
other matrix $\HH$. In standard analyses, one crucially aims to bound
the spectral norm 
\[
\left\Vert \XX-\Id\right\Vert 
\]
where
\[
\XX=\HH^{-1/2}\HHtil\HH^{-1/2}\ .
\]
Naturally, the closer this is to $0$, the better of a preconditioner
$\HHtil$ is for $\HH$. As this function is not necessarily the easiest
one to work with, Dunagan and Harvey introduced the excentricity potential
defined as 
\[
\exc\left(\XX\right)=\det\left(\frac{\XX^{1/2}+\XX^{-1/2}}{2}\right)=\frac{1}{2^{n}}\frac{\det\left(\XX+\Id\right)}{\sqrt{\det\left(\XX\right)}}\,,
\]
for which one can easily verify that it is minimized when $\XX=\Id$,
in which case $\HHtil$ is a perfect preconditioner. This potential
function enjoys several favorable properties, including the fact that
its value is easy to track after performing rank-$1$ updates on either
$\HHtil$ or $\HH$. In the following lemma we see how how excentricity
evolves after a rank-$1$ update (proof in Appendix \ref{subsec:Proof-of-Lemma-excentricity-update}).
\begin{lem}
\label{lem:excentricity-update}A rank-$1$ update 
\[
\XXp=\XX+\vu\vu^{\top}
\]
determines the multiplicative change in excentricity:
\[
\frac{\mathcal{\exc}\left(\XXp\right)}{\mathcal{\exc}\left(\XX\right)}=\frac{1+\vu^{\top}\left(\Id+\XX\right)^{-1}\vu}{\sqrt{1+\vu^{\top}\XX^{-1}\vu}}\ .
\]
\end{lem}

We also show that if we have access to a vector $\vu$ that aligns
well with very small/large eigenvectors of $\XX$, then we can use
it to obtain a rank-$1$ update that significantly reduces excentricity.
\begin{lem}
\label{lem:excent-upd}If a unit vector $\vu$ satisfies
\begin{enumerate}
\item $\vu^{\top}\XX^{-1}\vu\geq\gamma$, then $\frac{\exc\left(\XX+\vu\vu^{\top}\right)}{\exc\left(\XX\right)}\leq\frac{2}{\sqrt{1+\gamma}}$,
\item $\vu^{\top}\XX\vu\geq\gamma$, then $\frac{\exc\left(\XX-\frac{\XX\vu\vu^{\top}\XX}{1+\vu^{\top}\XX\vu}\right)}{\exc\left(\XX\right)}\leq\frac{2}{\sqrt{1+\gamma}}$.
\end{enumerate}
\end{lem}

We show this in Appendix \ref{subsec:Proof-of-Lemma-excent-upd}.
Crucially, both the proofs and our algorithm relies on the Sherman-Morrison
formula for the inverse of a matrix after performing a rank-1 update.
\begin{lem}
[Sherman-Morrison]\label{lem:shermanmorrison} Suppose $\AA\in\mathbb{R}^{n\times n}$
is an invertible matrix, and $\vu,\vv\in\mathbb{R}^{n}$ are column
vectors. Then $\AA+\vu\vv^{\top}$ is invertible iff $1+\vv^{\top}\AA^{-1}\vu\neq0$.
In this case,
\[
\left(\AA+\vu\vv^{\top}\right)^{-1}=\AA^{-1}-\frac{\AA^{-1}\vu\vv^{\top}\AA^{-1}}{1+\vv^{\top}\AA^{-1}\vu}\,.
\]
\end{lem}

To simplify our proofs, we also state a useful property of excentricity.
\begin{lem}
\label{lem:Excentricity-similarity}Excentricity is invariant under
similarity transformations. Given invertible $\XX$, $\YY$, we have
\[
\exc\left(\XX\right)=\exc\left(\YY\XX\YY^{-1}\right)\,.
\]
Additionally
\[
\exc\left(\XX^{-1}\right)=\exc\left(\XX\right)\,.
\]
\end{lem}

Finally, we strongly rely on the following Lemma, which relies on
an eigenvalue bound that we prove in Section \ref{sec:path-following}
via majorization techniques. It allows us to control the change in
excentricity when the involved matrix changes slightly.
\begin{lem}
\label{lem:excprod}Let $\AA,\BB$ be invertible matrices. Then
\[
\exc\left(\AA\BB\right)\leq\exc\left(\AA\right)\cdot\sqrt{\frac{\det\left(\BB_{\geq1}\right)}{\det\left(\BB_{<1}\right)}}\,.
\]
\end{lem}

\begin{proof}
Using Lemma \ref{lem:eig-ineq}, we write:
\[
\exc\left(\AA\BB\right)=\frac{1}{2^{n}}\frac{\det\left(\AA\BB+\Id\right)}{\sqrt{\det\left(\AA\BB\right)}}\leq\frac{1}{2^{n}}\cdot\frac{\det\left(\AA+\Id\right)\cdot\det\left(\BB_{\geq1}\right)}{\sqrt{\det\left(\AA\right)}\cdot\sqrt{\det\left(\BB\right)}}=\exc\left(\AA\right)\cdot\sqrt{\frac{\det\left(\BB_{\geq1}\right)}{\det\left(\BB_{<1}\right)}}\,.
\]
\end{proof}

\subsection{Richardson Iteration}

The Richardson iteration is probably the most important iterative
method for solving linear systems of equations. Generally, given a
linear system $\HH\vx=\vb$, the iteration consists of steps that
attempt to improve the current iterate by performing iterations of
the form $\vxp=\vx+\eta\left(\vb-\HH\vx\right)$, for some appropriate
step size $\eta$. Denoting the residual by $\vr=\vb-\HH\vx$, this
iteration updates the solution by moving into it a small fraction
of the residual. To improve the convergence rate of this iterative
method, one often uses preconditioning. Namely, by using a matrix
$\HHt$ which approximates $\HH$ for some appropriate notion of approximation,
and whose inverse is available, we can instead run the iteration $\vxp=\vx+\eta\HHt^{-1}\left(\vb-\HH\vx\right)$.
This is equivalent to running the vanilla Richardson iteration on
the original system, after doing a change of basis by letting $\vy=\HHt^{1/2}\vx$
and considering the equivalent system
\[
\HHt^{-1/2}\HH\HHt^{-1/2}\cdot\vy=\HHt^{-1/2}\vb\,.
\]

\subsection{Strongly Self-Concordant Barriers and Interior Point Methods\label{subsec:Strongly-Self-Concordant-Barrier}}

We solve general optimization problems of the form 
\begin{equation}
\min_{\vy\in K}\left\langle \vc,\vy\right\rangle \label{eq:original-problem}
\end{equation}
where $K$ is a convex domain. To do so we resort to a path-following
interior-point method which solves a sequence of barrier objectives,
which are convex minimization problems of the form
\[
g_{\mu}\left(\vy\right)=\frac{\left\langle \vc,\vy\right\rangle }{\mu}+\phi\left(\vy\right)\,,
\]
where $\phi$ is a barrier function for the domain $K$. Each function
$g_{\mu}$ has a unique minimizer. The set of minimizers $\vys_{\mu}=\arg\min g_{\mu}\left(\vy\right)$,
for all $\mu\in\left(0,\infty\right)$ represent the central path
corresponding to the barrier objective. To solve (\ref{eq:original-problem}),
classical interior point literature (Lemma \ref{lem:optimality-near})
shows that it suffices to obtain a near minimizer to $g_{\mu}$ for
a sufficiently small value of the centrality parameter $\mu$. 

To do so we implement a path-following procedure. In each iteration
of the procedure, for a sufficiently small scalar $\delta$, we solve
the minimization problem form $g_{\mu/\left(1+\delta\right)}$by warm
starting it with a (near) optimizer of $g_{\mu}$. The choice of $\delta$
depends on the properties of the barrier function, and hence determine
the speed of convergence of the method. The properties of the barrier
function are crucial to obtaining an efficient algorithm. While the
standard theory of interior point methods uses properties of self-concordant
barrier functions \cite{nesterov1994interior}, in this work we use
a slightly stronger property, namely the strong self-concordance property
\cite{laddha2020strong}.

This gives a slightly more powerful condition on the change in the
Hessian when moving between iterates than vanilla self-concordance.
While not explicitly mentioned in the classical literature, it is,
however, enjoyed by many standard barrier functions including the
logarithmic, universal, and entropic barriers.
\begin{defn}
[strongly self-concordant function]\label{def:strong-self-conc}
Given a convex domain $K$ and $g:K\rightarrow\mathbb{R}$, we say
that a convex function $g$ is strongly self-concordant if for any
$\vy\in K$, and $\vh$:
\[
\left\Vert \HH_{\vy}^{-1/2}\frac{d}{dt}\HH_{\vy+t\vh}\HH_{\vy}^{-1/2}\right\Vert _{F}\leq2\left\Vert \vh\right\Vert _{\HH_{\vy}}\,.
\]
\end{defn}

We recall that standard self-concordance replaces the bound involving
the Frobenius norm with the weaker spectral norm.

\begin{lem}
\label{lem:strongly-self-concordant-property}Given any strongly self-concordant
function $g:K\rightarrow\mathbb{R}$, for any $\vx,\vdelta$, such
that $\vx,\vx+\vdelta\in K$, $\left\Vert \vdelta\right\Vert _{\HH_{\vy}}<1$,
\[
\left\Vert \HH_{\vy}^{-1/2}\left(\HH_{\vy+\vdelta}-\HH_{\vy}\right)\HH_{\vy}^{-1/2}\right\Vert _{F}\leq\frac{\left\Vert \vdelta\right\Vert _{\HH_{\vy}}}{\left(1-\left\Vert \vdelta\right\Vert _{\HH_{\vy}}\right)^{2}}\,.
\]
Furthermore, if $g$ is only self-concordant, then 
\[
\HH_{\vy}\cdot\left(1-\left\Vert \vdelta\right\Vert _{\HH_{\vy}}\right)^{2}\preceq\HH_{\vy+\vdelta}\preceq\HH_{\vy}\cdot\frac{1}{\left(1-\left\Vert \vdelta\right\Vert _{\HH_{\vy}}\right)^{2}}\,.
\]
\end{lem}

In addition to these properties, we also use the notion of $\nu$-strongly
self-concordant barriers, which is the analog of self-concordant barriers
which also satisfy strong self-concordance.
\begin{defn}
[$\nu$-strongly self-concordant barrier] \label{def:nu-strongly-concordant}Given
a convex domain $K$ and a strongly self-concordant function $g:K\rightarrow\mathbb{R}$,
we say that $g$ is a $\nu$-strongly self-concordant barrier if $g\left(\vy\right)\rightarrow\infty$
as $\vy\rightarrow\partial K$ and
\[
\left\Vert \nabla g\left(\vy\right)\right\Vert _{\HH_{\vy}^{-1}}^{2}\leq\nu\,,
\]
for all $\vy\in K$.
\end{defn}

The stronger property helps to obtain tighter bounds on the increase
in excentricity when only slightly changing the Hessian matrix. We
rely on the following important lemma, which we prove in Section \ref{sec:path-following}:
\begin{lem}
\label{lem:new-excent}Let $K\subseteq\mathbb{R}^{n}$ be a convex
set. Given any strongly self-concordant function $g:K\rightarrow\mathbb{R}$,
for any $\vy,\vdelta$, such that $\vy,\vy+\vdelta\in K$, $\left\Vert \vdelta\right\Vert _{\HH_{\vy}}<\left(1-\left\Vert \vdelta\right\Vert _{\HH_{\vy}}\right)^{2}<1$,
and any preconditioner $\HHt$,
\[
\exc\left(\HHt^{-1}\HH_{\vy+\vdelta}\right)\leq\exc\left(\HHt^{-1}\HH_{\vy}\right)\cdot\exp\left(\frac{1}{2}\sqrt{n}\cdot\frac{\left\Vert \vdelta\right\Vert _{\HH_{\vy}}}{\left(1-\left\Vert \vdelta\right\Vert _{\HH_{\vy}}\right)^{2}-\left\Vert \vdelta\right\Vert _{\HH_{\vy}}}\right)\,.
\]
\end{lem}

In addition, we require some guarantees involving the well-conditionedness
of the points on the neighborhood of the central path. This is captured
by the following definition.
\begin{defn}
[$\kappa(\mu)$-conditioned objective] A a barrier objective (\ref{eq:barrier-obj})
is $\kappa\left(\mu\right)$-conditioned if letting $\vy_{\mu'}$
be the minimizer of $g_{\mu'}$, one has that 
\[
\max\left\{ \left\Vert \HH_{\vy_{\mu'}}\right\Vert ,\left\Vert \HH_{\vy_{\mu'}}^{-1}\right\Vert \right\} \leq\kappa\left(\mu\right)
\]
for all $\mu'\geq\mu$. 
\end{defn}

This captures how large or small the eigenvalues of the Hessian corresponding
to points on the central path can be, for all centrality parameters
above a given threshold $\mu$. While only implicitly used in standard
literature, this quantity is relevant in most instantiations of interior
point methods, as even in the case where one uses fast matrix multiplication
for solving the linear systems involved, the condition number of the
barrier objective $\kappa\left(\mu\right)$ determines the number
of bits of precision required to store the Hessian matrices and their
inverses. A standard feature of this upper bound is that it also holds
for points in the neighborhood of the central path, which we show
in Appendix \ref{subsec:Proof-of-Lemma-well-conditioned-neighborhood}. 
\begin{lem}
\label{lem:well-conditioned-neighborhood}Let $g_{\mu}:K\rightarrow\mathbb{R}$
be a barrier objective with a self-concordant barrier function as
in (\ref{eq:barrier-obj}), and let $\vy\in\text{int}\left(K\right)$,
such that $\left\Vert \nabla g_{\mu}\left(\vy\right)\right\Vert _{\HH_{\vy}^{-1}}\leq1/3$.
Then 
\[
\max\left\{ \left\Vert \HH_{\vy}\right\Vert ,\left\Vert \HH_{\vy}^{-1}\right\Vert \right\} \leq4\kappa\left(\mu\right)\,.
\]
\end{lem}

These facts determine the notion of $\epsilon$-condition number of
a barrier formulation, which we formally define and discuss in Definition
\ref{def:epsilon-condition-number}, which we will use to provide
our convergence guarantees.

\section{Solving Linear Systems with an Adaptive Preconditioner\label{sec:linear}}

We first consider the case of solving linear systems of the form $\HH\vx=\vb$,
and provide an algorithm based on the techniques used by \cite{dunagan2007iteratively}.
We stress the importance of our contribution, since \cite{dunagan2007iteratively}
uses a slightly different set of updates, for which theoretical guarantees
are difficult to prove in the situation where quantities involving
the matrix can only be approximated, as we do in the following Section
\ref{sec:nonlinear}. 

While the analysis in this section assumes access to $\HH$, which
will not be true in the subsequent sections, we note that we only
access it through matrix vector products. This will be important,
as estimating Hessian vector products can be done by making a constant
number of gradient queries.

In this section we show that failure to make a lot of progress within
a single step of the Richardson iteration produces a certificate which
allows us to improve our current preconditioner. This analysis roughly
follows the same ideas as in Dunagan-Harvey, but is slightly simplified
from a technical point of view, as it does not attempt to match the
conjugate gradient algorithm. Instead, it merely performs the Richardson
iteration with a step size chosen such that it minimizes the norm
of the residual. We provide the essential lemmas, then we show how
they can be used to recover the main result in \cite{dunagan2007iteratively}.
In the following section we will extend these to the case where the
Hessian is non-constant, for which we will leverage additional techniques
from matrix analysis.

\subsection{Minimizing Residual Norm via the Richardson Iteration}

We give the lemma which provides the certificate of excentricity in
case a single step of (non-preconditioned) Richardson update fails
to reduce gradient norm significantly. In Appendix \ref{subsec:Excentricity-Certificates-from}
we proceed by providing an analysis for the non-preconditioned case
(see Lemma \ref{lem:Richardson-progress-certificate}). As a corollary,
we obtain a general version of the Lemma corresponding to making a
preconditioned step. 
\begin{lem}
\label{lem:prec-Richardson-progress-certificate}Let $\HH,\HHt\in\mathbb{R}^{n\times n}$,
and vectors $\vb,\vx\in\mathbb{R}^{n}$. Let $\vr=\vb-\HH\vx$, and
consider the step
\[
\vxp=\vx+\frac{\left\Vert \HHt^{-1}\vr\right\Vert _{\HH}^{2}}{\left\Vert \HH\HHt^{-1}\vr\right\Vert _{\HHt^{-1}}^{2}}\HHt^{-1}\vr\,.
\]
Let $\beta\in\left(0,1\right)$ be a scalar. Provided that the new
residual $\vrp=\vb-\HH\vxp$ satisfies
\[
\left\Vert \vrp\right\Vert _{\HHt^{-1}}^{2}\geq\left(1-\beta\right)\left\Vert \vr\right\Vert _{\HHt^{-1}}^{2}\,,
\]
we obtain at least one of the following excentricity certificates:
\begin{enumerate}
\item $\frac{\left\Vert \HH^{1/2}\HHt^{-1}\vr\right\Vert _{\HH^{-1/2}\HHt\HH^{-1/2}}^{2}}{\left\Vert \HH^{1/2}\HHt^{-1}\vr\right\Vert _{2}^{2}}\geq\frac{1}{\sqrt{\beta}}$,
\item $\frac{\left\Vert \HH^{1/2}\HHt^{-1}\vr\right\Vert _{\HH^{1/2}\HHt^{-1}\HH^{1/2}}^{2}}{\left\Vert \HH^{1/2}\HHt^{-1}\vr\right\Vert _{2}^{2}}\geq\frac{1}{\sqrt{\beta}}$.
\end{enumerate}
\end{lem}

The proof can be found in Appendix \ref{subsec:Proof-of-Lemma-proof-of-prec-richardson-progress}.
Finally, based on the certificates provided in Lemma \ref{lem:prec-Richardson-progress-certificate}
we can design a routine which updates our fake Hessian $\HHt^{-1}$.
Its effect on excentricity is given in the following lemma, which
we prove in Appendix \ref{subsec:Proof-of-Lemma-precon-update-1}.
\begin{lem}
\label{lem:precon-update-1}Let $\HH,\HHt\in\mathbb{R}^{n\times n}$,
and suppose that the inverse $\HHt^{-1}$ is available. Given an excentricity
certificate of type 1 or 2 as provided by Lemma \ref{lem:prec-Richardson-progress-certificate},
there is an algorithm (Algorithm \ref{alg:precon-update-linear})
which performs a rank-1 update on $\HHt$ and on its inverse to obtain
a new preconditioner $\HHtp$ such that 
\[
\mathcal{E}\left(\HHtp^{-1}\HH\right)\leq\mathcal{E}\left(\HHt^{-1}\HH\right)\cdot\frac{2}{\sqrt{1+\frac{1}{\sqrt{\beta}}}}\ .
\]
This update can be implemented in $O\left(n^{2}\right)$ time. Furthermore,
either 
\begin{enumerate}
\item $1\leq\left\Vert \HHt^{-1/2}\HHtp\HHt^{-1/2}\right\Vert \leq\left(2\cdot\frac{\exc\left(\HHt\HH^{-1}\right)}{\exc\left(\HHtp\HH^{-1}\right)}\right)^{2}$,
or
\item $1\leq\left\Vert \HHt^{1/2}\HHtp^{-1}\HHt^{1/2}\right\Vert \leq\left(2\cdot\frac{\exc\left(\HHt\HH^{-1}\right)}{\exc\left(\HHtp\HH^{-1}\right)}\right)^{2}$.
\end{enumerate}
\end{lem}

\begin{algorithm}
\begin{algorithmic}[1]

\Ensure Updates the preconditioner $\HHt$ and its inverse in $O\left(n^{2}\right)$
time, plus a constant number of gradient queries.

\Procedure{UpdatePreconditioner}{$\HHt,\HHt^{-1},\vr,\,\text{type}$}

\If{type $ = 1$}

\State$\HHtp=\HHt-\frac{\vr\vr^{\top}}{\left\Vert \HHt^{-1}\vr\right\Vert _{\HH}^{2}+\left\Vert \vr\right\Vert _{\HHt^{-1}}^{2}}$,
$\HHtp^{-1}=\HHt^{-1}+\frac{\HHt^{-1}\vr\vr^{\top}\HHt^{-1}}{\left\Vert \HHt^{-1}\vr\right\Vert _{\HH}^{2}}$

\Else \Comment{type $ = 2$}

\State$\HHtp=\HHt+\frac{\HH\HHt^{-1}\vr\vr^{\top}\HHt^{-1}\HH}{\left\Vert \HHt^{-1}\vr\right\Vert _{\HH}^{2}}$,
$\HHtp^{-1}=\HHt^{-1}-\frac{\HHt^{-1}\HH\HHt^{-1}\vr\vr^{\top}\HHt^{-1}\HH\HHt^{-1}}{\left\Vert \HHt^{-1}\vr\right\Vert _{\HH}^{2}+\left\Vert \HH\HHt^{-1}\vr\right\Vert _{\HHt^{-1}}^{2}}$

\EndIf

\State\Return$\left(\HHtp,\HHtp^{-1}\right)$

\EndProcedure

\end{algorithmic}

\medskip{}

\caption{Pseudocode for preconditioner rank-1 updates, given excentricity certificates.\label{alg:precon-update-linear}}
\end{algorithm}

\subsection{Taking Stock}

Combining Lemma \ref{lem:prec-Richardson-progress-certificate} and
Lemma \ref{lem:precon-update-1} (Algorithm \ref{alg:precon-update-linear})
we obtain a procedure (Algorithm \ref{alg:step-or-update-linear})
which either performs a step that reduces the residual $\left\Vert \vb-\HH\vx\right\Vert _{\HHt^{-1}}^{2}$
by a constant multiplicative factor, or updates $\HHt$ and $\HHt^{-1}$
in $O\left(n^{2}\right)$ time, such that excentricity reduces by
a constant factor. This will be the main driver of the path following
method described in the next section. 

\begin{algorithm}
\begin{algorithmic}[1]

\Ensure Returns a new iterate $\vxp$ such that $\left\Vert \vb-\HH\vxp\right\Vert _{\HHt^{-1}}\leq\left(1-\beta\right)\left\Vert \vb-\HH\vx\right\Vert _{\HHt^{-1}}$,
or returns a new preconditioner $\HHtp$ together with its inverse
such that $\exc\left(\HHtp\HH^{-1}\right)\leq\exc\left(\HHt\HH^{-1}\right)\cdot\frac{2}{\sqrt{1+\frac{1}{\sqrt{\beta}}}}$,
and either $1\leq\left\Vert \HHt^{-1/2}\HHtp\HHt^{-1/2}\right\Vert \leq\left(2\cdot\frac{\exc\left(\HHt\HH^{-1}\right)}{\exc\left(\HHtp\HH^{-1}\right)}\right)^{2}$,
or $1\leq\left\Vert \HHt^{1/2}\HHtp^{-1}\HHt^{1/2}\right\Vert \leq\left(2\cdot\frac{\exc\left(\HHt\HH^{-1}\right)}{\exc\left(\HHtp\HH^{-1}\right)}\right)^{2}$.

\Procedure{StepOrUpdate}{$\HH,\HHt,\HHt^{-1},\vb,\vx$}

\State $\vr=\vb-\HH\vx$

\State $\vxp=\vx+\frac{\left\Vert \HHt^{-1}\vr\right\Vert _{\HH}^{2}}{\left\Vert \HH\HHt^{-1}\vr\right\Vert _{\HHt^{-1}}^{2}}\HHt^{-1}\vr$,
$\vrp=\vb-\HH\vxp$

\If{$\left\Vert \vrp\right\Vert _{\HHt^{-1}}^{2}\leq\left(1-\beta\right)\left\Vert \vr\right\Vert _{\HHt^{-1}}^{2}$}

\State\Return$\left(\HHt,\HHt^{-1},\vxp\right)$\Comment{return new iterate}

\Else

\If{$\frac{\left\Vert \vr\right\Vert _{\HHt^{-1}}^{2}}{\left\Vert \HHt^{-1}\vr\right\Vert _{\HH}^{2}}\geq\frac{1}{\sqrt{\beta}}$}\Comment{test for type 1 excentricity certificate}

\State$\left(\HHtp,\HHtp^{-1}\right)=$\Call{UpdatePreconditioner}{$\HHt,\HHt^{-1},\vr,\text{type = }1$}

\Else\Comment{if test fails, then we must have a type 2 excentricity certificate}

\State$\left(\HHtp,\HHtp^{-1}\right)=$\Call{UpdatePreconditioner}{$\HHt,\HHt^{-1},\vr,\text{type = }2$}

\EndIf

\State\Return$\left(\HHtp,\HHtp^{-1},\vx\right)$\Comment{return new preconditioner}

\EndIf

\EndProcedure

\end{algorithmic}

\medskip{}

\caption{Pseudocode for the step problem.\label{alg:step-or-update-linear}}
\end{algorithm}

While this procedure suffices for our interior point method, we first
provide as a warm-up an analysis for the regime where we indent to
solve a single linear system, thus recovering the main result from
\cite{dunagan2007iteratively}. Its full proof can be found in Appendix
\ref{subsec:Proof-of-Lemma-main-linear-system}.
\begin{lem}
\label{lem:main-linear-system}Consider the linear system $\HH\vx=\vb$,
where $\HH$ is a symmetric positive definite matrix, let $\vx_{0}$
be an initial solution, and let $\HHt_{0}=\Id$ be an initial preconditioner.
Running the iteration
\[
\left(\HHt_{t+1},\HHt_{t+1}^{-1},\vx_{t+1}\right)=\textsc{StepOrUpdate}\left(\HH,\HHt_{t},\HHt_{t}^{-1},\vb,\vx_{t}\right)
\]
for $T=100\left(\ln\mathcal{E}\left(\HH\right)+\ln\frac{1}{\epsilon}\right)$
steps, we obtain a vector $\vx_{T}$ such that
\[
\left\Vert \vb-\HH\vx_{T}\right\Vert \leq\epsilon\cdot\left\Vert \vb-\HH\vx_{0}\right\Vert \,.
\]
\end{lem}

 \global\long\def\vzero{\boldsymbol{\mathit{0}}}\global\long\def\vx{\boldsymbol{\mathit{x}}}\global\long\def\vb{\boldsymbol{\mathit{b}}}\global\long\def\vv{\boldsymbol{\mathit{v}}}\global\long\def\vu{\boldsymbol{\mathit{u}}}\global\long\def\vr{\boldsymbol{\mathit{r}}}\global\long\def\vDelta{\boldsymbol{\mathit{\Delta}}}\global\long\def\vz{\boldsymbol{\mathit{z}}}\global\long\def\vs{\boldsymbol{\mathit{s}}}\global\long\def\vrho{\boldsymbol{\mathit{\rho}}}\global\long\def\vdelta{\boldsymbol{\mathit{\delta}}}\global\long\def\vDeltat{\boldsymbol{\widetilde{\mathit{\Delta}}}}\global\long\def\vc{\boldsymbol{\mathit{c}}}\global\long\def\vh{\boldsymbol{\mathit{h}}}\global\long\def\vup{\boldsymbol{\mathit{u'}}}\global\long\def\vrp{\boldsymbol{\mathit{r'}}}

\global\long\def\vxh{\boldsymbol{\mathit{\widehat{x}}}}\global\long\def\vbh{\boldsymbol{\mathit{\widehat{b}}}}\global\long\def\vxhp{\boldsymbol{\mathit{\widehat{x}'}}}\global\long\def\vxp{\boldsymbol{\mathit{x'}}}\global\long\def\vxs{\boldsymbol{\mathit{x^{\star}}}}\global\long\def\vsp{\boldsymbol{\mathit{s'}}}\global\long\def\vy{\boldsymbol{\mathit{y}}}\global\long\def\vyp{\boldsymbol{\mathit{y'}}}\global\long\def\vg{\boldsymbol{\mathit{g}}}\global\long\def\vrt{\boldsymbol{\mathit{\widetilde{r}}}}\global\long\def\vrtp{\boldsymbol{\mathit{\widetilde{r}'}}}\global\long\def\vys{\boldsymbol{\mathit{y^{\star}}}}

\global\long\def\vrh{\boldsymbol{\mathit{\widehat{r}}}}\global\long\def\vrhp{\boldsymbol{\mathit{\widehat{r}'}}}

\global\long\def\exc{\mathcal{E}}

\global\long\def\HH{\boldsymbol{\mathit{H}}}\global\long\def\HHtil{\boldsymbol{\mathit{\widetilde{H}}}}\global\long\def\XX{\boldsymbol{\mathit{X}}}\global\long\def\XXp{\boldsymbol{\mathit{X'}}}\global\long\def\Id{\boldsymbol{\mathit{I}}}\global\long\def\PP{\boldsymbol{\mathit{P}}}\global\long\def\YY{\boldsymbol{\mathit{Y}}}

\global\long\def\HHh{\boldsymbol{\mathit{\widehat{H}}}}\global\long\def\HHb{\boldsymbol{\mathit{\overline{H}}}}\global\long\def\HHt{\boldsymbol{\mathit{\widetilde{H}}}}\global\long\def\HHtp{\boldsymbol{\mathit{\widetilde{H}'}}}\global\long\def\HHp{\boldsymbol{\mathit{H'}}}

\global\long\def\AA{\boldsymbol{A}}\global\long\def\DD{\boldsymbol{D}}\global\long\def\MM{\boldsymbol{M}}\global\long\def\RR{\boldsymbol{R}}\global\long\def\SS{\mathit{\boldsymbol{S}}}\global\long\def\SSp{{\it \mathit{\boldsymbol{S'}}}}\global\long\def\BB{\boldsymbol{B}}\global\long\def\CC{\boldsymbol{C}}\global\long\def\XXs{\boldsymbol{X^{\star}}}

\global\long\def\PProj{\mathit{\boldsymbol{\Pi}}}

\global\long\def\ks{\kappa_{\star}}

\global\long\def\diag#1{\mathbb{D}\left(#1\right)}\global\long\def\epsilon{\varepsilon}\global\long\def\ln{\log}

\section{Estimating Hessian-Vector Products\label{sec:nonlinear}}

While the analysis in Section \ref{sec:linear} provides a tight bounds
on the evolution of the preconditioner's quality, the algorithms described
there rely almost entirely on having access to the matrix $\HH$.
In our setting our matrix is a Hessian matrix for a convex function,
which we can not directly access. Instead, we note that our entire
interaction with $\HH$ occurs in the form of matrix-vector products.
Hence we should expect that rather than having to compute products
of the form $\HH_{\vy}\vv$ we could instead approximate them
\[
\HH_{\vy}\vv\approx\frac{1}{\tau}\left(\nabla g\left(\vy+\tau\vv\right)-\nabla g\left(\vy\right)\right)\,,
\]
for some appropriate step size $\tau$. To do so it is important to
consider the conditioning of $\HH_{\vy}$, as obtaining an accurate
estimate requires setting $\tau$ sufficiently large. However, we
do not want $\tau$ to be extremely large, as its magnitude determines
among others the number of bits of precision required to evaluate
this approximation. 

In this section we provide robust versions of the algorithms from
Section \ref{sec:linear}, and show that they satisfy similar guarantees,
even though we only access the function $g$ through gradient queries.
In Appendix \ref{subsec:hvp-estimation} we provide formal statements
concerning the quality of the approximations $p_{\vy}\left(\vv\right)\approx\HH_{\vy}\left(\vv\right)$
and $n_{\vy}\left(v\right)\approx\left\Vert \vv\right\Vert _{\HH_{\vy}}$
that we employ. 

At this point we are ready to present the appropriate modifications
to the algorithms from Section \ref{sec:linear}.

\begin{algorithm}
\begin{algorithmic}[1]

\Ensure Updates the preconditioner $\HHt$ and its inverse in $O\left(n^{2}\right)$
time, plus a constant number of gradient queries.

\Procedure{RobustUpdatePreconditioner}{$\vy,\HHt,\HHt^{-1},\vr,\,\text{type}$}

\If{type $ = 1$}

\State$\HHtp=\HHt-\frac{\vr\vr^{\top}}{n_{\vy}\left(\HHt^{-1}\vr\right)^{2}+\left\Vert \vr\right\Vert _{\HHt^{-1}}^{2}}$,
$\HHtp^{-1}=\HHt^{-1}+\frac{\HHt^{-1}\vr\vr^{\top}\HHt^{-1}}{n_{\vy}\left(\HHt^{-1}\vr\right)^{2}}$

\Else \Comment{type $ = 2$}

\State$\HHtp=\HHt+\frac{p_{\vy}\left(\HHt^{-1}\vr\right)p_{\vy}\left(\HHt^{-1}\vr\right)^{\top}}{\left(1+\frac{1}{400\cdot B^{20}}\right)^{2}n_{\vy}\left(\HHt^{-1}\nabla g\left(\vy\right)\right)^{2}}$

\State $\HHtp^{-1}=\HHt^{-1}-\frac{\HHt^{-1}p_{\vy}\left(\HHt^{-1}\vr\right)p_{\vy}\left(\HHt^{-1}\vr\right)^{\top}\HHt^{-1}}{\left(1+\frac{1}{400\cdot B^{20}}\right)^{2}n_{\vy}\left(\HHt^{-1}\vr\right)^{2}+\left\Vert p_{\vy}\left(\HHt^{-1}\vr\right)^{2}\right\Vert _{\HHt^{-1}}^{2}}$

\EndIf

\State\Return$\left(\HHtp,\HHtp^{-1}\right)$

\EndProcedure

\medskip{}

\Procedure{$p_{\vy}$}{$\vv$}

\State $\tau=\frac{1}{1000\left\Vert \vv\right\Vert B^{21}}$

\State\Return $\frac{1}{\tau}\left(\nabla g\left(\vy+\tau\vv\right)-\nabla g\left(\vy\right)\right)$

\EndProcedure

\medskip{}

\Procedure{$n_{\vy}$}{$\vv$}

\State $\tau=\frac{1}{1000\left\Vert \vv\right\Vert B}$

\State\Return $\left(\frac{1}{1-\frac{1}{1000}}\right)\cdot\sqrt{\frac{1}{\tau}\cdot\left\langle \vv,\nabla g\left(\vy+\tau\vv\right)-\nabla g\left(\vy\right)\right\rangle }$

\EndProcedure

\end{algorithmic}

\medskip{}

\caption{Pseudocode for preconditioner rank-1 updates, given excentricity certificates.\label{alg:precon-update-linear-nohess}}
\end{algorithm}

\begin{algorithm}
\begin{algorithmic}[1]

\Ensure Returns a new iterate $\vxp$ such that $\left\Vert \vb-\HH_{\vy}\vxp\right\Vert _{\HHt^{-1}}\leq\left(1-\beta\right)\left\Vert \vb-\HH_{\vy}\vx\right\Vert _{\HHt^{-1}}$,
or returns a new preconditioner $\HHtp$ together with its inverse
such that $\exc\left(\HHtp\HH_{\vy}^{-1}\right)\leq\exc\left(\HHt\HH_{\vy}^{-1}\right)\cdot\frac{2}{\sqrt{1+\frac{1}{\sqrt{\beta}}}}$,
and either $1\leq\left\Vert \HHt^{-1/2}\HHtp\HHt^{-1/2}\right\Vert \leq\left(2\cdot\frac{\exc\left(\HHt\HH^{-1}\right)}{\exc\left(\HHtp\HH^{-1}\right)}\right)^{2}$,
or $1\leq\left\Vert \HHt^{1/2}\HHtp^{-1}\HHt^{1/2}\right\Vert \leq\left(2\cdot\frac{\exc\left(\HHt\HH^{-1}\right)}{\exc\left(\HHtp\HH^{-1}\right)}\right)^{2}$.

\Procedure{RobustStepOrUpdate}{$\vy,\HHt,\HHt^{-1},\vb,\vx$}

\State $\vr=\vb-p_{\vy}\left(\vx\right)$

\If{$\left\Vert \vr\right\Vert _{\HHt^{-1}}\leq1/B$}

\State\Return$\left(\HHt,\HHt^{-1},\vx\right)$\Comment{return the original iterate, since the residual is small}

\EndIf

\State $\vxp=\vx+\frac{n_{\vy}\left(\vr\right)^{2}}{\left\Vert p_{\vy}\left(\HHt^{-1}\vr\right)\right\Vert _{\HHt^{-1}}^{2}}\HHt^{-1}\vr$,
$\vrp=\vb-p_{\vy}\left(\vxp\right)$

\If{$\left\Vert \vr\right\Vert _{\HHt^{-1}}\leq1/B$}

\State\Return$\left(\HHt,\HHt^{-1},\vxp\right)$\Comment{return the new iterate, since the residual is small}

\EndIf

\If{$\left\Vert \vrp\right\Vert _{\HHt^{-1}}^{2}\leq\left(1-\beta\right)\left\Vert \vr\right\Vert _{\HHt^{-1}}^{2}$}

\State\Return$\left(\HHt,\HHt^{-1},\vxp\right)$\Comment{return new iterate}

\Else

\If{$\frac{\left\Vert \vr\right\Vert _{\HHt^{-1}}^{2}}{n_{\vy}\left(\HHt^{-1}\vr\right)^{2}}\geq\frac{1}{\sqrt{\frac{20}{9}\beta}}$}\Comment{test for type 1 excentricity certificate}

\State$\left(\HHtp,\HHtp^{-1}\right)=$\Call{RobustUpdatePreconditioner}{$\HHt,\HHt^{-1},\vr,\text{type = }1$}

\Else\Comment{if test fails, then we must have a type 2 excentricity certificate}

\State$\left(\HHtp,\HHtp^{-1}\right)=$\Call{RobustUpdatePreconditioner}{$\HHt,\HHt^{-1},\vr,\text{type = }2$}

\EndIf

\State\Return$\left(\HHtp,\HHtp^{-1},\vx\right)$\Comment{return new preconditioner}

\EndIf

\EndProcedure

\end{algorithmic}

\medskip{}

\caption{Pseudocode for the step problem.\label{alg:step-or-update-linear-nohess}}
\end{algorithm}

\begin{lem}
\label{lem:nonlinear-progress-certificate}Let $g:K\rightarrow\mathbb{R}$
be a self-concordant function, let $\vy\in\text{int}\left(K\right)$,
let $\HH_{\vy}=\nabla^{2}g\left(\vy\right)$, let $\HHt\in\mathbb{R}^{n\times n}$,
and let vectors $\vb,\vx\in\mathbb{R}^{n}$. Furthermore, suppose
that $\max\left\{ \left\Vert \HH_{\vy}\right\Vert ,\left\Vert \HH_{\vy}^{-1}\right\Vert ,\left\Vert \HHt\right\Vert ,\left\Vert \HHt^{-1}\right\Vert ,\left\Vert \vb\right\Vert _{2}\right\} \leq B$,
for some scalar $B\geq1000$. Let $\vr=\vb-p_{\vy}\left(\vx\right)$,
and consider the step
\[
\vxp=\vx+\frac{n_{\vy}\left(\HHt^{-1}\vr\right)^{2}}{\left\Vert p_{\vy}\left(\HHt^{-1}\vr\right)\right\Vert _{\HHt^{-1}}^{2}}\HHt^{-1}\left(\vb-p_{\vy}\left(\vx\right)\right)\,.
\]
Let $\beta\in\left(\frac{1}{B},1\right)$ be a scalar. Let the new
residual $\vrp=\vb-p_{\vy}\left(\vxp\right)$. Provided that $\min\left\{ \left\Vert \vr\right\Vert _{\HHt^{-1}},\left\Vert \vrp\right\Vert _{\HHt^{-1}}\right\} \geq\frac{1}{B}$,
and
\[
\left\Vert \vrp\right\Vert _{\HHt^{-1}}^{2}\geq\left(1-\beta\right)\left\Vert \vr\right\Vert _{\HHt^{-1}}^{2}\,,
\]
we obtain at least one of the following excentricity certificates:
\begin{enumerate}
\item $\frac{\left\Vert \HH_{\vy}^{1/2}\HHt^{-1}\vr\right\Vert _{\HH_{\vy}^{1/2}\HHt^{-1}\HH_{\vy}^{1/2}}^{2}}{\left\Vert \HH_{\vy}^{1/2}\HHt^{-1}\vr\right\Vert _{2}^{2}}\geq\frac{1}{\sqrt{\frac{20}{9}\beta}}$,
\item $\frac{\left\Vert \HH_{\vy}^{1/2}\HHt^{-1}\vr\right\Vert _{\HH_{\vy}^{1/2}\HHt^{-1}\HH_{\vy}^{1/2}}^{2}}{\left\Vert \HH_{\vy}^{1/2}\HHt^{-1}\vr\right\Vert _{2}^{2}}\geq\frac{1}{\sqrt{\frac{20}{9}\beta}}$.
\end{enumerate}
\end{lem}

The proof is laborious but not particularly insightful, so we defer
it to Appendix \ref{subsec:Proof-of-Lemma-nonlinear-progress-certificate}.
The approach is based on providing multiplicative error bounds between
all the quantities involving $\HH_{\vy}$ and their approximations
using $p_{\vy}$ or $n_{\vy}$. We obtain these by showing that approximations
introduce small additive errors, and use the upper bounds on matrix
norms to show that these are much smaller than the other quantities
involved. To do so we also need to assume that the residuals $\left\Vert \vr\right\Vert _{\HHt^{-1}}$
, $\left\Vert \vrp\right\Vert _{\HHt^{-1}}$ are sufficiently large.
This, however, is automatically true by the hypothesis. Using this
assumption will not hurt us, since we always stop the iterative method
once the norm becomes small enough.

Finally, we need to show that the rank-1 updates performed in Algorithm
\ref{alg:precon-update-linear-nohess}, which involve the approximations
$p_{\vy}$ and $n_{\vy}$ still provide a sufficient decrease in excentricity.
To do so we require a robust version of Lemma \ref{lem:precon-update-1}.
\begin{lem}
\label{lem:precon-update-robust}Let $g:K\rightarrow\mathbb{R}$ be
a self-concordant function, let $\vy\in\text{int}\left(K\right)$,
let $\HH_{\vy}=\nabla^{2}g\left(\vy\right)$, let $\HHt\in\mathbb{R}^{n\times n}$,
and suppose that the inverse $\HHt^{-1}$ is available. Furthermore,
suppose that $\max\left\{ \left\Vert \HH_{\vy}\right\Vert ,\left\Vert \HH_{\vy}^{-1}\right\Vert ,\left\Vert \HHt\right\Vert ,\left\Vert \HHt^{-1}\right\Vert \right\} \leq B$,
for some scalars $B\geq1000$. Given an excentricity certificate of
type 1 or 2 as provided by Lemma \ref{lem:nonlinear-progress-certificate},
there is an algorithm (Algorithm \ref{alg:precon-update-linear-nohess})
which performs a rank-$1$ update on $\HHt$ and on its inverse to
obtain a new preconditioner $\HHtp$ such that 
\[
\mathcal{E}\left(\HHtp^{-1}\HH\right)\leq\mathcal{E}\left(\HHt^{-1}\HH\right)\cdot\frac{2}{\sqrt{1+\frac{99}{100}\cdot\frac{1}{\sqrt{\frac{20}{9}\beta}}}}\ .
\]
This update can be implemented in $O\left(n^{2}\right)$ time, plus
the time required to make a constant number of gradient queries for
$g$. In addition, $\max\left\{ \left\Vert \HHtp\right\Vert -\left\Vert \HHt\right\Vert ,\left\Vert \HHtp^{-1}\right\Vert -\left\Vert \HHt^{-1}\right\Vert \right\} \leq2\max\left\{ \left\Vert \HH_{\vy}\right\Vert ,\left\Vert \HH_{\vy}^{-1}\right\Vert \right\} ^{2}$.
\end{lem}

The proof is based on the one we gave for Lemma \ref{lem:precon-update-1},
and is presented in Appendix \ref{subsec:Proof-of-Lemma-precon-update-hess}.
Finally, combining Lemmas \ref{lem:nonlinear-progress-certificate}
and \ref{lem:precon-update-robust} we obtain the main result of this
section.
\begin{lem}
[Robust Step-or-Update]\label{lem:robust-step-or-update} Let $g:K\rightarrow\mathbb{R}$
be a self-concordant function, let $\vy\in\text{int}\left(K\right)$,
let $\HH_{\vy}=\nabla^{2}g\left(\vy\right)$, let $\HHt\in\mathbb{R}^{n\times n}$,
and let vectors $\vb,\vx\in\mathbb{R}^{n}$. Furthermore, suppose
that $\max\left\{ \left\Vert \HH_{\vy}\right\Vert ,\left\Vert \HH_{\vy}^{-1}\right\Vert ,\left\Vert \HHt\right\Vert ,\left\Vert \HHt^{-1}\right\Vert ,\left\Vert \vb\right\Vert _{2}\right\} \leq B$,
for some scalar $B\geq1000$. Let $\beta\in\left(\frac{1}{B},1\right)$
be a scalar. There is an algorithm (Algorithm \ref{alg:step-or-update-linear-nohess})
which has exactly one of the following two outcomes:
\begin{enumerate}
\item Returns new iterate $\vxp$ such that $\left\Vert \vb-\HH_{\vy}\vxp\right\Vert _{\HHt^{-1}}^{2}\leq\max\left\{ \frac{2}{B},\left(1-\frac{\beta}{2}\right)\left\Vert \vb-\HH_{\vy}\vx\right\Vert _{\HHt^{-1}}^{2}\right\} \,,$
\item Updates $\HHt$ and to obtain a new preconditioner $\HHtp$ such that
\[
\mathcal{E}\left(\HHtp^{-1}\HH\right)\leq\mathcal{E}\left(\HHt^{-1}\HH\right)\cdot\frac{2}{\sqrt{1+\frac{99}{100}\cdot\frac{1}{\sqrt{\frac{20}{9}\beta}}}}\ .
\]
\end{enumerate}
This step can be implemented in $O\left(n^{2}\right)$ time, plus
the time required to make a constant number of gradient queries for
$g$. In addition, $\max\left\{ \left\Vert \HHtp\right\Vert -\left\Vert \HHt\right\Vert ,\left\Vert \HHtp^{-1}\right\Vert -\left\Vert \HHt^{-1}\right\Vert \right\} \leq2\max\left\{ \left\Vert \HH_{\vy}\right\Vert ,\left\Vert \HH_{\vy}^{-1}\right\Vert \right\} ^{2}$.
\end{lem}

\begin{proof}
In addition to invoking Lemmas \ref{lem:nonlinear-progress-certificate}
and \ref{lem:precon-update-robust} we perform two checks. First we
verify if the $\HHt^{-1}$ norm original residual $\vr=\vb-p_{\vy}\left(\vx\right)$
is too small, in which case we return the iterate unchanged. Then,
after computing the new iterate and the corresponding residual, we
again verify whether $\vrp$ is too small, in which case we again
return it. Following these checks we verify whether the norm $\left\Vert \vrp\right\Vert _{\HHt^{-1}}\leq\left(1-\beta\right)\left\Vert \vr\right\Vert _{\HHt^{-1}}$.
If this check fails, then we call Algorithm \ref{alg:precon-update-linear-nohess}
to update the preconditioner. The total time required to run this
routine is again dominated by the gradient queries used in the estimation
procedures, and the time to perform a rank-1 update which is $O\left(n^{2}\right)$.

Finally, we need to convert the guarantees involving the residual
$\vrp=\vb-p_{\vy}\left(\vxp\right)$ to guarantees involving the linear
function $\vb-\HH_{\vy}\vxp$.

Note that in the former, using the same approximations we proved in
Lemma \ref{lem:nonlinear-progress-certificate} (equation (\ref{eq:residual-approx-imp})),
we have that the bound $\left\Vert \vb-p_{\vy}\left(\vxp\right)\right\Vert _{\HHt^{-1}}^{2}\leq\frac{1}{B}$
implies 
\[
\left\Vert \vb-\HH_{\vy}\vxp\right\Vert _{\HHt^{-1}}\leq\left(\left\Vert \vb-p_{\vy}\left(\vx\right)\right\Vert _{\HHt^{-1}}+\frac{1}{400\cdot B^{16}}\right)\frac{1}{1-\frac{1}{400\cdot B^{18}}}\leq\frac{2}{B}\,.
\]
In the other case using (\ref{eq:e4}) we have
\begin{align*}
 & \left\Vert \vb-\HH_{\vy}\vxp\right\Vert _{\HHt^{-1}}\left(1-\frac{1}{200\cdot B^{15}}\right)\leq\left\Vert \vb-p_{\vy}\left(\vxp\right)\right\Vert _{\HHt^{-1}}^{2}\\
 & \leq\left(1-\beta\right)\left\Vert \vb-p_{\vy}\left(\vx\right)\right\Vert _{\HHt^{-1}}^{2}\leq\left(1-\beta\right)\left\Vert \vb-\HH_{\vy}\vx\right\Vert _{\HHt^{-1}}\left(1+\frac{1}{200\cdot B^{15}}\right)\,,
\end{align*}
and so 
\begin{align*}
\left\Vert \vb-\HH_{\vy}\vxp\right\Vert _{\HHt^{-1}}
\leq
\left(1-\beta\right)\left\Vert \vb-\HH_{\vy}\vx\right\Vert _{\HHt^{-1}}\left(1+\frac{1}{200\cdot B^{15}}\right)\frac{1}{1-\frac{1}{200\cdot B^{15}}} \\
\leq
\left(1-\beta/2\right)\left\Vert \vb-\HH_{\vy}\vx\right\Vert _{\HHt^{-1}}\,,
\end{align*}
which concludes the proof.
\end{proof}

\paragraph{Handling matrices with a null space.\label{par:low-rank}}

While we have previously assumed that the underlying matrix $\HH_{\vy}$
is full-rank, our analysis carries over to the case where it has a
non-trivial null space $\mathcal{S}$, if the right hand side vector
$\vb\in\mathcal{S}^{\perp}$, i.e. it is orthogonal to $\mathcal{S}$.
To show this, we note that the analysis of Algorithm \ref{alg:step-or-update-linear-nohess}
can be restricted to $\mathcal{S}^{\perp}$. Note that by self-concordance
the Hessian of the underlying function $g$ always has $\mathcal{S}$
as null space. If this were not true, by slightly perturbing the argument
we could move to a new point where the new Hessian fails to approximate
the original one because one entire direction in the span gets zeroed
out, so the two Hessians fail to approximate each other, which contradicts
self-concordance. This shows that $p_{\vy}\left(\vv\right)$ always
outputs a vector orthogonal to $\mathcal{S}$, and therefore the estimated
residuals $\vr$ are always orthogonal to $\mathcal{S}$. Similarly,
$n_{\vy}\left(\vv\right)$ zeroes out the components of $\vv$ parallel
to $\mathcal{S}$, so everything it evaluates occurs in $\mathcal{S}^{\perp}$.
While we do not explicitly impose restrictions on the preconditioner
$\HHt$, we can replace it in the analysis with the projection $\PProj_{\mathcal{S}^{\perp}}\HHt\PProj_{\mathcal{S}^{\perp}}$
without affecting the provided guarantees, since this does not affect
the excentricity certificates. 

\paragraph{Numerical precision aspects.}

The exposition of the analysis so far is under the assumption of exact
arithmetic. Our algorithms do in fact tolerate finite fixed-point
precision on the scale of the natural parameters of the problem ($n$,
$\epsilon$, $B$). Assuming we represent numbers using a number of
bits that is polylogarithmic in these parameters, we can ensure that
all the round-off errors are bounded by $1/\text{poly}\left(n,\epsilon,B\right)$.
These add only a polynomially small additive errors to our evaluations
of the residual norm, which are already captured by the analysis.
Crucially, we need to ensure that these errors do not affect the excentricity
potential. First, polynomially small additive error does not affect
the excentricity certificates -- they still certify the existence
of a poorly conditioned direction, potentially suffering a polynomially
small multiplicative decrease in the quality of the bound provided
by Lemma \ref{lem:nonlinear-progress-certificate}. Second, and more
importantly, we need to argue that the rank-$1$ updates to the fake
preconditioner and its inverse still decrease excentricity as needed,
even though round-off errors are introduced. This is because we can
view the round-off errors as an adversarial perturbation $\vDelta$
done on the preconditioned matrix $\XX=\HH_{\vy}^{-1/2}\HHtil\HH_{\vy}^{-1/2}$
such that all entries of $\vDelta$ are polynomially small. We show
that such perturbations do not significantly increase excentricity.
Indeed, since the determinant is multiplicatively stable under small
perturbations,
\begin{align*}
\det\left(\AA+\vDelta\right) 
 =
 \det\left(\AA\right)\det\left(\Id+\AA^{-1/2}\vDelta\AA^{-1/2}\right)\leq\det\left(\AA\right)\cdot\left(1+\frac{\left\Vert \vDelta\right\Vert }{\left\Vert \AA^{-1}\right\Vert }\right)^{n} 
 \\
 \leq
 \det\left(\AA\right)\cdot\left(1+\frac{n\max_{i,j}\left|\vDelta_{ij}\right|}{\left\Vert \AA^{-1}\right\Vert }\right)^{n}\,,
\end{align*}
we can make increases in excentricity caused by the noise in $\vDelta$
only affect by a $1+\frac{1}{\text{poly}\left(n,\epsilon,B\right)}$
multiplicative factor. Thus polynomially small error is sufficient,
requiring only logarithmically many bits.
 \global\long\def\vzero{\boldsymbol{\mathit{0}}}\global\long\def\vx{\boldsymbol{\mathit{x}}}\global\long\def\vb{\boldsymbol{\mathit{b}}}\global\long\def\vv{\boldsymbol{\mathit{v}}}\global\long\def\vu{\boldsymbol{\mathit{u}}}\global\long\def\vr{\boldsymbol{\mathit{r}}}\global\long\def\vDelta{\boldsymbol{\mathit{\Delta}}}\global\long\def\vz{\boldsymbol{\mathit{z}}}\global\long\def\vs{\boldsymbol{\mathit{s}}}\global\long\def\vrho{\boldsymbol{\mathit{\rho}}}\global\long\def\vdelta{\boldsymbol{\mathit{\delta}}}\global\long\def\vDeltat{\boldsymbol{\widetilde{\mathit{\Delta}}}}\global\long\def\vc{\boldsymbol{\mathit{c}}}\global\long\def\vh{\boldsymbol{\mathit{h}}}\global\long\def\vup{\boldsymbol{\mathit{u'}}}\global\long\def\vrp{\boldsymbol{\mathit{r'}}}

\global\long\def\vxh{\boldsymbol{\mathit{\widehat{x}}}}\global\long\def\vbh{\boldsymbol{\mathit{\widehat{b}}}}\global\long\def\vxhp{\boldsymbol{\mathit{\widehat{x}'}}}\global\long\def\vxp{\boldsymbol{\mathit{x'}}}\global\long\def\vxs{\boldsymbol{\mathit{x^{\star}}}}\global\long\def\vsp{\boldsymbol{\mathit{s'}}}\global\long\def\vy{\boldsymbol{\mathit{y}}}\global\long\def\vyp{\boldsymbol{\mathit{y'}}}\global\long\def\vg{\boldsymbol{\mathit{g}}}\global\long\def\vrt{\boldsymbol{\mathit{\widetilde{r}}}}\global\long\def\vrtp{\boldsymbol{\mathit{\widetilde{r}'}}}\global\long\def\vys{\boldsymbol{\mathit{y^{\star}}}}

\global\long\def\vrh{\boldsymbol{\mathit{\widehat{r}}}}\global\long\def\vrhp{\boldsymbol{\mathit{\widehat{r}'}}}

\global\long\def\exc{\mathcal{E}}

\global\long\def\HH{\boldsymbol{\mathit{H}}}\global\long\def\HHtil{\boldsymbol{\mathit{\widetilde{H}}}}\global\long\def\XX{\boldsymbol{\mathit{X}}}\global\long\def\XXp{\boldsymbol{\mathit{X'}}}\global\long\def\Id{\boldsymbol{\mathit{I}}}\global\long\def\PP{\boldsymbol{\mathit{P}}}\global\long\def\YY{\boldsymbol{\mathit{Y}}}

\global\long\def\HHh{\boldsymbol{\mathit{\widehat{H}}}}\global\long\def\HHb{\boldsymbol{\mathit{\overline{H}}}}\global\long\def\HHt{\boldsymbol{\mathit{\widetilde{H}}}}\global\long\def\HHtp{\boldsymbol{\mathit{\widetilde{H}'}}}\global\long\def\HHp{\boldsymbol{\mathit{H'}}}

\global\long\def\AA{\boldsymbol{A}}\global\long\def\DD{\boldsymbol{D}}\global\long\def\MM{\boldsymbol{M}}\global\long\def\RR{\boldsymbol{R}}\global\long\def\SS{\mathit{\boldsymbol{S}}}\global\long\def\SSp{{\it \mathit{\boldsymbol{S'}}}}\global\long\def\BB{\boldsymbol{B}}\global\long\def\CC{\boldsymbol{C}}\global\long\def\XXs{\boldsymbol{X^{\star}}}

\global\long\def\PProj{\mathit{\boldsymbol{\Pi}}}

\global\long\def\ks{\kappa_{\star}}

\global\long\def\diag#1{\mathbb{D}\left(#1\right)}\global\long\def\epsilon{\varepsilon}\global\long\def\ln{\log}

\section{Path Following with a Fake Hessian\label{sec:path-following}}

The robust version of Lemmas \ref{lem:prec-Richardson-progress-certificate}
and \ref{lem:precon-update-1} (Lemmas and \ref{lem:nonlinear-progress-certificate}
and \ref{lem:precon-update-robust}) give us a method for maintaining
a preconditioner for the Hessian of a self-concordant function $g$
at a fixed point $\vy$. To use them within a path following method,
we need to be able to control the quality of the preconditioner $\HHt$
as we move to a new point on the central path. Fortunately, strong
self-concordance turns out to be exactly the property required to
ensure that the excentricity $\exc\left(\HHt^{-1}\HH_{\vy}\right)$
does not drastically increase as we move to a new point $\vyp$.

Here we crucially rely on Lemma \ref{lem:new-excent}, whose proof
we provide at the end of the section. To prove Lemma \ref{lem:new-excent}
we use a determinant inequality (Lemma \ref{lem:eig-ineq}) which
we prove by using a majorization bound for the eigenvalues of the
sum of two matrices, together with some properties related to the
Schur convexity of elementary symmetric polynomials.

\subsection{Setup}

We consider the barrier formulation
\begin{equation}
g_{\mu}\left(\vy\right)=\frac{\left\langle \vc,\vy\right\rangle }{\mu}+\phi\left(\vy\right)\,,\label{eq:barrier-obj}
\end{equation}
where $\phi:K\rightarrow\mathbb{R}$ is a $\nu$-strongly self-concordant
barrier for a convex set $K$. Our goal is to obtain a near-minimizer
for $g_{\mu}$, for a sufficiently small $\mu$. This is quantitatively
given by the following lemma, which is standard in interior point
methods.
\begin{lem}
[Approximate optimality \cite{renegar2001mathematical}]\label{lem:optimality-near}Suppose
that $\epsilon_{N}\leq1/10$, a feasible solution $\vy\in\mathbb{R}^{m}$
and the parameter $\mu\leq1$ satisfy the following bound on the Newton
step size
\[
\left\Vert \nabla g_{\mu}\left(\vy\right)\right\Vert _{\HH_{\vy}^{-1}}\leq\epsilon_{N}\,.
\]
Let $\vys$ be an optimal solution to the objective original objective
$\min\left\{ \left\langle \vc,\vy\right\rangle :\vy\in K\right\} $.
Then $\left\langle \vc,\vy\right\rangle \leq\left\langle \vc,\vys\right\rangle +\mu\nu\left(1+2\epsilon_{N}\right)\,.$
\end{lem}

Having established via Lemma \ref{lem:optimality-near} that to solve
the problem to precision $\epsilon$ we require finding the minimizer
of $g_{\mu}$ for $\mu\apprle\frac{\epsilon}{2\nu}$, we know that
throughout the entire execution of a standard interior point method
the matrices in the neighborhood of the central path will have eigenvalues
(and their inverses) that are bounded by the $\epsilon$-condition
number of the barrier formulation
\begin{defn}
[$\epsilon$-condition number]\label{def:epsilon-condition-number}Given
a barrier formulation (\ref{eq:barrier-obj}), where $\phi$ is a
$\nu$-strongly self-concordant barrier, we define the $\epsilon$-condition
number of the formulation as 
\begin{align}
\ks & :=4\kappa\left(\frac{\epsilon}{2\nu}\right)\,.\label{eq:kstardef}
\end{align}

This parameter determines the precision to which we have to estimate
Hessian-vector products, as well as the number of iterations of the
entire algorithm. The latter is true because, being unable to measure
the true dual local norm of our gradients, we generally do not know
whether we have managed to be close enough to the central path, that
we can advance by dialing down $\mu$. The promise on the Hessian
eigenvalues, together with upper bounds we maintain on the eigenvalues
of the preconditioner and its inverse, enable us to approximate $\left\Vert \nabla g\left(\vy\right)\right\Vert _{\HH_{\vy}^{-1}}$. 
\end{defn}

Throughout the entire execution of the algorithm we will maintain
the following invariant on the preconditioner and the Hessian matrices
in the neighborhood of the central path.

\paragraph{Invariant 1. }

For $B_{\HH}=4\ks$ and $B_{\HHt}=10^{9}\cdot\ks^{2}\cdot\left(n\ln\ks+\sqrt{\nu n}\ln\frac{n\nu}{\epsilon}\right)$,
we always have $\max\left\{ \left\Vert \HH_{\vy}\right\Vert ,\left\Vert \HH_{\vy}^{-1}\right\Vert \right\} \leq B_{\HH}$
and $\max\left\{ \left\Vert \HHt\right\Vert ,\left\Vert \HHt^{-1}\right\Vert \right\} \leq B_{\HHt}$. 

While the upper bound on $B_{\HH}$ automatically holds if we always
stay close to the central path, per Lemma \ref{lem:well-conditioned-neighborhood}
provided that we stop the interior point method at $\mu=\frac{\epsilon}{2\nu}$,
we also need to show that the preconditioner $\HHt$ never becomes
too poorly conditioned. This will show that this holds by tracking
its evolution throughout the execution of the interior point method.
\begin{lem}
\label{lem:one-step}Let $\vy$ such that $\left\Vert \nabla g_{\mu}\left(\vy\right)\right\Vert _{\HH_{\vy}^{-1}}\leq\frac{1}{20\sqrt{n}}$.
Then Algorithm \ref{alg:Path-Following} returns a new iterate $\vyp$
such that $\left\Vert \nabla g_{\mu/\left(1+\frac{1}{20\sqrt{\nu n}}\right)}\left(\vyp\right)\right\Vert _{\HH_{\vyp}^{-1}}\leq\frac{1}{20\sqrt{n}}$. 
\end{lem}

\begin{proof}
The algorithm essentially corresponds to performing a single Newton
step to restore centrality, followed by adjusting the centrality parameter
$\mu$. Over the course of its execution, it may call the \textsc{RobustStepOrUpdate}
routine several times, as it may need to adjust its preconditioner.
When calling it we use the parameters $B=B_{\HHt}^{10}$, and $\beta=1/100$.
First we show that indeed, at the end of the execution, the returned
iterate still has a small dual local norm. 

After a successful call when the preconditioner does not change, per
the guarantees of Lemma \ref{lem:robust-step-or-update}, we have
that
\[
\left\Vert \nabla g_{\mu}\left(\vy\right)-\HH_{\vy}\vxp\right\Vert _{\HHt^{-1}}\leq\max\left\{ \frac{2}{B_{\HHt}^{10}},\left(1-\frac{\beta}{2}\right)\left\Vert \nabla g_{\mu}\left(\vy\right)-\HH_{\vy}\vx\right\Vert _{\HHt^{-1}}^{2}\right\} \,.
\]
and the counter $t$ gets increased. Since the counter stops after
$T=\frac{100}{\beta}\ln\left(B_{\HH}B_{\HHt}\right)$ successful iterations,
at the end of the execution the returned variable $\vx$ satisfies
\begin{align*}
\left\Vert \nabla g_{\mu}\left(\vy\right)-\HH_{\vy}\vx\right\Vert _{\HHt^{-1}} & \leq\max\left\{ \frac{2}{B_{\HHt}^{10}},\left(1-\frac{\beta}{2}\right)^{T}\left\Vert \nabla g_{\mu}\left(\vy\right)\right\Vert _{\HHt^{-1}}^{2}\right\} \\
 & \leq\max\left\{ \frac{2}{B_{\HHt}^{10}},\left(1-\frac{\beta}{2}\right)^{T}\left\Vert \nabla g_{\mu}\left(\vy\right)\right\Vert _{\HH_{\vy}^{-1}}^{2}\cdot\left\Vert \HH_{\vy}\right\Vert \left\Vert \HHt^{-1}\right\Vert \right\} \leq\frac{2}{B_{\HHt}^{10}}\,.
\end{align*}
Therefore 
\[
\left\Vert \nabla g_{\mu}\left(\vy\right)-\HH_{\vy}\vx\right\Vert _{\HH_{\vy}^{-1}}\leq\left\Vert \HH_{\vy}^{-1}\right\Vert \left\Vert \HHt\right\Vert \frac{2}{B_{\HHt}^{10}}\leq\frac{2}{B_{\HHt}^{8}}\,.
\]
From a standard argument shown in Lemma \ref{lem:ipm-prog} we obtain
that
\begin{align*}
\left\Vert \nabla g_{\mu}\left(\vyp\right)\right\Vert _{\HH_{\vyp}^{-1}} & \leq\frac{4}{B_{\HHt}^{8}}+7\cdot\left\Vert \nabla g_{\mu}\left(\vy\right)\right\Vert _{\HH_{\vy}^{-1}}^{2}\,.
\end{align*}
Thus moving from centrality $\mu$ to $\mu/\left(1+\frac{1}{200\sqrt{\nu n}}\right)$
we obtain via Lemma \ref{lem:new-grad} that 
\begin{align*}
\left\Vert \nabla g_{\mu/\left(1+\frac{1}{200\sqrt{\nu n}}\right)}\left(\vyp\right)\right\Vert _{\HH_{\vyp}^{-1}} & =\left\Vert \left(1+\frac{1}{200\sqrt{\nu n}}\right)\nabla g_{\mu}\left(\vyp\right)-\frac{1}{200\sqrt{\nu n}}\nabla\phi\left(\vyp\right)\right\Vert _{\HH_{\vyp}^{-1}}\\
 & \leq\left(1+\frac{1}{200\sqrt{\nu n}}\right)\left\Vert \nabla g_{\mu}\left(\vyp\right)\right\Vert _{\HH_{\vyp}^{-1}}+\frac{1}{200\sqrt{\nu n}}\left\Vert \nabla\phi\left(\vyp\right)\right\Vert _{\HH_{\vyp}^{-1}}\\
 & \leq\left(1+\frac{1}{200\sqrt{\nu n}}\right)\left\Vert \nabla g_{\mu}\left(\vyp\right)\right\Vert _{\HH_{\vyp}^{-1}}+\frac{1}{200\sqrt{n}}\\
 & \leq\left(1+\frac{1}{200\sqrt{\nu n}}\right)\left(\frac{4}{B_{\HHt}^{8}}+7\cdot\left\Vert \nabla g_{\mu}\left(\vy\right)\right\Vert _{\HH_{\vy}^{-1}}^{2}\right)+\frac{1}{200\sqrt{n}}\\
 & \leq\left(1+\frac{1}{200\sqrt{\nu n}}\right)\left(\frac{4}{B_{\HHt}^{8}}+7\cdot\left(\frac{1}{20\sqrt{n}}\right)^{2}\right)+\frac{1}{200\sqrt{n}}\\
 & \leq\frac{1}{20\sqrt{n}}\,.
\end{align*}

Next we prove an upper bound on the total number of calls to the preconditioner
update routine.
\end{proof}
\begin{lem}
Let $\vy_{0}$ such that $\left\Vert \nabla g_{\mu}\left(\vy_{0}\right)\right\Vert _{\HH_{\vy_{0}}^{-1}}\leq\frac{1}{20\sqrt{n}}$,
for $\mu=n^{O\left(1\right)}$, and let $\HHt_{0}=\Id$ be the preconditioner
at initialization. In $\widetilde{O}\left(\sqrt{\nu n}\ln\left(\frac{n\nu}{\epsilon}\right)\right)$
calls to Algorithm \ref{alg:Path-Following} we obtain a new point
$\vyp$ such that $\left\langle \vc,\vyp\right\rangle -\left\langle \vc,\vys\right\rangle \leq\epsilon$.
Over the entire course of the execution, the algorithm, makes at most
$O\left(\ln\exc\left(\HH_{\vy_{0}}\right)+\sqrt{\nu n}\ln\frac{n\nu}{\epsilon}\right)$
preconditioner updates.
\end{lem}

\begin{proof}
First, by Lemma \ref{lem:one-step}, in a single call to Algorithm
\ref{alg:Path-Following}, the centrality parameter $\mu$ gets reduced
by a factor of $1+\frac{1}{20\sqrt{\nu n}}$ while maintaining an
upper bound of $1/\left(20\sqrt{n}\right)$ on the dual local norm.
By Lemma \ref{lem:optimality-near}, stopping once $\mu\leq\frac{\epsilon}{\nu\left(1+2\cdot\frac{1}{20}\right)}$
ensures that error in the objective is at most $\epsilon$, so we
are done. Therefore in $O\left(\sqrt{\nu n}\ln\frac{n\nu}{\epsilon}\right)$
calls, the algorithm can return a solution with $\epsilon$ additive
error. We bound the number of preconditioner updates by analyzing
the evolution of excentricity. In each call to Algorithm \ref{alg:Path-Following},
excentricity can increase exactly once, whenever we move from $\vy$
to $\vyp$ and hence we replace $\exc\left(\HHt^{-1}\HH_{\vy}\right)$
with $\exc\left(\HHt^{-1}\HH_{\vyp}\right)$. 

When doing so, from Lemma \ref{lem:ipm-prog}, we have that the change
in iterate $\left\Vert \vdelta\right\Vert _{\HH_{\vy}}=\left\Vert \vy-\vyp\right\Vert _{\HH_{\vy}}\leq\frac{2}{B_{\HHt}^{10}}+\left\Vert \nabla g_{\mu}\left(\vy\right)\right\Vert _{\HH_{\vy}^{-1}}\leq\frac{1}{10\sqrt{n}}$.
Thus by Lemma \ref{lem:new-excent}, 
\[
\exc\left(\HHt^{-1}\HH_{\vyp}\right)\leq\exc\left(\HHt^{-1}\HH_{\vy}\right)\cdot\exp\left(\frac{1}{2}\sqrt{n}\cdot\frac{\left\Vert \vdelta\right\Vert _{\HH_{\vy}}}{\left(1-\left\Vert \vdelta\right\Vert _{\HH_{\vy}}\right)^{2}-\left\Vert \vdelta\right\Vert _{\HH_{\vy}}}\right)\leq\exc\left(\HHt^{-1}\HH_{\vy}\right)\cdot\exp\left(\frac{1}{10}\right)\,.
\]
In addition all the other changes in excentricity are only due to
preconditioner updates. And we know by Lemma \ref{lem:robust-step-or-update}
that these decrease excentricity in the sense that
\[
\exc\left(\HHtp^{-1}\HH_{\vy}\right)\leq\exc\left(\HHt^{-1}\HH_{\vy}\right)\cdot\frac{2}{\sqrt{1+\frac{99}{100}\cdot\frac{1}{\sqrt{\frac{20}{9}\beta}}}}\leq\exc\left(\HHt^{-1}\HH_{\vy}\right)\cdot\frac{3}{4}\,.
\]
As over the entire course of the algorithm excentricity can increase
by at most $\exp\left(O\left(\sqrt{\nu n}\ln\frac{n\nu}{\epsilon}\right)\right)$,
and each other preconditioner update decreases it by at least a factor
of $3/4$, the total number of preconditioner updates can not be too
large (as excentricity is always at least $1$). Hence we bound the
total number of preconditioner updates to
\[
O\left(\ln\exc\left(\HHt_{0}^{-1}\HH_{\vy_{0}}\right)+\sqrt{\nu n}\ln\frac{n\nu}{\epsilon}\right)=O\left(\ln\exc\left(\HH_{\vy_{0}}\right)+\sqrt{\nu n}\ln\frac{n\nu}{\epsilon}\right)\,,
\]
which concludes the proof. 
\end{proof}
Finally we need to argue that over the entire execution Invariant
1 is preserved. To do this we need to bound $\left\Vert \HHt\right\Vert $,
$\left\Vert \HHt^{-1}\right\Vert $. This is easy to see, as from
Lemma \ref{lem:robust-step-or-update} each step increases both of
these norms by at most $2\max\left\{ \left\Vert \HH_{\vy}\right\Vert ,\left\Vert \HH_{\vy}^{-1}\right\Vert \right\} ^{2}\leq2\ks^{2}$,
so by the end we have 
\begin{align*}
\max\left\{ \left\Vert \HHt\right\Vert ,\left\Vert \HHt^{-1}\right\Vert \right\}  & \leq2\ks^{2}\cdot\left(\ln\exc\left(\HH_{\vy_{0}}\right)+O\left(\sqrt{\nu n}\ln\frac{n\nu}{\epsilon}\right)\right)\\
 & \leq2\ks^{2}\cdot\left(n\ln\ks+O\left(\sqrt{\nu n}\ln\frac{n\nu}{\epsilon}\right)\right)\,,
\end{align*}
which shows that invariants hold throughout the entire execution.

Putting everything together, together with classical techniques involving
initializing interior point methods \cite{nesterov1998introductory}
the main theorem of this section follows.
\begin{thm}
\label{thm:main-opt} Given a convex set $K\subseteq\mathbb{R}^{n}$,
with a $\nu$-strongly self-concordant barrier $\phi:K\rightarrow\mathbb{R}$,
which we can access through a gradient oracle, and given an initial
point $\vy_{0}\in\text{int}\left(K\right)$, $\left\Vert \nabla\phi\left(\vy_{0}\right)\right\Vert _{\HH_{\vy_{0}}^{-1}}\leq\frac{1}{20\sqrt{n}}$,
in time 
\[
O\left(\left(\ln\exc\left(\HH_{\vy_{0}}\right)+\sqrt{\nu n}\ln\frac{n\nu}{\epsilon}\right)\cdot\left(n^{2}+\mathcal{T}_{\text{gradient}}\right)\right)
\]
we can obtain an iterate $\vyp$ such that $\left\langle \vc,\vy\right\rangle \leq\left\langle \vc,\vys\right\rangle +\epsilon$,
where $\vys$ is the minimizer of the original objective $\min\left\{ \left\langle \vc,\vy\right\rangle :\vy\in K\right\} $.
Furthermore, all the matrices we encounter, together their inverses
have eigenvalues bounded by $\left(\ks n\nu\ln\frac{1}{\epsilon}\right)^{O\left(1\right)}$,
where $\ks$ is the $\epsilon$-condition number of the barrier formulation
(Definition \ref{def:epsilon-condition-number}).
\end{thm}

Note that the condition on the dual local norm on the barrier gradient
corresponds to being close to the central path. This is usually handled
through appropriately modifying the optimization problem such that
the guarantee is enforced. 

It is important to note that for as long as $\ks$ is quasi-polynomially
bounded, all the involved quantities are quasi-polynomial, hence we
can implement the interior point method using only poly-logarithmically
many bits of precision for each number.

Additionally this implies the following important corollary for optimizing
linear functions over convex domains. This follows from the $\widetilde{O}\left(n\right)$-strong
self-concordance of the universal and entropic barriers \cite{laddha2020strong},
which in turn results from recent progress on the KLS conjecture \cite{chen2021almost,jambulapati2022slightly}.
While \cite{laddha2020strong} prove that the universal and entropic
barriers satisfy an approximate version of Definition \ref{def:strong-self-conc},
with a polylogarithmic constant in front of the right-hand side, they
can be brought to satisfy it exactly by slightly scaling them down.
In exchange they suffer a small polylogarithmic increase in the self-concordance
parameter, but this does not affect our bounds.
\begin{cor}
\label{cor:main-cor}Given a convex set $K\subseteq\mathbb{R}^{n}$,
with query access to a gradient oracle for the universal or entropic
barrier $\phi:K\rightarrow\mathbb{R}$, and given an initial point
$\vy_{0}\in\text{int}\left(K\right)$, $\left\Vert \nabla\phi\left(\vy_{0}\right)\right\Vert _{\HH_{\vy_{0}}^{-1}}\leq\frac{1}{20\sqrt{n}}$,
in time 
\[
O\left(\ln\left(\frac{n\ks}{\epsilon}\right)\cdot\left(n^{3}+n\mathcal{T}_{\text{gradient}}\right)\right)
\]
we can obtain an iterate $\vyp$ such that $\left\langle \vc,\vy\right\rangle \leq\left\langle \vc,\vys\right\rangle +\epsilon$,
where $\vys$ is the minimizer of the original objective $\min\left\{ \left\langle \vc,\vy\right\rangle :\vy\in K\right\} $.
Furthermore, all the matrices we encounter, together their inverses
have eigenvalues bounded by $\left(\ks n\nu\ln\frac{1}{\epsilon}\right)^{O\left(1\right)}$,
where $\ks$ is the $\epsilon$-condition number of the barrier formulation
(Definition \ref{def:epsilon-condition-number}).
\end{cor}

\begin{algorithm}
\begin{algorithmic}[1]\begin{comment}
TODO: fix parameters
\end{comment}

\Require $\left\Vert \nabla g_{\mu}\left(\vy\right)\right\Vert _{\HH_{\vy}^{-1}}\leq\frac{1}{20\sqrt{n}}$,
$B_{\HH}=4\ks$, $B_{\HHt}=10^{9}\cdot\ks^{2}\cdot\left(n\ln\ks+\sqrt{\nu n}\ln\frac{n\nu}{\epsilon}\right),$$B=B_{\HHt}^{10}$,
$\beta=1/100$

\Ensure Returns a new iterate $\vyp$ such that $\left\Vert \nabla g_{\mu/\left(1+\frac{1}{10\nu}\right)}\left(\vyp\right)\right\Vert _{\HH_{\vyp}^{-1}}\leq\frac{1}{20\sqrt{n}}$.

\Procedure{PathFollowing}{$\vy,\HHt,\HHt^{-1},\mu$}

\State$t=0,T=\frac{100}{\beta}\ln\left(B_{\HH}B_{\HHt}\right),\vx=\vzero$

\While{$t<T$}

\State$\left(\HHtp,\HHtp^{-1},\vxp\right)=$\Call{RobustStepOrUpdate}{$\vy,\HHt,\HHt^{-1},\nabla g_{\mu}\left(\vy\right),\vx$}

\If{$\HHtp=\HHt$}\Comment{preconditioner did not change, so progress was made}

\State$\vx=\vxp$, $t=t+1$

\EndIf

\EndWhile

\State\Return$\left(\vy-\vx,\HHt,\HHt^{-1},\mu/\left(1+\frac{1}{200\sqrt{n\nu}}\right)\right)$

\EndProcedure

\end{algorithmic}

\medskip{}

\caption{Path Following.\label{alg:Path-Following}}
\end{algorithm}

\subsection{Excentricity Proofs\label{subsec:excent-proofs}}

First we provide the proof of Lemma \ref{lem:new-excent}.
\begin{proof}
[Proof of Lemma \ref{lem:new-excent}] Using Lemma \ref{lem:excprod}
we write
\begin{align*}
\exc\left(\HHt^{-1}\HH_{\vy+\vdelta}\right) & =\exc\left(\HHt^{-1}\HH_{\vy}\cdot\HH_{\vy}^{-1}\HH_{\vy+\vdelta}\right)\leq\exc\left(\HHt^{-1}\HH_{\vy}\right)\cdot\sqrt{\frac{\det\left(\left(\HH_{\vy}^{-1/2}\HH_{\vy+\vdelta}\HH_{\vy}^{-1/2}\right)_{\geq1}\right)}{\det\left(\left(\HH_{\vy}^{-1/2}\HH_{\vy+\vdelta}\HH_{\vy}^{-1/2}\right)_{<1}\right)}}\,.
\end{align*}
Let $\lambda_{1},\dots,\lambda_{n}$ be the eigenvalues of $\HH_{\vy}^{-1/2}\HH_{\vy+\vdelta}\HH_{\vy}^{-1/2}$.
Since $g$ is strongly self-concordant, we also have by Lemma \ref{lem:strongly-self-concordant-property},
for $\alpha=\frac{\left\Vert \vdelta\right\Vert _{\HH_{\vy}}}{\left(1-\left\Vert \vdelta\right\Vert _{\HH_{\vy}}\right)^{2}}$,
that
\[
\sqrt{\sum_{i=1}^{n}\left(\lambda_{i}-1\right)^{2}}\leq\alpha\,.
\]
Therefore, using $\frac{1}{1-x}\leq1+\beta x$ whenever $0\leq x\leq1-\frac{1}{\beta},$
we have
\[
\left|\lambda_{i}-1\right|\leq\alpha=1-\frac{1}{\left(1-\alpha\right)^{-1}}\,,
\]
so that
\[
\frac{1}{1-\left|\lambda_{i}-1\right|}\leq1+\frac{1}{1-\alpha}\left|\lambda_{i}-1\right|\,.
\]
Thus, 
\begin{align*}
\prod_{i=1}^{n}\max\left\{ \lambda_{i},\frac{1}{\lambda_{i}}\right\}  & =\prod_{i=1}^{n}\max\left\{ 1+\left(\lambda_{i}-1\right),\frac{1}{1+\left(\lambda_{i}-1\right)}\right\} \leq\prod_{i=1}^{n}\max\left\{ 1+\left|\lambda_{i}-1\right|,\frac{1}{1-\left|\lambda_{i}-1\right|}\right\} \\
 & \leq\prod_{i=1}^{n}\max\left\{ 1+\left|\lambda_{i}-1\right|,1+\frac{1}{1-\alpha}\left|\lambda_{i}-1\right|\right\} \\
 & \leq\prod_{i=1}^{n}\left(1+\frac{1}{1-\alpha}\left|\lambda_{i}-1\right|\right)\\
 & \leq\left(1+\frac{1}{n}\sum_{i=1}^{n}\frac{1}{1-\alpha}\left|\lambda_{i}-1\right|\right)^{n}\tag{AM-GM}\\
 & \leq\left(1+\frac{1}{1-\alpha}\sqrt{\frac{1}{n}\sum_{i=1}^{n}\left|\lambda_{i}-1\right|^{2}}\right)^{n}\tag{QM-AM}\\
 & \leq\left(1+\frac{\alpha}{1-\alpha}\sqrt{\frac{1}{n}}\right)^{n}\\
 & \leq\exp\left(\sqrt{n}\cdot\frac{\alpha}{1-\alpha}\right)\,.
\end{align*}
Plugging this back into the bound on the new excentricity we obtain
\begin{align*}
\exc\left(\HHt^{-1}\HH_{\vy+\vdelta}\right)
&\leq
\exc\left(\HHt^{-1}\HH_{\vy}\right)\cdot\exp\left(\frac{1}{2}\sqrt{n}\cdot\frac{\alpha}{1-\alpha}\right)
\\
&=
\exc\left(\HHt^{-1}\HH_{\vy}\right)\cdot\exp\left(\frac{1}{2}\sqrt{n}\cdot\frac{\left\Vert \vdelta\right\Vert _{\HH_{\vy}}}{\left(1-\left\Vert \vdelta\right\Vert _{\HH_{\vy}}\right)^{2}-\left\Vert \vdelta\right\Vert _{\HH_{\vy}}}\right)\,,
\end{align*}
which concludes the proof.
\end{proof}
To prove Lemma \ref{lem:new-excent} we rely on the following important
inequality.\begin{comment}
TODO: expand
\end{comment}

\begin{lem}
\label{lem:eig-ineq}Let $\AA$ and $\DD$ be symmetric positive definite
matrices. Let $\DD_{\geq1}$ denote the matrix obtained from $\DD$
by increasing all the sub-unitary eigenvalues to $1$, and let $\DD_{<1}$
be defined analogously. Then 
\[
\det\left(\AA\DD+\Id\right)\leq\det\left(\AA+\Id\right)\det\left(\DD_{\text{\ensuremath{\geq1}}}\right)
\]
\end{lem}

The proof follows from majorization bounds on the eigenvalues of a
sum of two PSD matrices.
\begin{defn}
[majorization]Given $\vx,\vy\in\mathbb{R}^{d}$ we say that $\vy$
majorizes $\vx$ (denoted by $\vx\prec\vy$) iff 
\[
\sum_{i=1}^{k}\vx_{i}\leq\sum_{i=1}^{k}\vy_{i}\,,\text{ for all }\text{1\ensuremath{\leq k\leq d}},
\]
and 
\begin{align*}
\sum_{i=1}^{d}\vx_{i} & =\sum_{i=1}^{d}\vy_{i}\,.
\end{align*}
\end{defn}

Given two symmetric positive definite matrices, one can show a majorization
relation involving their eigenvalues. The following lemma follows
from \cite{ando1994majorizations}.
\begin{lem}
[\cite{ando1994majorizations}]\label{lem:majorization}Given a
symmetric PSD matrix $\MM$, let $\lambda^{\uparrow}\left(\MM\right)$
and $\lambda^{\downarrow}\left(\MM\right)$ denote the eigenvalues
of $\MM$ in ascending and descending order, respectively. Then
\[
\lambda^{\downarrow}\left(\AA\right)+\lambda^{\uparrow}\left(\DD^{-1}\right)\prec\lambda^{\downarrow}\left(\AA+\DD^{-1}\right)\,.
\]
\end{lem}

Such bounds are helpful, since the theory of Schur convexity allows
to prove inequalities involving functions applied to vectors related
by a majorization relation.
\begin{defn}
[Schur-convexity]A function $f:\mathbb{R}^{d}\rightarrow\mathbb{R}$
is called Schur-convex if for any $\vx,\vy\in\mathbb{R}^{d}$ such
that $\vx\prec\vy$, one has $f\left(\vx\right)\leq f\left(\vy\right)$.
A function $f$ is called Schur-concave if the reverse inequality
holds, i.e. $\vx\prec\vy$ implies $f\left(\vx\right)\geq f\left(\vy\right)$.
\end{defn}

There is extensive theory on Schur-convex/Schur-concave functions
\cite{varberg1973convex}. Here we are only concerned with a particular
function which is the elementary symmetric polynomial.
\begin{lem}
\label{lem:schur-concave-elementary}The function $f:\mathbb{R}_{>0}^{d}\rightarrow\mathbb{R}$
defined by $f\left(\vx\right)=\prod_{i=1}^{d}\vx_{i}$ is Schur-concave.
\end{lem}

Using Lemma \ref{lem:schur-concave-elementary}, together with \ref{lem:majorization}
we are ready to prove Lemma \ref{lem:eig-ineq}.
\begin{proof}
[Proof of Lemma\ref{lem:eig-ineq}]Let $\lambda^{\uparrow}\left(\MM\right)$
and $\lambda^{\downarrow}\left(\MM\right)$ denote the eigenvalues
of $\MM$ in ascending and descending order, respectively. By Lemma
\ref{lem:majorization}, the eigenvalues satisfy the majorization
relation:
\[
\lambda^{\downarrow}\left(\AA\right)+\lambda^{\uparrow}\left(\DD^{-1}\right)\prec\lambda^{\downarrow}\left(\AA+\DD^{-1}\right)\,.
\]
Using \ref{lem:schur-concave-elementary} we know that the map $\left(x_{1},\dots,x_{n}\right)\rightarrow\prod_{i=1}^{n}x_{i}$
is an elementary symmetric function, it is Schur concave whenever
all $x_{i}>0$. Applying it to $\lambda^{\downarrow}\left(\AA\right)+\lambda^{\uparrow}\left(\DD^{-1}\right)$
and $\lambda^{\downarrow}\left(\AA+\DD^{-1}\right)$, respectively,
we obtain:
\begin{align*}
\det\left(\AA+\DD^{-1}\right) & =\prod_{j}\lambda_{j}^{\downarrow}\left(\AA+\DD^{-1}\right)\leq\prod_{j}\left(\lambda_{j}^{\downarrow}\left(\AA\right)+\lambda_{j}^{\uparrow}\left(\DD^{-1}\right)\right)\\
 & \leq\prod_{j}\begin{cases}
\left(\lambda_{j}^{\downarrow}\left(\AA\right)+1\right)\lambda_{j}^{\uparrow}\left(\DD^{-1}\right) & \lambda_{j}^{\uparrow}\left(\DD^{-1}\right)\geq1\\
\lambda_{j}^{\downarrow}\left(\AA\right)+1 & \lambda_{j}^{\uparrow}\left(\DD^{-1}\right)\leq1
\end{cases}\\
 & =\prod_{j}\left(\lambda_{j}^{\downarrow}\left(\AA\right)+1\right)\cdot\prod_{j}\max\left\{ 1,\lambda_{j}^{\uparrow}\left(\DD^{-1}\right)\right\} \\
 & =\det\left(\AA+\Id\right)\cdot\prod_{j}\frac{1}{\min\left\{ 1,\lambda_{j}^{\uparrow}\left(\DD\right)\right\} }\\
 & =\det\left(\AA+\Id\right)\frac{1}{\det\left(\DD_{<1}\right)}\,.
\end{align*}
Therefore
\begin{align*}
\det\left(\AA\DD+\Id\right)
=
\det\left(\DD\right)\det\left(\AA+\DD^{-1}\right)\leq\det\left(\DD\right)\det\left(\AA+\Id\right)\frac{1}{\det\left(\DD_{<1}\right)}
\\
=
\det\left(\AA+\Id\right)\det\left(\DD_{\geq1}\right)\,,
\end{align*}
which concludes the proof.
\end{proof}

\section{Solving SDPs}

Here we provide the main theorem on solving semidefinite programs.
We provide a statement that matches the form of the one in \cite{jiang2020faster}.
\begin{thm}
\label{thm:sdp}Consider a semidefinite program with variable size
$n\times n$ and $m\geq n$ constraints of the form
\[
\max\left\{ \left\langle \BB,\XX\right\rangle :\XX\succeq0,\left\langle \AA_{i},\XX\right\rangle =\vc_{i},\text{ for all }1\leq i\leq m\right\} \,.
\]
Assume that any feasible solution $\XX\succeq0$ satisfies $\left\Vert \XX\right\Vert \leq R$.
Then for any error parameter $0<\epsilon<0.01$, there is an algorithm
that outputs in time
\[
\widetilde{O}\left(mn^{4}+m^{1.25}n^{3.5}\ln\frac{1}{\epsilon}\right)
\]
a matrix $\XX\succeq0$ such that
\[
\left\langle \BB,\XX\right\rangle \geq\left\langle \BB,\XXs\right\rangle -\epsilon\left\Vert \BB\right\Vert R
\]
and
\[
\sum_{i=1}^{m}\left|\left\langle \AA_{i},\XX\right\rangle -\vc_{i}\right|\leq4n\epsilon\left(R\sum_{i=1}^{m}\left\Vert \AA_{i}\right\Vert _{1}+\left\Vert \vc\right\Vert _{1}\right)\,.
\]
Furthermore, provided that the barrier objective (\ref{eq:sdp-barrier-obj}))
has a quasi-polynomial $\epsilon$-condition number (Definition \ref{def:epsilon-condition-number}),
all the matrices we encounter, together with their inverses, have
quasi-polynomially bounded eigenvalues.
\end{thm}

To prove this statement we require mapping back the solution obtained
from the solution to the barrier objective to a solution to the primal
objective. We point the reader to Theorem 5.1 of \cite{jiang2020faster}
for a careful treatment of this matter. Similarly for initializing
the interior point method, please consult Section 9 from \cite{jiang2020faster}.
With these technicalities settled, it suffices to understand how fast
one can solve the corresponding barrier objective. 

To apply Theorem \ref{thm:main-opt} we need to find an appropriate
strongly self-concordant barrier for the feasible set. Unfortunately,
in the case of SDPs we can not prove that the standard logarithmic
barrier is strongly self-concordant. In fact, this does not even appear
to be true. Instead, we show that after scaling it up by a factor
of $\sqrt{m}$ it becomes strongly self-concordant. The corresponding
self-concordance parameter of this scaled barrier becomes $\nu=n\sqrt{m}$.
We prove this in Lemma \ref{lem:strong-selfc-logdet}.

Note that the gradients of the scaled barrier
\[
\phi\left(\vy\right)=-\sqrt{m}\cdot\log\det\left(\BB-\sum_{i=1}^{m}\vy_{i}\AA_{i}\right)
\]
defined by 
\[
\left[\nabla\phi\left(\vy\right)\right]_{i}=\sqrt{m}\cdot\left\langle \AA_{i},\left(\BB-\sum_{i=1}^{m}\vy_{i}\AA_{i}\right)^{-1}\right\rangle 
\]
can be evaluated in time $\mathcal{T}_{\text{gradient}}=O\left(mn^{2}+n^{\omega}\right)$.
To do so we first evaluate the slack matrix $\SS_{\vy}=\BB-\sum_{i=1}^{m}\vy_{i}\AA_{i}$
in time $O\left(mn^{2}\right)$, we invert it in time $O\left(n^{\omega}\right)$,
then we evaluate all the inner products $\left\langle \AA_{i},\SS_{\vy}^{-1}\right\rangle $
in time $O\left(mn^{2}\right)$. 

Next we use Theorem \ref{thm:main-opt}. Here, although the fake Hessian
we maintain has dimension $m\times m$, we note that its rank is always
at most $n^{2}$. This is because 
\[
\nabla^{2}\phi\left(\vy\right)=\sqrt{m}\cdot\mathcal{A}^{\top}\left(\SS_{\vy}^{-1}\otimes\SS_{\vy}^{-1}\right)\mathcal{A}\,,
\]
where $\mathcal{A}\in\mathbb{R}^{n^{2}\times m}$ is the matrix whose
$i^{th}$ column is obtained by flattening $\AA_{i}$ into a vector
$\text{vec}\left(\AA_{i}\right)$ of length $n^{2}$. Therefore, per
our discussion from Section \ref{par:low-rank} we can run the entire
analysis in a subspace of ambient dimension $\min\left\{ m,n^{2}\right\} $.
By Theorem \ref{thm:main-opt}, after bounding $\ln\exc\left(\HH_{\vy_{0}}\right)=\widetilde{O}\left(\min\left\{ m,n^{2}\right\} \right)$,
we require time
\begin{align*}
 & \widetilde{O}\left(\left(\text{\ensuremath{\min\left\{ m,n^{2}\right\} }}+\sqrt{n\sqrt{m}\cdot\min\left\{ m,n^{2}\right\} }\ln\frac{nm}{\epsilon}\right)\cdot\left(n^{2}+\mathcal{T}_{\text{gradient}}\right)\right)\\
= \,& \widetilde{O}\left(\left(n^{2}+n^{3/2}m^{1/4}\ln\frac{1}{\epsilon}\right)\left(mn^{2}+n^{\omega}\right)\right)
\end{align*}
to solve the dual problem to $\epsilon$ precision. Assuming that
$m\geq n$, this time is
\[
\widetilde{O}\left(mn^{4}+m^{1.25}n^{3.5}\ln\frac{1}{\epsilon}\right)\,,
\]
which gives our claimed bound.

\begin{lem}
\label{lem:strong-selfc-logdet}The barrier $\phi\left(\vy\right)=-\sqrt{m}\cdot\log\det\left(\BB-\sum_{i=1}^{m}\vy_{i}\AA_{i}\right)$
is a strongly self-concordant barrier function with $\nu=n\sqrt{m}$. 
\end{lem}

\begin{proof}
Let $\psi\left(\vy\right)=-\log\det\left(\BB-\sum_{i=1}^{m}\vy_{i}\AA_{i}\right)$
be the standard logarithmic barrier. Let $\HH_{\vy}=\nabla^{2}\psi\left(\vy\right)$.
We will show that for any $\vh$ it satisfies
\[
\left\Vert \HH_{\vy}^{-1/2}\frac{d}{dt}\HH_{\vy+t\vh}\HH_{\vy}^{-1/2}\right\Vert _{F}\leq\frac{2}{\sqrt{m}}\left\Vert \vh\right\Vert _{\HH_{\vy}}\,.
\]
Hence scaling it up by a factor of $\sqrt{m}$ yields a $n\sqrt{m}$-self-concordant
barrier function that matches the requirements from Definition \ref{def:strong-self-conc}.
To verify this, we first compute the entries of the Hessian matrix
\begin{align*}
\left(\HH_{\vy}\right)_{ij} & =\text{tr}\left(\SS_{\vy}^{-1}\AA_{i}\SS_{\vy}^{-1}\AA_{j}\right)\,,
\end{align*}
which allows us to write
\[
\HH_{\vy}=\mathcal{A}^{\top}\left(\SS_{\vy}^{-1}\otimes\SS_{\vy}^{-1}\right)\mathcal{A}\,,
\]
where $\SS_{\vy}=\BB-\sum_{i=1}^{m}\vy_{i}\AA_{i}$ is the slack matrix,
$\mathcal{A}\in\mathbb{R}^{n^{2}\times m}$ is the matrix whose $i^{th}$
column is obtained by flattening $\AA_{i}$ into a vector of length
$n^{2}$, and $\otimes$ denotes the Kronecker product. For any $\vh\in\mathbb{R}^{m}$,
let $\Delta\SS\left(\vh\right)=\SS_{\vy+\vh}-\SS_{\vy}=-\sum_{i}\vh_{i}\AA_{i}$.
Therefore, using the expansion $\left(\XX+t\YY\right)^{-1}=\XX^{-1}-t\XX^{-1}\YY\XX^{-1}+O_{\XX,\YY}\left(t^{2}\right)\Id$,
we calculate :
\begin{align*}
\frac{d}{dt}\HH_{\vy+t\vh} & =\lim_{t\rightarrow0}\frac{1}{t}\cdot\mathcal{A}^{\top}\left(\SS_{\vy}^{-1}\otimes\SS_{\vy}^{-1}-\SS_{\vy+t\vdelta}^{-1}\otimes\SS_{\vy+t\vdelta}^{-1}\right)\mathcal{A}\\
 & =\lim_{t\rightarrow0}\frac{1}{t}\cdot\mathcal{A}^{\top}\left(\SS_{\vy}^{-1}\otimes\SS_{\vy}^{-1}-\left(\SS_{\vy}+\Delta\SS\left(t\vdelta\right)\right)^{-1}\otimes\left(\SS_{\vy}+\Delta\SS\left(t\vdelta\right)\right)^{-1}\right)\mathcal{A}\\
 & =\lim_{t\rightarrow0}\frac{1}{t}\cdot\mathcal{A}^{\top}\left(\SS_{\vy}^{-1}\otimes\SS_{\vy}^{-1}-\left(\SS_{\vy}^{-1}-t\SS_{\vy}^{-1}\Delta\SS\left(\vdelta\right)\SS_{\vy}^{-1}\right)\otimes\left(\SS_{\vy}^{-1}-t\SS_{\vy}^{-1}\Delta\SS\left(\vdelta\right)\SS_{\vy}^{-1}\right)\right)\mathcal{A}\\
 & =\lim_{t\rightarrow0}\frac{1}{t}\cdot\mathcal{A}^{\top}\left(\SS_{\vy}^{-1}\otimes\SS_{\vy}^{-1}-\left(\SS_{\vy}^{-1}-t\SS_{\vy}^{-1}\Delta\SS\left(\vdelta\right)\SS_{\vy}^{-1}\right)\otimes\left(\SS_{\vy}^{-1}-t\SS_{\vy}^{-1}\Delta\SS\left(\vdelta\right)\SS_{\vy}^{-1}\right)\right)\mathcal{A}\\
 & =\mathcal{A}^{\top}\left(\SS_{\vy}^{-1}\otimes\SS_{\vy}^{-1}\Delta\SS\left(\vdelta\right)\SS_{\vy}^{-1}+\SS_{\vy}^{-1}\Delta\SS\left(\vdelta\right)\SS_{\vy}^{-1}\otimes\SS_{\vy}^{-1}\right)\mathcal{A}\,.
\end{align*}
Using the inequality 
\[
\left\Vert \left(\MM^{\top}\XX\MM\right)^{-1/2}\left(\MM^{\top}\YY\MM\right)\left(\MM^{\top}\XX\MM\right)^{-1/2}\right\Vert _{F}\leq\left\Vert \XX^{-1/2}\YY\XX^{-1/2}\right\Vert _{F}=\left\Vert \XX^{-1}\YY\right\Vert _{F}\,,
\]
we can bound the Frobenius norm
\begin{align}
 & \left\Vert \HH_{\vy}^{-1/2}\frac{d}{dt}\HH_{\vy+t\vh}\HH_{\vy}^{-1/2}\right\Vert _{F}\nonumber \\
 & \leq\left\Vert \left(\SS_{\vy}^{-1}\otimes\SS_{\vy}^{-1}\right)^{-1}\left(\SS_{\vy}^{-1}\otimes\SS_{\vy}^{-1}\Delta\SS\left(\vdelta\right)\SS_{\vy}^{-1}+\SS_{\vy}^{-1}\Delta\SS\left(\vdelta\right)\SS_{\vy}^{-1}\otimes\SS_{\vy}^{-1}\right)\right\Vert _{F}\nonumber \\
 & =\left\Vert \left(\SS_{\vy}\otimes\SS_{\vy}\right)\left(\SS_{\vy}^{-1}\otimes\SS_{\vy}^{-1}\Delta\SS\left(\vdelta\right)\SS_{\vy}^{-1}+\SS_{\vy}^{-1}\Delta\SS\left(\vdelta\right)\SS_{\vy}^{-1}\otimes\SS_{\vy}^{-1}\right)\right\Vert _{F}\nonumber \\
 & =\left\Vert \Id_{m}\otimes\Delta\SS\left(\vdelta\right)\SS_{\vy}^{-1}+\Delta\SS\left(\vdelta\right)\SS_{\vy}^{-1}\otimes\Id_{m}\right\Vert _{F}\nonumber \\
 & \leq2\left\Vert \Id_{m}\otimes\Delta\SS\left(\vdelta\right)\SS_{\vy}^{-1}\right\Vert _{F}\nonumber \\
 & =2\left\Vert \Id_{m}\right\Vert _{F}\left\Vert \Delta\SS\left(\vdelta\right)\SS_{\vy}^{-1}\right\Vert _{F}\nonumber \\
 & =2\sqrt{m}\left\Vert \SS_{\vy}^{-1}\Delta\SS\left(\vdelta\right)\right\Vert _{F}\,.\label{eq:frob-norm-hess-change}
\end{align}
Give a matrix $\XX$ let $\text{vec}\left(\XX\right)$ be its flattening
into a vector. Using the identity $\left(\BB^{\top}\otimes\AA\right)\text{vec}\left(\XX\right)=\text{vec}\left(\AA\XX\BB\right)$
we obtain:

\begin{align}
\left\Vert \vh\right\Vert _{\HH_{\vy}}^{2} & =\vh^{\top}\mathcal{A}^{\top}\left(\SS_{\vy}^{-1}\otimes\SS_{\vy}^{-1}\right)\mathcal{A}\vh\nonumber \\
 & =\text{vec}\left(\Delta\SS\left(\vdelta\right)\right)^{\top}\left(\SS_{\vy}^{-1}\otimes\SS_{\vy}^{-1}\right)\text{vec\ensuremath{\left(\Delta\SS\left(\vdelta\right)\right)}}\nonumber \\
 & =\left\langle \text{vec}\left(\Delta\SS\left(\vdelta\right)\right),\text{vec}\left(\SS_{\vy}^{-1}\Delta\SS\left(\vdelta\right)\SS_{\vy}^{-1}\right)\right\rangle \nonumber \\
 & =\left\langle \Delta\SS\left(\vdelta\right),\SS_{\vy}^{-1}\Delta\SS\left(\vdelta\right)\SS_{\vy}^{-1}\right\rangle \nonumber \\
 & =\left\Vert \SS_{\vy}^{-1}\Delta\SS\left(\vdelta\right)\right\Vert _{F}^{2}\,.\label{eq:local-norm-update-sdp}
\end{align}
Combining (\ref{eq:frob-norm-hess-change}) and (\ref{eq:local-norm-update-sdp})
we obtain the desired bound on the relative change in Hessian. 
Finally, since $\psi\left(\vy\right)$ is $n$-self-concordant by
standard arguments, scaling it by $\sqrt{m}$ shows that $\nabla\phi\left(\vy\right)^{\top}\left(\nabla^{2}\phi\left(\vy\right)\right)^{-1}\nabla\phi\left(\vy\right)=\sqrt{m}\cdot\nabla\psi\left(\vy\right)^{\top}\left(\nabla^{2}\psi\left(\vy\right)\right)^{-1}\nabla\psi\left(\vy\right)\leq n\sqrt{m}$.
This concludes the proof.
\end{proof}

\section*{Acknowledgements. }

We acknowledge the support of the French Agence Nationale de la Recherche
(ANR), under grant ANR-21-CE48-0016 (project COMCOPT), and the support
of CNRS with a CoopIntEER IEA grant (project ALFRED).

\newpage 

\global\long\def\vzero{\boldsymbol{\mathit{0}}}\global\long\def\vx{\boldsymbol{\mathit{x}}}\global\long\def\vb{\boldsymbol{\mathit{b}}}\global\long\def\vv{\boldsymbol{\mathit{v}}}\global\long\def\vu{\boldsymbol{\mathit{u}}}\global\long\def\vr{\boldsymbol{\mathit{r}}}\global\long\def\vDelta{\boldsymbol{\mathit{\Delta}}}\global\long\def\vz{\boldsymbol{\mathit{z}}}\global\long\def\vs{\boldsymbol{\mathit{s}}}\global\long\def\vrho{\boldsymbol{\mathit{\rho}}}\global\long\def\vdelta{\boldsymbol{\mathit{\delta}}}\global\long\def\vDeltat{\boldsymbol{\widetilde{\mathit{\Delta}}}}\global\long\def\vc{\boldsymbol{\mathit{c}}}\global\long\def\vh{\boldsymbol{\mathit{h}}}\global\long\def\vup{\boldsymbol{\mathit{u'}}}\global\long\def\vrp{\boldsymbol{\mathit{r'}}}

\global\long\def\vxh{\boldsymbol{\mathit{\widehat{x}}}}\global\long\def\vbh{\boldsymbol{\mathit{\widehat{b}}}}\global\long\def\vxhp{\boldsymbol{\mathit{\widehat{x}'}}}\global\long\def\vxp{\boldsymbol{\mathit{x'}}}\global\long\def\vxs{\boldsymbol{\mathit{x^{\star}}}}\global\long\def\vsp{\boldsymbol{\mathit{s'}}}\global\long\def\vy{\boldsymbol{\mathit{y}}}\global\long\def\vyp{\boldsymbol{\mathit{y'}}}\global\long\def\vg{\boldsymbol{\mathit{g}}}\global\long\def\vrt{\boldsymbol{\mathit{\widetilde{r}}}}\global\long\def\vrtp{\boldsymbol{\mathit{\widetilde{r}'}}}\global\long\def\vys{\boldsymbol{\mathit{y^{\star}}}}

\global\long\def\vrh{\boldsymbol{\mathit{\widehat{r}}}}\global\long\def\vrhp{\boldsymbol{\mathit{\widehat{r}'}}}

\global\long\def\exc{\mathcal{E}}

\global\long\def\HH{\boldsymbol{\mathit{H}}}\global\long\def\HHtil{\boldsymbol{\mathit{\widetilde{H}}}}\global\long\def\XX{\boldsymbol{\mathit{X}}}\global\long\def\XXp{\boldsymbol{\mathit{X'}}}\global\long\def\Id{\boldsymbol{\mathit{I}}}\global\long\def\PP{\boldsymbol{\mathit{P}}}\global\long\def\YY{\boldsymbol{\mathit{Y}}}

\global\long\def\HHh{\boldsymbol{\mathit{\widehat{H}}}}\global\long\def\HHb{\boldsymbol{\mathit{\overline{H}}}}\global\long\def\HHt{\boldsymbol{\mathit{\widetilde{H}}}}\global\long\def\HHtp{\boldsymbol{\mathit{\widetilde{H}'}}}\global\long\def\HHp{\boldsymbol{\mathit{H'}}}

\global\long\def\AA{\boldsymbol{A}}\global\long\def\DD{\boldsymbol{D}}\global\long\def\MM{\boldsymbol{M}}\global\long\def\RR{\boldsymbol{R}}\global\long\def\SS{\mathit{\boldsymbol{S}}}\global\long\def\SSp{{\it \mathit{\boldsymbol{S'}}}}\global\long\def\BB{\boldsymbol{B}}\global\long\def\CC{\boldsymbol{C}}\global\long\def\XXs{\boldsymbol{X^{\star}}}

\global\long\def\PProj{\mathit{\boldsymbol{\Pi}}}

\global\long\def\ks{\kappa_{\star}}

\global\long\def\diag#1{\mathbb{D}\left(#1\right)}\global\long\def\epsilon{\varepsilon}\global\long\def\ln{\log}

\appendix

\section{Proofs from Section \ref{sec:prelim}}

\subsection{Proof of Lemma \ref{lem:excentricity-update}\label{subsec:Proof-of-Lemma-excentricity-update}}
\begin{proof}
We use the identity
\begin{equation}
\begin{aligned}
\exc\left(\XX\right)=\det\left(\frac{\XX^{1/2}+\XX^{-1/2}}{2}\right)
=
\frac{1}{2^{n}}\det\left(\left(\XX+\Id\right)\XX^{-1/2}\right) 
\\
=
\frac{1}{2^{n}}\det\left(\XX+\Id\right)\det\left(\XX^{1/2}\right)^{-1}
=
\frac{1}{2^{n}}\frac{\det\left(\XX+\Id\right)}{\sqrt{\det\left(\XX\right)}}\ .\label{eq:excsimple}
\end{aligned}
\end{equation}
Therefore we have
\begin{align*}
\mathcal{E}\left(\XXp\right) & =\frac{1}{2^{n}}\frac{\det\left(\XXp+\Id\right)}{\sqrt{\det\left(\XXp\right)}}=\frac{1}{2^{n}}\frac{\det\left(\XX+\Id+\vu\vu^{\top}\right)}{\sqrt{\det\left(\XX+\vu\vu^{\top}\right)}}=\frac{1}{2^{n}}\frac{\det\left(\XX+\Id\right)\cdot\left(1+\vu^{\top}\left(\Id+\XX\right)^{-1}\vu\right)}{\sqrt{\det\left(\XX\right)\left(1+\vu^{\top}\XX^{-1}\vu\right)}}\\
 & =\mathcal{E}\left(\XX\right)\cdot\frac{1+\vu^{\top}\left(\Id+\XX\right)^{-1}\vu}{\sqrt{1+\vu^{\top}\XX^{-1}\vu}}\ ,
\end{align*}
where we used the identity $\det\left(\XX+\vv\vv^{\top}\right)=\det\left(\XX\right)\left(1+\vv^{\top}\XX^{-1}\vv\right)$.
\end{proof}

\subsection{Proof of Lemma \ref{lem:excent-upd}\label{subsec:Proof-of-Lemma-excent-upd}}

\begin{proof}
Per Lemma \ref{lem:excentricity-update} we need to upper bound
\[
\frac{1+\vu^{\top}\left(\Id+\XX\right)^{-1}\vu}{\sqrt{1+\vu^{\top}\XX^{-1}\vu}}\leq\frac{1+\vu^{\top}\vu}{\sqrt{1+\gamma}}=\frac{2}{\sqrt{1+\gamma}}\ .
\]
We analyze the second case using the fact that $\exc\left(\XX\right)=\exc\left(\XX^{-1}\right)$
and hence we can make a rank-$1$ update on $\XX^{-1}$. More precisely,
via the Sherman-Morrison formula (Lemma \ref{lem:shermanmorrison})
we have:
\[
\frac{\exc\left(\XX-\frac{\XX\vu\vu^{\top}\XX}{1+\vu^{\top}\XX\vu}\right)}{\exc\left(\XX\right)}=\frac{\exc\left(\left(\XX^{-1}+\vu\vu^{\top}\right)^{-1}\right)}{\exc\left(\XX\right)}=\frac{\exc\left(\XX^{-1}+\vu\vu^{\top}\right)}{\exc\left(\XX^{-1}\right)}\leq\frac{2}{\sqrt{1+\gamma}}\ .
\]
\end{proof}

\subsection{Proof of Lemma \ref{lem:Excentricity-similarity}}
\begin{proof}
We apply the formula from (\ref{eq:excsimple}), and using $\det\left(\AA\BB\right)=\det\left(\AA\right)\det\left(\BB\right)$:
\begin{align*}
\mathcal{E}\left(\XX\right)
&=
\frac{1}{2^{n}}\frac{\det\left(\XX+\Id\right)}{\sqrt{\det\left(\XX\right)}} 
=
\frac{1}{2^{n}}\frac{\det\left(\YY\right)\det\left(\XX+\Id\right)\det\left(\YY^{-1}\right)}{\sqrt{\det\left(\YY\right)\det\left(\XX\right)\det\left(\YY^{-1}\right)}}
\\
&=
\frac{1}{2^{n}}\frac{\det\left(\YY\XX\YY^{-1}+\Id\right)}{\sqrt{\det\left(\YY\XX\YY^{-1}\right)}}=
\exc\left(\YY\XX\YY^{-1}\right)\,.
\end{align*}
The second property holds by definition.
\end{proof}

\subsection{Proof of Lemma \ref{lem:well-conditioned-neighborhood}\label{subsec:Proof-of-Lemma-well-conditioned-neighborhood}}

\begin{proof}
From the self-concordance property we have that if $\vys$ is the
minimizer of $g_{\mu}$, then from standard arguments \cite{renegar2001mathematical}
it follows that 
\begin{align*}
\left\Vert \vy-\vys\right\Vert _{\HH_{\vys}} & \leq\frac{\left\Vert \nabla g_{\mu}\left(\vy\right)\right\Vert _{\HH_{\vy}^{-1}}}{1-\left\Vert \nabla g_{\mu}\left(\vy\right)\right\Vert _{\HH_{\vy}^{-1}}}\leq\frac{1}{2}\,.
\end{align*}
Therefore
\begin{align*}
\HH_{\vy}\cdot\left(1-\left\Vert \vy-\vys\right\Vert _{\HH_{\vy}}\right)^{2} & \preceq\HH_{\vys}\preceq\HH_{\vy}\cdot\left(\frac{1}{1-\left\Vert \vy-\vys\right\Vert _{\HH_{\vy}}}\right)^{2}
\end{align*}
which gives 
\[
\HH_{\vys}\cdot\frac{1}{4}\preceq\HH_{\vy}\preceq\HH_{\vys}\cdot4\,.
\]
Using the bound from the hypothesis we obtain $\max\left\{ \left\Vert \HH_{\vys}\right\Vert ,\left\Vert \HH_{\vys}^{-1}\right\Vert \right\} \leq\kappa\left(\mu\right)$,
which yields the claim.
\end{proof}

\section{Proofs from Section \ref{sec:linear}}

\subsection{Excentricity Certificates from the Richardson Iteration\label{subsec:Excentricity-Certificates-from}}
\begin{lem}
\label{lem:Richardson-progress-certificate}Let $\HH\in\mathbb{R}^{n\times n}$,
and vectors $\vb,\vx\in\mathbb{R}^{n}$. Let $\vr=\vb-\HH\vx$, and
consider the step
\[
\vxp=\vx+\frac{\left\langle \vr,\HH\vr\right\rangle }{\left\Vert \HH\vr\right\Vert _{2}^{2}}\vr\,.
\]
Let $\beta\in\left(0,1\right)$ be a scalar. Provided that the new
residual $\vrp=\vb-\HH\vxp$ satisfies
\[
\frac{\left\Vert \vr\right\Vert _{2}^{2}-\left\Vert \vrp\right\Vert _{2}^{2}}{\left\Vert \vr\right\Vert _{2}^{2}}\leq\beta\,,
\]
we obtain at least one of the following excentricity certificates:
\begin{enumerate}
\item $\frac{\left\Vert \vr\right\Vert _{\HH}^{2}}{\left\Vert \vr\right\Vert _{2}^{2}}\leq\sqrt{\beta}\iff\frac{\left\Vert \HH^{1/2}\vr\right\Vert _{\HH^{-1}}^{2}}{\left\Vert \HH^{1/2}\vr\right\Vert _{2}^{2}}\geq\frac{1}{\sqrt{\beta}}$,
\item $\frac{\left\Vert \HH\vr\right\Vert _{2}^{2}}{\left\Vert \vr\right\Vert _{H}^{2}}\geq\frac{1}{\sqrt{\beta}}\iff\frac{\left\Vert \HH^{1/2}\vr\right\Vert _{\HH}^{2}}{\left\Vert \HH^{1/2}\vr\right\Vert _{2}^{2}}\geq\frac{1}{\sqrt{\beta}}$.
\end{enumerate}
\end{lem}

\begin{proof}
We note that the step size we consider is the one that minimizes the
$\ell_{2}$ norm of the new residual $\vrp$. Using this update, we
measure the new norm:
\begin{align*}
\left\Vert \vrp\right\Vert _{2}^{2} & =\left\Vert \vb-\HH\vxp\right\Vert _{2}^{2}=\left\Vert \vb-\HH\vx-\HH\left(\vxp-\vx\right)\right\Vert _{2}^{2}=\left\Vert \vr-\HH\left(\vxp-\vx\right)\right\Vert _{2}^{2}\\
 & =\left\Vert \vr-\HH\cdot\frac{\left\langle \vr,\HH\vr\right\rangle }{\left\Vert \HH\vr\right\Vert _{2}^{2}}\vr\right\Vert _{2}^{2}=\left\Vert \vr\right\Vert _{2}^{2}+\frac{\left\langle \vr,\HH\vr\right\rangle ^{2}}{\left\Vert \HH\vr\right\Vert _{2}^{4}}\left\Vert \HH\vr\right\Vert _{2}^{2}-2\cdot\frac{\left\langle \vr,\HH\vr\right\rangle ^{2}}{\left\Vert \HH\vr\right\Vert _{2}^{2}}\\
 & =\left\Vert \vr\right\Vert _{2}^{2}-\frac{\left\langle \vr,\HH\vr\right\rangle ^{2}}{\left\Vert \HH\vr\right\Vert _{2}^{2}}\,.
\end{align*}
Therefore we can write multiplicative progress as:
\[
\frac{\left\Vert \vr\right\Vert _{2}^{2}-\left\Vert \vrp\right\Vert _{2}^{2}}{\left\Vert \vr\right\Vert _{2}^{2}}=\frac{\left\langle \vr,\HH\vr\right\rangle ^{2}}{\left\Vert \HH\vr\right\Vert _{2}^{2}\left\Vert \vr\right\Vert _{2}^{2}}\leq\beta\,.
\]
Hence failure to make a lot of progress gives us that one of the two
following conditions must be true. Either $\left\Vert \vr\right\Vert _{2}\geq\left\Vert \HH\vr\right\Vert _{2}$,
in which case the upper bound on multiplicative progress implies that
\[
\frac{\left\Vert \vr\right\Vert _{\HH}^{4}}{\left\Vert \vr\right\Vert _{2}^{4}}\leq\beta\iff\frac{\left\Vert \vr\right\Vert _{\HH}^{2}}{\left\Vert \vr\right\Vert _{2}^{2}}\leq\sqrt{\beta}\iff\frac{\left\Vert \HH^{1/2}\vr\right\Vert _{\HH^{-1}}^{2}}{\left\Vert \HH^{1/2}\vr\right\Vert _{2}^{2}}\geq\frac{1}{\sqrt{\beta}}\ .
\]
Otherwise we must have $\left\Vert \HH\vr\right\Vert _{2}\geq\left\Vert \vr\right\Vert _{2}$,
in which case we have
\[
\frac{\left\Vert \HH\vr\right\Vert _{2}^{4}}{\left\Vert \vr\right\Vert _{H}^{4}}\geq\frac{1}{\beta}\iff\frac{\left\Vert \HH\vr\right\Vert _{2}^{2}}{\left\Vert \vr\right\Vert _{H}^{2}}\geq\frac{1}{\sqrt{\beta}}\iff\frac{\left\Vert \HH^{1/2}\vr\right\Vert _{\HH}^{2}}{\left\Vert \HH^{1/2}\vr\right\Vert _{2}^{2}}\geq\frac{1}{\sqrt{\beta}}\ .
\]
This concludes the proof.
\end{proof}

\subsection{Proof of Lemma \ref{lem:prec-Richardson-progress-certificate}\label{subsec:Proof-of-Lemma-proof-of-prec-richardson-progress}}
\begin{proof}
We perform the change of variable $\vy=\HHt^{1/2}\vx$ and consider
the equivalent system $\HHt^{-1/2}\HH\HHt^{-1/2}\cdot\vy=\HHt^{-1/2}\vb$.
Let the non-preconditioned residual $\vr=\vb-\HH\vx$, and let the
preconditioned residual $\vrh=\HHt^{-1/2}\left(\vb-\HH\HHt^{-1/2}\vy\right)=\HHt^{-1/2}\vr\,.$
Per Lemma \ref{lem:Richardson-progress-certificate}, the iteration
\begin{align*}
\vyp & =\vy+\frac{\left\langle \vrh,\left(\HHt^{-1/2}\HH\HHt^{-1/2}\right)\vrh\right\rangle }{\left\Vert \left(\HHt^{-1/2}\HH\HHt^{-1/2}\right)\vrh\right\Vert _{2}^{2}}\cdot\vrh
\end{align*}
maps back to
\begin{align*}
\vxp & =\HHt^{-1/2}\vyp=\HHt^{-1/2}\left(\vy+\frac{\left\langle \vrh,\left(\HHt^{-1/2}\HH\HHt^{-1/2}\right)\vrh\right\rangle }{\left\Vert \left(\HHt^{-1/2}\HH\HHt^{-1/2}\right)\vrh\right\Vert _{2}^{2}}\cdot\vrh\right)\\
 & =\vx+\frac{\left\langle \HHt^{-1/2}\vr,\left(\HHt^{-1/2}\HH\HHt^{-1/2}\right)\HHt^{-1/2}\vr\right\rangle }{\left\Vert \left(\HHt^{-1/2}\HH\HHt^{-1/2}\right)\HHt^{-1/2}\vr\right\Vert _{2}^{2}}\cdot\HHt^{-1/2}\cdot\HHt^{-1/2}\vr\\
 & =\vx+\frac{\left\Vert \HHt^{-1}\vr\right\Vert _{\HH}^{2}}{\left\Vert \HH\HHt^{-1}\vr\right\Vert _{\HHt^{-1}}^{2}}\cdot\HHt^{-1}\vr\,,
\end{align*}
in the original parametrization. The condition on the decrease in
$\left\Vert \vrh\right\Vert $ maps back to 
\[
\frac{\left\Vert \vrh\right\Vert _{2}^{2}-\left\Vert \vrhp\right\Vert _{2}^{2}}{\left\Vert \vrh\right\Vert _{2}^{2}}=\frac{\left\Vert \vr\right\Vert _{\HHt^{-1}}^{2}-\left\Vert \vrp\right\Vert _{\HHt^{-1}}^{2}}{\left\Vert \vr\right\Vert _{\HHt^{-1}}^{2}}\leq\beta\,.
\]
Finally, the provided excentricity certificates are either
\[
\frac{\left\Vert \vrh\right\Vert _{\HHt^{-1/2}\HH\HHt^{-1/2}}^{2}}{\left\Vert \vrh\right\Vert _{2}^{2}}=\frac{\left\Vert \HHt^{-1/2}\vr\right\Vert _{\HHt^{-1/2}\HH\HHt^{-1/2}}^{2}}{\left\Vert \HHt^{-1/2}\vr\right\Vert _{2}^{2}}=\frac{\left\Vert \HH^{1/2}\HHt^{-1}\vr\right\Vert _{2}^{2}}{\left\Vert \HH^{1/2}\HHt^{-1}\vr\right\Vert _{\HH^{-1/2}\HHt\HH^{-1/2}}^{2}}\leq\sqrt{\beta}
\]
or
\[
\frac{\left\Vert \HHt^{-1/2}\HH\HHt^{-1/2}\cdot\vrh\right\Vert _{2}^{2}}{\left\Vert \vrh\right\Vert _{\HHt^{-1/2}\HH\HHt^{-1/2}}^{2}}=\frac{\left\Vert \HH^{1/2}\HHt^{-1}\vr\right\Vert _{\HH^{1/2}\HHt^{-1}\HH^{1/2}}^{2}}{\left\Vert \HH^{1/2}\HHt^{-1}\vr\right\Vert _{2}^{2}}\geq\frac{1}{\sqrt{\beta}}\,,
\]
which concludes the proof.
\end{proof}

\subsection{Proof of Lemma \ref{lem:precon-update-1}\label{subsec:Proof-of-Lemma-precon-update-1}}
\begin{proof}
We analyze each of the two types of certificates. 

\paragraph{Type 1. }

Letting $\XX=\HH^{-1/2}\HHt\HH^{-1/2}$ and $\vu=\HH^{1/2}\HHt^{-1}\vr$,
we have $\vu^{\top}\XX\vu\geq\frac{1}{\sqrt{\beta}}\vu^{\top}\vu$.
Using Lemma \ref{lem:excent-upd} we obtain that changing the preconditioner
to $\HHtp$ such that
\begin{align*}
\HH^{-1/2}\HHtp\HH^{-1/2} & =\HH^{-1/2}\HHt\HH^{-1/2}-\frac{\HH^{-1/2}\HHt\HH^{-1/2}\left(\HH^{1/2}\HHt^{-1}\vr\right)\left(\HH^{1/2}\HHt^{-1}\vr\right)^{\top}\HH^{-1/2}\HHt\HH^{-1/2}}{\left\Vert \HH^{1/2}\HHt^{-1}\vr\right\Vert _{2}^{2}+\left\Vert \HH^{1/2}\HHt^{-1}\vr\right\Vert _{\HH^{-1/2}\HHt\HH^{-1/2}}^{2}}\\
 & =\HH^{-1/2}\HHt\HH^{-1/2}-\frac{\HH^{-1/2}\vr\vr^{\top}\HH^{-1/2}}{\left\Vert \HHt^{-1}\vr\right\Vert _{\HH}^{2}+\left\Vert \vr\right\Vert _{\HHt^{-1}}^{2}}\,,
\end{align*}
and equivalently setting
\[
\HHtp=\HHt-\frac{\vr\vr^{\top}}{\left\Vert \HHt^{-1}\vr\right\Vert _{\HH}^{2}+\left\Vert \vr\right\Vert _{\HHt^{-1}}^{2}}\,,
\]
we obtain
\[
\exc\left(\HHtp\HH^{-1}\right)\leq\exc\left(\HHt\HH^{-1}\right)\cdot\frac{2}{\sqrt{1+\frac{1}{\sqrt{\beta}}}}\,.
\]

\paragraph{Type 2. }

Letting $\XX=\HH^{-1/2}\HHt\HH^{-1/2}$ and $\vu=\HH^{1/2}\HHt^{-1}\vr$,
we have $\vu^{\top}\XX^{-1}\vu\geq\frac{1}{\sqrt{\beta}}\vu^{\top}\vu$.
Using Lemma \ref{lem:excent-upd} we obtain that changing the preconditioner
to $\HHtp$ such that
\[
\HH^{-1/2}\HHtp\HH^{-1/2}=\HH^{-1/2}\HHt\HH^{-1/2}+\frac{\left(\HH^{1/2}\HHt^{-1}\vr\right)\left(\HH^{1/2}\HHt^{-1}\vr\right)^{\top}}{\left\Vert \HH^{1/2}\HHt^{-1}\vr\right\Vert _{2}^{2}}\,,
\]
and equivalently setting
\[
\HHtp=\HHt+\frac{\HH\HHt^{-1}\vr\vr^{\top}\HHt^{-1}\HH}{\left\Vert \HHt^{-1}\vr\right\Vert _{\HH}^{2}}\,,
\]
we obtain
\[
\exc\left(\HHtp\HH^{-1}\right)\leq\exc\left(\HHt\HH^{-1}\right)\cdot\frac{2}{\sqrt{1+\frac{1}{\sqrt{\beta}}}}\,.
\]
We describe the preconditioner updates in Algorithm \ref{alg:precon-update-linear}.
The updates to the inverse preconditioner $\HHt^{-1}$ follow from
applying the Sherman-Morrison formula (Lemma \ref{lem:shermanmorrison}).

\paragraph{Change in preconditioner norm.}

Now, we verify that the updates do not increase the norms of $\HHtp$
and $\HHtp^{-1}$ by too much compared to those of $\HHt$ and $\HHt^{-1}$,
respectively. Whenever the preconditioner norm increases, it is because
of a type 2 certificate, in which case we have:
\[
\HHtp=\HHt+\frac{\HH\HHt^{-1}\vr\vr^{\top}\HHt^{-1}\HH}{\left\Vert \HHt^{-1}\vr\right\Vert _{\HH}^{2}}=\HHt^{1/2}\left(\Id+\frac{\HHt^{-1/2}\HH\HHt^{-1}\vr\vr^{\top}\HHt^{-1}\HH\HHt^{-1/2}}{\left\Vert \HHt^{-1}\vr\right\Vert _{\HH}^{2}}\right)\HHt^{1/2}\,,
\]
and therefore 
\[
\left\Vert \HHt^{-1/2}\HHtp\HHt^{-1/2}\right\Vert \leq1+\frac{\left\Vert \HHt^{-1/2}\HH\HHt^{-1}\vr\vr^{\top}\HHt^{-1}\HH\HHt^{-1/2}\right\Vert }{\left\Vert \HHt^{-1}\vr\right\Vert _{\HH}^{2}}=1+\frac{\left\Vert \HH^{1/2}\HHt^{-1}\vr\right\Vert _{\HH^{1/2}\HHt^{-1}\HH^{1/2}}^{2}}{\left\Vert \HH^{1/2}\HHt^{-1}\vr\right\Vert _{\HH}^{2}}\,.
\]
At this point we use the fact that the decrease in excentricity is
determined by the magnitude of the eigenvalue proved by the certificate.
In fact, setting $\beta$ to match exactly the ratio $\beta=\vu^{\top}\XX^{-1}\vu/\vu^{\top}\vu$,
we have: 
\[
\exc\left(\HHtp\HH^{-1}\right)\leq\exc\left(\HHt\HH^{-1}\right)\cdot\frac{2}{\sqrt{1+\frac{\left\Vert \HH^{1/2}\HHt^{-1}\vr\right\Vert _{\HH^{1/2}\HHt^{-1}\HH^{1/2}}^{2}}{\left\Vert \HH^{1/2}\HHt^{-1}\vr\right\Vert _{\HH}^{2}}}}\,,
\]
and combining with the previous inequality we can now upper bound:

\[
\left\Vert \HHt^{-1/2}\HHtp\HHt^{-1/2}\right\Vert \leq\left(2\cdot\frac{\exc\left(\HHt\HH^{-1}\right)}{\exc\left(\HHtp\HH^{-1}\right)}\right)^{2}\,.
\]
The other case is similar. Whenever the inverse preconditioner norm
increases, it is because of a type 1 certificate, in which case we
have:

\[
\HHtp^{-1}=\HHt^{-1}+\frac{\HHt^{-1}\vr\vr^{\top}\HHt^{-1}}{\left\Vert \HHt^{-1}\vr\right\Vert _{\HH}^{2}}=\HHt^{-1/2}\left(\Id+\frac{\HHt^{-1/2}\vr\vr^{\top}\HHt^{-1/2}}{\left\Vert \HHt^{-1}\vr\right\Vert _{\HH}^{2}}\right)\HHt^{-1/2}\,,
\]
so
\[
\left\Vert \HHt^{1/2}\HHtp^{-1}\HHt^{1/2}\right\Vert \leq1+\frac{\left\Vert \HHt^{-1/2}\vr\right\Vert _{2}^{2}}{\left\Vert \HHt^{-1}\vr\right\Vert _{\HH}^{2}}\,.
\]
Again we use the fact that the decrease in excentricity is determined
by the magnitude of the eigenvalue proved by the certificate. Setting
$\beta$ to match exactly the ratio $\beta=\vu^{\top}\XX\vu/\vu^{\top}\vu$,
we have: 
\[
\exc\left(\HHtp\HH^{-1}\right)\leq\exc\left(\HHt\HH^{-1}\right)\cdot\frac{2}{\sqrt{1+\frac{\left\Vert \HH^{1/2}\HHt^{-1}\vr\right\Vert _{\HH^{-1/2}\HHt\HH^{-1/2}}^{2}}{\left\Vert \HH^{1/2}\HHt^{-1}\vr\right\Vert _{2}^{2}}}}=\exc\left(\HHt\HH^{-1}\right)\cdot\frac{2}{\sqrt{1+\frac{\left\Vert \HHt^{-1/2}\vr\right\Vert _{2}^{2}}{\left\Vert \HHt^{-1}\vr\right\Vert _{\HH}^{2}}}}\,,
\]
and combining with the previous inequality we upper bound:

\[
\left\Vert \HHt^{1/2}\HHtp^{-1}\HHt^{1/2}\right\Vert \leq\left(2\cdot\frac{\exc\left(\HHt\HH^{-1}\right)}{\exc\left(\HHtp\HH^{-1}\right)}\right)^{2}\,.
\]
\end{proof}

\subsection{Proof of Lemma \ref{lem:main-linear-system}\label{subsec:Proof-of-Lemma-main-linear-system}}
\begin{proof}
Let $\vr_{t}=\vb-\HH\vx_{t}$. We note that for any vector $\vv$,
and any $\HHt$, $\HHtp$,
\[
\left\Vert \vv\right\Vert _{\HHt^{-1}}^{2}=\left\Vert \HHt^{-1/2}\vv\right\Vert _{2}^{2}=\left\Vert \HHt^{-1/2}\HHtp^{1/2}\cdot\HHtp^{-1/2}\vv\right\Vert _{2}^{2}\leq\left\Vert \HHtp^{1/2}\HHt^{-1}\HHtp^{1/2}\right\Vert \cdot\left\Vert \vv\right\Vert _{\HHtp^{-1}}^{2}\,.
\]
Using this inequality, we can write:
\begin{align*}
\left\Vert \vr_{T}\right\Vert ^{2} & \leq\left(\prod_{t=0}^{T-1}\left\Vert \HHt_{t+1}^{1/2}\HHt_{t}^{-1}\HHt_{t+1}^{1/2}\right\Vert \right)\left\Vert \HHt_{T}^{-1/2}\vr_{T}\right\Vert _{2}^{2}\\
 & =\left(\prod_{t=0}^{T-1}\left\Vert \HHt_{t+1}^{1/2}\HHt_{t}^{-1}\HHt_{t+1}^{1/2}\right\Vert \right)\left(\prod_{t=0}^{t-1}\frac{\left\Vert \HHt_{t+1}^{-1/2}\vr_{t+1}\right\Vert ^{2}}{\left\Vert \HHt_{t}^{-1/2}\vr_{t}\right\Vert ^{2}}\right)\left\Vert \HHt_{0}^{-1/2}\vr_{0}\right\Vert _{2}^{2}\\
 & \leq\prod_{t=0}^{T-1}\left(\left\Vert \HHt_{t+1}^{1/2}\HHt_{t}^{-1}\HHt_{t+1}^{1/2}\right\Vert \cdot\frac{\left\Vert \vr_{t+1}\right\Vert _{\HHt_{t+1}^{-1}}^{2}}{\left\Vert \vr_{t}\right\Vert _{\HHt_{t}^{-1}}^{2}}\right)\cdot\left\Vert \vr_{0}\right\Vert _{\HHt_{0}^{-1}}^{2}\\
 & =\Phi_{T}\cdot\left\Vert \vr_{0}\right\Vert _{\HHt_{0}^{-1}}^{2}\,,
\end{align*}
where we use the potential function $\Phi$ defined as
\begin{align*}
\Phi_{t} & =\begin{cases}
1\,, & \text{if }t=0\,,\\
\prod_{t=0}^{t-1}\left(\left\Vert \HHt_{t+1}^{1/2}\HHt_{t}^{-1}\HHt_{t+1}^{1/2}\right\Vert \cdot\frac{\left\Vert \vr_{t+1}\right\Vert _{\HHt_{t+1}^{-1}}^{2}}{\left\Vert \vr_{t}\right\Vert _{\HHt_{t}^{-1}}^{2}}\right)\,, & \text{if }t\geq1\,.
\end{cases}
\end{align*}
Now we see how $\Phi_{t}$ evolves. There are two cases. In the former,
$\left\Vert \vr_{t+1}\right\Vert _{\HHt_{t+1}^{-1}}^{2}\leq\left(1-\beta\right)\left\Vert \vr_{t}\right\Vert _{\HHt_{t}^{-1}}^{2}$
and thus the preconditioner stays unchanged, so $\Phi_{t}\leq\left(1-\beta\right)\Phi_{t-1}$.
In the latter we perform a preconditioner update, but keep the iterate
unchanged $\vr_{t+1}=\vr_{t}$. Based on Lemma \ref{lem:precon-update-1},
there are two possibilities.

If $1\leq\left\Vert \HHt_{t}^{-1/2}\HHt_{t+1}\HHt_{t}^{-1/2}\right\Vert \leq\left(2\cdot\frac{\exc\left(\HHt\HH^{-1}\right)}{\exc\left(\HHt_{t+1}\HH^{-1}\right)}\right)^{2}$,
then $\left\Vert \vr_{t+1}\right\Vert _{\HHt_{t+1}^{-1}}\leq\left\Vert \vr_{t+1}\right\Vert _{\HHt_{t}^{-1}}$,
but $\left\Vert \HHt_{t+1}^{1/2}\HHt_{t}^{-1}\HHt_{t+1}^{1/2}\right\Vert =\left\Vert \HHt_{t}^{-1/2}\HHt_{t+1}\HHt_{t}^{-1/2}\right\Vert $
and therefore
\[
\Phi_{t}\leq\Phi_{t-1}\cdot\left(2\cdot\frac{\exc\left(\HHt\HH^{-1}\right)}{\exc\left(\HHt_{t+1}\HH^{-1}\right)}\right)^{2}\,.
\]

If $1\leq\left\Vert \HHt_{t}^{1/2}\HHt_{t+1}^{-1}\HHt_{t}^{1/2}\right\Vert \leq\left(2\cdot\frac{\exc\left(\HHt\HH^{-1}\right)}{\exc\left(\HHt_{t+1}\HH^{-1}\right)}\right)^{2}$,
then $\left\Vert \HHt_{t+1}^{1/2}\HHt_{t}^{-1}\HHt_{t+1}^{1/2}\right\Vert =\frac{1}{\left\Vert \HHt_{t}^{1/2}\HHt_{t+1}^{-1}\HHt_{t}^{1/2}\right\Vert }\leq1$,
but $\frac{\left\Vert \vr_{t+1}\right\Vert _{\HHt_{t+1}^{-1}}^{2}}{\left\Vert \vr_{t}\right\Vert _{\HHt_{t}^{-1}}^{2}}=\frac{\left\Vert \HHt_{t}^{-1/2}\vr_{t}\right\Vert _{\HHt_{t}^{1/2}\HHt_{t+1}^{-1}\HHt_{t}^{1/2}}^{2}}{\left\Vert \HHt_{t}^{-1/2}\vr_{t}\right\Vert _{2}^{2}}\leq\left\Vert \HHt_{t}^{1/2}\HHt_{t+1}^{-1}\HHt_{t}^{1/2}\right\Vert $,
and therefore
\[
\Phi_{t}\leq\Phi_{t-1}\cdot\left(2\cdot\frac{\exc\left(\HHt\HH^{-1}\right)}{\exc\left(\HHt_{t+1}\HH^{-1}\right)}\right)^{2}\,.
\]
Finally since whenever the preconditioner changes, excentricity gets
reduced in the worst case by a factor of $\frac{2}{\sqrt{1+\frac{1}{\sqrt{\beta}}}}\geq\beta^{1/4}$
(Lemma \ref{lem:Richardson-progress-certificate}), this can happen
at most $T_{\text{prec}}=\frac{4\ln\mathcal{E}\left(\XX_{0}\right)}{\ln\left(1/\beta\right)}$
times, since $\mathcal{E}\left(\XX\right)\geq1$ at all times. Therefore
letting 
\[
T=T_{\text{prec}}+\ln_{1-\beta}\epsilon=\frac{4\ln\mathcal{E}\left(\XX_{0}\right)}{\ln\left(\frac{1}{\beta}\right)}+\frac{\ln\left(1/\epsilon\right)}{\ln\left(1/\left(1-\beta\right)\right)}
\]
we are guaranteed that 
\[
\Phi_{T}\leq\Phi_{0}\left(1-\beta\right)^{T-T_{\text{prec}}}=\Phi_{0}\cdot\epsilon=\epsilon\,,
\]
which yields $\left\Vert \vr_{T}\right\Vert _{2}^{2}\leq\epsilon\left\Vert \vr_{0}\right\Vert _{2}^{2}$.
Finally, setting $\beta=\frac{1}{2}$, we obtain $T=O\left(\ln\mathcal{E}\left(\XX_{0}\right)+\ln\frac{1}{\epsilon}\right)$,
which concludes the proof.
\end{proof}

\section{Proofs from Section \ref{sec:nonlinear}}

\subsection{Robustness Proofs\label{subsec:Robustness-Proofs}}
\begin{lem}
\label{lem:excent-upd-robust}Let a vector $\vu$ and an over-estimate
for its norm $n\left(\vu\right)$ such that $\left\Vert \vu\right\Vert \leq n\left(\vu\right)\leq\left\Vert \vu\right\Vert \cdot\alpha$
for some $\alpha\geq1$. If $\vu$ satisfies
\begin{enumerate}
\item $\vu^{\top}\XX^{-1}\vu\geq\gamma\cdot\vu^{\top}\vu$, then $\frac{\exc\left(\XX+\frac{\vu\vu^{\top}}{n\left(\vu\right)^{2}}\right)}{\exc\left(\XX\right)}\leq\frac{2}{\sqrt{1+\frac{\gamma}{\alpha^{2}}}}$,
\item $\vu^{\top}\XX\vu\geq\gamma\cdot\vu^{\top}\vu$, then $\frac{\exc\left(\XX-\frac{\XX\vu\vu^{\top}\XX}{n\left(\vu\right)^{2}+\vu^{\top}\XX\vu}\right)}{\exc\left(\XX\right)}\leq\frac{2}{\sqrt{1+\frac{\gamma}{\alpha^{2}}}}$.
\end{enumerate}
\end{lem}

\begin{proof}
We use a similar proof to the one for Lemma \ref{lem:excent-upd}.
In the first case, using Lemma \ref{lem:excentricity-update}, we
obtain
\[
\frac{\exc\left(\XX+\frac{\vu\vu^{\top}}{n\left(\vu\right)^{2}}\right)}{\exc\left(\XX\right)}=\frac{1+\frac{\vu^{\top}\left(\Id+\XX\right)^{-1}\vu}{n\left(\vu\right)^{2}}}{\sqrt{1+\frac{\vu^{\top}\XX^{-1}\vu}{n\left(\vu\right)^{2}}}}\leq\frac{2}{\sqrt{1+\frac{\vu^{\top}\XX^{-1}\vu}{n\left(\vu\right)^{2}}}}\leq\frac{2}{\sqrt{1+\frac{\vu^{\top}\XX^{-1}\vu}{\alpha^{2}\left\Vert \vu\right\Vert ^{2}}}}\leq\frac{2}{\sqrt{1+\frac{\gamma}{\alpha^{2}}}}\ .
\]
For the second case we use the fact that excentricity is invariant
under inversion, and hence
\begin{align*}
\frac{\exc\left(\XX-\frac{\XX\vu\vu^{\top}\XX}{n\left(\vu\right)^{2}+\vu^{\top}\XX\vu}\right)}{\exc\left(\XX\right)}&=
\frac{\exc\left(\XX-\frac{\frac{\XX\vu\vu^{\top}\XX}{n\left(\vu\right)^{2}}}{1+\frac{\vu^{\top}\XX\vu}{n\left(\vu\right)^{2}}}\right)}{\exc\left(\XX\right)}=\frac{\exc\left(\left(\XX^{-1}+\frac{\vu\vu^{\top}}{n\left(\vu\right)^{2}}\right)^{-1}\right)}{\exc\left(\XX\right)}=\frac{\exc\left(\XX^{-1}+\frac{\vu\vu^{\top}}{n\left(\vu\right)^{2}}\right)}{\exc\left(\XX^{-1}\right)}
\\
&\leq\frac{2}{\sqrt{1+\frac{\gamma}{\alpha^{2}}}}\ .
\end{align*}
\end{proof}

\subsection{Lemmas for Estimating Hessian-Vector Products\label{subsec:hvp-estimation}}

The following two estimation lemmas use properties of self-concordant
functions.
\begin{lem}
\label{lem:hessian_norm_estimation}Let $g:K\rightarrow\mathbb{R}$
be a self-concordant function, let $\vy\in\text{int}\left(K\right)$
and let $\HH_{\vy}:=\nabla^{2}g\left(\vy\right)$ be the Hessian at
$\vy$, which satisfies $\left\Vert \HH_{\vy}\right\Vert \leq B$.
Then for any vector $\vv\in\mathbb{R}^{n}$, using two calls to a
gradient oracle for $g$ we can obtain an estimate 
\begin{align*}
n_{\vy}\left(\vv\right) & =\left(\frac{1}{1-\frac{1}{1000}}\right)\cdot\sqrt{\frac{1}{\tau}\cdot\left\langle \vv,\nabla g\left(\vy+\tau\vv\right)-\nabla g\left(\vy\right)\right\rangle }\ ,\text{where}\\
\tau & =\frac{1}{1000\left\Vert \vv\right\Vert B}\ ,
\end{align*}
 such that
\[
\left\Vert \vv\right\Vert _{\HH_{\vy}}\leq n_{\vy}\left(\vv\right)\leq\left\Vert \vv\right\Vert _{\HH_{\vy}}\cdot\left(\frac{1}{1-\frac{1}{1000}}\right)^{2}\,.
\]
\end{lem}

\begin{proof}
We proceed to analyze the quality of the approximation. Since 
\[
\nabla g\left(\vy+\tau\vv\right)-\nabla g\left(\vy\right)=\left(\int_{0}^{1}\nabla^{2}g\left(\vy+t\cdot\tau\vv\right)dt\right)\cdot\left(\tau\vv\right)\,,
\]
we have
\begin{align*}
n_{\vy}\left(\vv\right) & =\frac{1}{1-\frac{1}{1000}}\cdot\sqrt{\frac{1}{\tau}\cdot\left\langle \vv,\left(\int_{0}^{1}\nabla^{2}g\left(\vy+t\cdot\tau\vv\right)dt\right)\tau\vv\right\rangle }\\
 & =\frac{1}{1-\frac{1}{1000}}\cdot\left\langle \vv,\left(\int_{0}^{1}\nabla^{2}g\left(\vy+t\cdot\tau\vv\right)dt\right)\vv\right\rangle ^{1/2}\,.
\end{align*}
Using the bound 
\[
\left\Vert \tau\vv\right\Vert _{\HH_{\vy}}\leq\tau\cdot\left\Vert \vv\right\Vert \cdot\left\Vert \HH_{\vy}\right\Vert \leq\tau\cdot\left\Vert \vv\right\Vert \cdot B=\frac{1}{1000}\,,
\]
we apply self-concordance to sandwich
\[
\left(1-\left\Vert \tau\vv\right\Vert _{\HH_{\vy}}\right)^{2}\cdot\HH_{\vy}\preceq\nabla^{2}g\left(\vy+t\cdot\tau\vv\right)\preceq\left(\frac{1}{1-\left\Vert \tau\vv\right\Vert _{\HH_{\vy}}}\right)^{2}\cdot\HH_{\vy}\,,
\]
which also implies that our query point is feasible, since the Hessian
stays bounded throughout the entire path between $\vy$ and $\vy+\tau\vv$.
\[
\left(1-\frac{1}{1000}\right)^{2}\cdot\HH_{\vy}\preceq\nabla^{2}g\left(\vy+t\cdot\tau\vv\right)\preceq\left(\frac{1}{1-\frac{1}{1000}}\right)^{2}\cdot\HH_{\vy}\,.
\]
 Plugging this back into our identity for $n_{\vy}\left(\vv\right)$
we have that
\[
\frac{1}{1-\frac{1}{1000}}\cdot\left(1-\frac{1}{1000}\right)\cdot\left\Vert \vv\right\Vert _{\HH_{\vy}}\leq n_{\vy}\left(\vv\right)\leq\left(\frac{1}{1-\frac{1}{1000}}\right)^{2}\cdot\left\Vert \vv\right\Vert _{\HH_{\vy}}\,,
\]
which concludes the proof.
\end{proof}
\begin{lem}
\label{lem:hessian-vector-product-estimation}Let $g:\mathbb{R}^{n}\rightarrow\mathbb{R}$
be a self-concordant function, let $\vy\in\mathbb{R}^{n}$ and let
$\HH_{\vy}:=\nabla^{2}g\left(\vy\right)$ be the Hessian at $\vy$,
which satisfies $\left\Vert \HH_{\vy}\right\Vert \leq B$, for $B\geq1$.
Then for any vector $\vv\in\mathbb{R}^{n}$, using two calls to a
gradient oracle for $g$ we can obtain an estimate for the Hessian
vector product:
\begin{align*}
p_{\vy}\left(\vv\right) & =\frac{1}{\tau}\left(\nabla g\left(\vy+\tau\vv\right)-\nabla g\left(\vy\right)\right)\ ,\text{where}\\
\tau & =\frac{1}{1000\left\Vert \vv\right\Vert B^{21}}\ ,
\end{align*}
 such that
\[
\left\Vert p_{\vy}\left(\vv\right)-\HH_{\vy}\vv\right\Vert _{\HH_{\vy}^{-1}}\leq\frac{1}{400\cdot B^{20}}\left\Vert \vv\right\Vert _{\HH_{\vy}}\ .
\]
\end{lem}

\begin{proof}
We proceed to analyze the error. For the same reason as in the proof
of Lemma \ref{lem:hessian_norm_estimation}, the query point is feasible.
Letting $\HHb_{\tau}=\int_{0}^{1}\HH_{\vy+t\tau\vv}dt$, we can write
$p_{\vy}\left(\vv\right)=\frac{1}{\tau}\HHb_{\tau}\cdot\left(\tau\vv\right)=\HHb_{\tau}\vv$.
Therefore using self-concordance we bound:
\begin{align*}
\left\Vert p_{\vy}\left(\vv\right)-\HH_{\vy}\vv\right\Vert _{\HH_{\vy}^{-1}} & =\left\Vert \HH_{\vy}^{-1/2}\left(\HHb_{\tau}-\HH_{\vy}\right)\HH_{\vy}^{-1/2}\cdot\HH_{\vy}^{1/2}\vv\right\Vert _{2}\leq\left\Vert \HH_{\vy}^{-1/2}\left(\HHb_{\tau}-\HH_{\vy}\right)\HH_{\vy}^{-1/2}\right\Vert \left\Vert \vv\right\Vert _{\HH_{\vy}}\\
 & \leq\left(\left(\frac{1}{1-\tau\left\Vert \vv\right\Vert _{\HH_{\vy}}}\right)^{2}-1\right)\left\Vert \vv\right\Vert _{\HH_{\vy}}\leq\left(\left(\frac{1}{1-\tau\cdot\left\Vert \HH_{\vy}\right\Vert \left\Vert \vv\right\Vert }\right)^{2}-1\right)\left\Vert \vv\right\Vert _{\HH_{\vy}}\\
 & \leq\left(\left(\frac{1}{1-\frac{1}{1000\cdot B^{20}}}\right)^{2}-1\right)\left\Vert \vv\right\Vert _{\HH_{\vy}}\leq\frac{1}{400\cdot B^{20}}\left\Vert \vv\right\Vert _{\HH_{\vy}}\,.
\end{align*}
\end{proof}

\subsection{Proof of Lemma \ref{lem:nonlinear-progress-certificate}\label{subsec:Proof-of-Lemma-nonlinear-progress-certificate}}
\begin{proof}
We perform the appropriate error analysis by comparing the residual
obtained using the estimators $p_{\vy}$ and $n_{\vy}$ to the one
we would have obtained had we used exact access to the matrix $\HH_{\vy}$
instead. First we use the fact that
\[
\left\Vert \HHt^{-1}\vr\right\Vert _{\HH_{\vy}}\leq n_{\vy}\left(\HHt^{-1}\vr\right)\leq\left\Vert \HHt^{-1}\vr\right\Vert _{\HH_{\vy}}\cdot\left(\frac{1}{1-\frac{1}{1000}}\right)^{2}
\]
and
\[
\left\Vert p_{\vy}\left(\HHt^{-1}\vr\right)\right\Vert _{\HHt^{-1}}=\left\Vert \HH_{\vy}\HHt^{-1}\vr+\left(\HH_{\vy}\left(\HHt^{-1}\vr\right)-p_{\vy}\left(\HHt^{-1}\vr\right)\right)\right\Vert _{\HHt^{-1}}\,,
\]
We bound 
\begin{align}
&\left\Vert \HH_{\vy}\left(\HHt^{-1}\vr\right)-p_{\vy}\left(\HHt^{-1}\vr\right)\right\Vert _{\HHt^{-1}} 
\\
& =\left\Vert \HHt^{-1/2}\HH_{\vy}^{1/2}\cdot\HH_{\vy}^{-1/2}\left(\HH_{\vy}\left(\HHt^{-1}\vr\right)-p_{\vy}\left(\HHt^{-1}\vr\right)\right)\right\Vert _{2}\nonumber \\
 & \leq\left\Vert \HHt^{-1/2}\HH_{\vy}^{1/2}\right\Vert \left\Vert \HH_{\vy}\left(\HHt^{-1}\vr\right)-p_{\vy}\left(\HHt^{-1}\vr\right)\right\Vert _{\HH_{\vy}^{-1}}\nonumber \\
 & \leq\left\Vert \HHt^{-1/2}\HH_{\vy}^{1/2}\right\Vert \cdot\frac{1}{400\cdot B^{20}}\cdot\left\Vert \HHt^{-1}\vr\right\Vert _{\HH_{\vy}}\nonumber \\
 & \leq\left\Vert \HHt^{-1/2}\HH_{\vy}^{1/2}\right\Vert \cdot\frac{1}{400\cdot B^{20}}\cdot\left\Vert \HH_{\vy}^{-1/2}\HHt^{1/2}\right\Vert \left\Vert \HHt^{-1/2}\HH_{\vy}^{1/2}\cdot\HH_{\vy}^{1/2}\HHt^{-1}\vr\right\Vert _{2}\nonumber \\
 & =\left\Vert \HHt^{-1/2}\HH_{\vy}^{1/2}\right\Vert \left\Vert \HH_{\vy}^{-1/2}\HHt^{1/2}\right\Vert \cdot\frac{1}{400\cdot B^{20}}\cdot\left\Vert \HH_{\vy}\HHt^{-1}\vr\right\Vert _{\HHt^{-1}}\nonumber \\
 & \leq\frac{1}{400\cdot B^{18}}\cdot\left\Vert \HH_{\vy}\HHt^{-1}\vr\right\Vert _{\HHt^{-1}}\,,\label{eq:e1}
\end{align}
which gives us that
\begin{equation}
\left\Vert \HH_{\vy}\HHt^{-1}\vr\right\Vert _{\HHt^{-1}}\cdot\left(1-\frac{1}{400B^{18}}\right)\leq\left\Vert p_{\vy}\left(\HHt^{-1}\vr\right)\right\Vert _{\HHt^{-1}}\leq\left\Vert \HH_{\vy}\HHt^{-1}\vr\right\Vert _{\HHt^{-1}}\cdot\left(1+\frac{1}{400B^{18}}\right)\,.\label{eq:e2}
\end{equation}
Thus we can sandwich
\begin{equation}
\frac{\left\Vert \HHt^{-1}\vr\right\Vert _{\HH_{\vy}}}{\left\Vert \HH_{\vy}\HHt^{-1}\vr\right\Vert _{\HHt^{-1}}}\cdot\frac{1}{\left(1+\frac{1}{400B^{18}}\right)}\leq\frac{n_{\vy}\left(\HHt^{-1}\vr\right)}{\left\Vert p_{\vy}\left(\HHt^{-1}\vr\right)\right\Vert _{\HHt^{-1}}}\leq\frac{\left\Vert \HHt^{-1}\vr\right\Vert _{\HH_{\vy}}}{\left\Vert \HH_{\vy}\HHt^{-1}\vr\right\Vert _{\HHt^{-1}}}\cdot\frac{\left(\frac{1}{1-\frac{1}{1000}}\right)^{2}}{\left(1-\frac{1}{400B^{18}}\right)}\,.\label{eq:e3b}
\end{equation}
Additionally, we can measure the residual error. Given an arbitrary
vector $\vv$, we have
\begin{align*}
\left\Vert \vb-p_{\vy}\left(\vv\right)\right\Vert _{\HHt^{-1}} & =\left\Vert \vb-\HH_{\vy}\vv+\HH_{\vy}\vv-p_{\vy}\left(\vv\right)\right\Vert _{\HHt^{-1}}\,,
\end{align*}
and we bound
\begin{align}
\left\Vert \HH_{\vy}\vv-p_{\vy}\left(\vv\right)\right\Vert _{\HHt^{-1}} & =\left\Vert \HHt^{-1/2}\HH_{\vy}^{1/2}\cdot\HH_{\vy}^{-1/2}\left(\HH_{\vy}\vv-p_{\vy}\left(\vv\right)\right)\right\Vert _{2}\nonumber \\
 & \leq\left\Vert \HHt^{-1/2}\HH_{\vy}^{1/2}\right\Vert \cdot\left\Vert \HH_{\vy}\vv-p_{\vy}\left(\vv\right)\right\Vert _{\HH_{\vy}^{-1}}\nonumber \\
 & \leq\left\Vert \HHt^{-1/2}\HH_{\vy}^{1/2}\right\Vert \cdot\frac{1}{400\cdot B^{20}}\cdot\left\Vert \vv\right\Vert _{\HH_{\vy}}\nonumber \\
 & \leq\left\Vert \HHt^{-1/2}\HH_{\vy}^{1/2}\right\Vert \cdot\frac{1}{400\cdot B^{20}}\cdot\left\Vert \HHt^{1/2}\HH_{\vy}^{-1/2}\right\Vert \left\Vert \HHt^{-1/2}\HH_{\vy}^{1/2}\cdot\HH_{\vy}^{1/2}\vv\right\Vert _{2}\nonumber \\
 & =\left\Vert \HHt^{-1/2}\HH_{\vy}^{1/2}\right\Vert \left\Vert \HHt^{1/2}\HH_{\vy}^{-1/2}\right\Vert \cdot\frac{1}{400\cdot B^{20}}\cdot\left\Vert \HH_{\vy}\vv\right\Vert _{\HHt^{-1}}\nonumber \\
 & \leq\frac{1}{400\cdot B^{18}}\cdot\left\Vert \HH_{\vy}\vv\right\Vert _{\HHt^{-1}}\nonumber \\
 & \leq\frac{1}{400\cdot B^{18}}\cdot\left\Vert \vb-\HH_{\vy}\vv\right\Vert _{\HHt^{-1}}+\frac{1}{400\cdot B^{18}}\left\Vert \vb\right\Vert _{\HHt^{-1}}\nonumber \\
 & \leq\frac{1}{400\cdot B^{18}}\cdot\left\Vert \vb-\HH_{\vy}\vv\right\Vert _{\HHt^{-1}}+\frac{1}{400\cdot B^{16}}\cdot\,,\label{eq:e3}
\end{align}
from where we conclude that
\begin{equation}
\begin{aligned}
\left\Vert \vb-\HH_{\vy}\vv\right\Vert _{\HHt^{-1}}\left(1-\frac{1}{400\cdot B^{18}}\right)-\frac{1}{400\cdot B^{16}}
\leq
\left\Vert \vb-p_{\vy}\left(\vv\right)\right\Vert _{\HHt^{-1}}
\\
\leq
\left\Vert \vb-\HH_{\vy}\vv\right\Vert _{\HHt^{-1}}\left(1+\frac{1}{400\cdot B^{18}}\right)+\frac{1}{400\cdot B^{16}}\,.\label{eq:residual-approx-imp}
\end{aligned}
\end{equation}
Hence provided that 
\[
\frac{1}{B}\leq\left\Vert \vb-p_{\vy}\left(\vv\right)\right\Vert _{\HHt^{-1}}\,,
\]
as guaranteed in the statement, we additionally have $\frac{1}{400\cdot B^{16}}\leq\frac{1}{400\cdot B^{15}}\cdot\left\Vert \vb-p_{\vy}\left(\vv\right)\right\Vert _{\HHt^{-1}}$,
and thus
\begin{align*}
 & \left\Vert \vb-\HH_{\vy}\vv\right\Vert _{\HHt^{-1}}\left(1-\frac{1}{400\cdot B^{18}}\right)\left(1+\frac{1}{400\cdot B^{15}}\right)^{-1}\\
 & \leq\left\Vert \vb-p_{\vy}\left(\vv\right)\right\Vert _{\HHt^{-1}}\leq\left\Vert \vb-\HH_{\vy}\vv\right\Vert _{\HHt^{-1}}\left(1+\frac{1}{400\cdot B^{18}}\right)\left(1+\frac{1}{400\cdot B^{15}}\right)\,,
\end{align*}
which implies the simpler condition that
\begin{equation}
\left\Vert \vb-\HH_{\vy}\vv\right\Vert _{\HHt^{-1}}\left(1-\frac{1}{200\cdot B^{15}}\right)\leq\left\Vert \vb-p_{\vy}\left(\vv\right)\right\Vert _{\HHt^{-1}}\leq\left\Vert \vb-\HH_{\vy}\vv\right\Vert _{\HHt^{-1}}\left(1+\frac{1}{200\cdot B^{15}}\right)\,.\label{eq:e4}
\end{equation}
Performing a similar analysis we can show that
\begin{align}
 & \left\Vert \HH_{\vy}\HHt^{-1}\left(\vb-\HH_{\vy}\vv\right)\right\Vert _{\HHt^{-1}}\left(1-\frac{1}{200\cdot B^{9}}\right)\nonumber \\
 & \leq\left\Vert \HH_{\vy}\HHt^{-1}\left(\vb-p_{\vy}\left(\vv\right)\right)\right\Vert _{\HHt^{-1}}\leq\left\Vert \HH_{\vy}\HHt^{-1}\left(\vb-\HH_{\vy}\vv\right)\right\Vert _{\HHt^{-1}}\left(1+\frac{1}{200\cdot B^{9}}\right)\,,\label{eq:e5}
 \\
&\left\Vert \HHt^{-1}\left(\vb-\HH_{\vy}\vv\right)\right\Vert _{\HH_{\vy}}\left(1-\frac{1}{200\cdot B^{9}}\right) \nonumber
 \\
 &\leq
 \left\Vert \HHt^{-1}\left(\vb-p_{\vy}\left(\vv\right)\right)\right\Vert _{\HH_{\vy}}
 \leq
 \left\Vert \HHt^{-1}\left(\vb-\HH_{\vy}\vv\right)\right\Vert _{\HH_{\vy}}\left(1+\frac{1}{200\cdot B^{9}}\right)\,.\label{eq:e6}
\end{align}
Now we consider the new residual $\left\Vert \vrp\right\Vert _{\HHt^{-1}}^{2}$.
Provided that $\left\Vert \vrp\right\Vert _{\HHt^{-1}}^{2}\geq\frac{1}{B}$,
as specified in the hypothesis, we can bound it as
\begin{align*}
 & \left\Vert \vrp\right\Vert _{\HHt^{-1}}^{2}=\left\Vert \vb-p_{\vy}\left(\vxp\right)\right\Vert _{\HHt^{-1}}^{2}\leq\left\Vert \vb-\HH_{\vy}\vxp\right\Vert _{\HHt^{-1}}^{2}\left(1+\frac{1}{200\cdot B^{15}}\right)^{2}\\
 & =\left\Vert \vb-\HH_{\vy}\vx-\HH_{\vy}\left(\frac{n_{\vy}\left(\HHt^{-1}\vr\right)^{2}}{\left\Vert p_{\vy}\left(\HHt^{-1}\vr\right)\right\Vert _{\HHt^{-1}}^{2}}\HHt^{-1}\left(\vb-p_{\vy}\left(\vx\right)\right)\right)\right\Vert _{\HHt^{-1}}^{2}\cdot\left(1+\frac{1}{200\cdot B^{15}}\right)^{2}\\
 & =\bigg(\left\Vert \vb-\HH_{\vy}\vx\right\Vert _{\HHt^{-1}}^{2}+\frac{n_{\vy}\left(\HHt^{-1}\vr\right)^{4}}{\left\Vert p_{\vy}\left(\HHt^{-1}\vr\right)\right\Vert _{\HHt^{-1}}^{4}}\cdot\left\Vert \HH_{\vy}\HHt^{-1}\left(\vb-p_{\vy}\left(\vx\right)\right)\right\Vert _{\HHt^{-1}}^{2}\\
 & -2\cdot\frac{n_{\vy}\left(\HHt^{-1}\vr\right)^{2}}{\left\Vert p_{\vy}\left(\HHt^{-1}\vr\right)\right\Vert _{\HHt^{-1}}^{2}}\left\langle \HHt^{-1/2}\left(\vb-\HH_{\vy}\vx\right),\HHt^{-1/2}\HH_{\vy}\HHt^{-1}\left(\vb-p_{\vy}\left(\vx\right)\right)\right\rangle \bigg)\cdot\left(1+\frac{1}{200\cdot B^{15}}\right)^{2}\,.
\end{align*}
We lower bound, using $\left\Vert \vb-p_{\vy}\left(\vx\right)\right\Vert _{\HHt^{-1}}\geq1/B$,
\begin{align}
 & \left\langle \HHt^{-1/2}\left(\vb-\HH_{\vy}\vx\right),\HHt^{-1/2}\HH_{\vy}\HHt^{-1}\left(\vb-p_{\vy}\left(\vx\right)\right)\right\rangle \nonumber \\
 & =\left\Vert \HHt^{-1}\left(\vb-p_{\vy}\left(\vx\right)\right)\right\Vert _{\HH_{\vy}}^{2}-\left\langle \HHt^{-1/2}\left(\vb-p_{\vy}\left(\vx\right)\right),\HHt^{-1/2}\HH_{\vy}\HHt^{-1/2}\cdot\HHt^{-1/2}\left(\HH_{\vy}\vx-p_{\vy}\left(\vx\right)\right)\right\rangle \nonumber \\
 & \geq\left\Vert \HHt^{-1}\left(\vb-p_{\vy}\left(\vx\right)\right)\right\Vert _{\HH_{\vy}}^{2}-\left\Vert \HHt^{-1/2}\HH_{\vy}\HHt^{-1/2}\right\Vert \left\Vert \vb-p_{\vy}\left(\vx\right)\right\Vert _{\HHt^{-1}}\left\Vert p_{\vy}\left(\vx\right)-\HH_{\vy}\vx\right\Vert _{\HHt^{-1}}\nonumber \\
 & \geq\left\Vert \HHt^{-1}\left(\vb-p_{\vy}\left(\vx\right)\right)\right\Vert _{\HH_{\vy}}^{2} \nonumber 
\\
&\quad-\left\Vert \HHt^{-1/2}\HH_{\vy}\HHt^{-1/2}\right\Vert \left\Vert \vb-p_{\vy}\left(\vx\right)\right\Vert _{\HHt^{-1}}\left(\frac{1}{400\cdot B^{18}}\cdot\left\Vert \vb-\HH_{\vy}\vx\right\Vert _{\HHt^{-1}}+\frac{1}{400\cdot B^{16}}\right)\nonumber \\
 & \geq\left\Vert \HHt^{-1}\left(\vb-p_{\vy}\left(\vx\right)\right)\right\Vert _{\HH_{\vy}}^{2}-\frac{1}{400\cdot B^{16}}\left\Vert \vb-p_{\vy}\left(\vx\right)\right\Vert _{\HHt^{-1}}\left\Vert \vb-\HH_{\vy}\vx\right\Vert _{\HHt^{-1}} \nonumber 
 \\
 &\quad-\frac{1}{400\cdot B^{14}}\left\Vert \vb-p_{\vy}\left(\vx\right)\right\Vert _{\HHt^{-1}}\nonumber \\
 & \geq\left\Vert \HHt^{-1}\left(\vb-p_{\vy}\left(\vx\right)\right)\right\Vert _{\HH_{\vy}}^{2}-\frac{1}{200\cdot B^{15}}\left\Vert \vb-p_{\vy}\left(\vx\right)\right\Vert _{\HHt^{-1}}^{2}-\frac{1}{400\cdot B^{14}}\left\Vert \vb-p_{\vy}\left(\vx\right)\right\Vert _{\HHt^{-1}}\nonumber \\
 & \geq\left\Vert \HHt^{-1}\left(\vb-p_{\vy}\left(\vx\right)\right)\right\Vert _{\HH_{\vy}}^{2}\left(1-\frac{1}{200\cdot B^{15}}-\frac{1}{400\cdot B^{14}}\right)\nonumber \\
 & \geq\left\Vert \HHt^{-1}\left(\vb-p_{\vy}\left(\vx\right)\right)\right\Vert _{\HH_{\vy}}^{2}\left(1-\frac{1}{200\cdot B^{14}}\right)\,.\label{eq:e9}
\end{align}
Returning to our inequality on the new residual, using (\ref{eq:e4}),
(\ref{eq:e3b}),(\ref{eq:e9}) we obtain:
\begin{align*}
 & \left\Vert \vrp\right\Vert _{\HHt^{-1}}^{2}\\
 & \leq\left\Vert \vr\right\Vert _{\HHt^{-1}}^{2}\cdot\left(1+\frac{1}{200\cdot B^{15}}\right)^{4}\\
 & +\frac{\left\Vert \HHt^{-1}\vr\right\Vert _{\HH_{\vy}}^{4}}{\left\Vert \HH_{\vy}\HHt^{-1}\vr\right\Vert _{\HHt^{-1}}^{4}}\left\Vert \HH_{\vy}\HHt^{-1}\vr\right\Vert _{\HHt^{-1}}^{2}\cdot\frac{\left(\frac{1}{1-\frac{1}{1000}}\right)^{8}}{\left(1-\frac{1}{400B^{18}}\right)^{4}}\cdot\left(1+\frac{1}{200\cdot B^{15}}\right)^{2}\\
 & -2\cdot\frac{\left\Vert \HHt^{-1}\vr\right\Vert _{\HH_{\vy}}^{2}}{\left\Vert \HH_{\vy}\HHt^{-1}\vr\right\Vert _{\HHt^{-1}}^{2}}\left\Vert \HHt^{-1}\vr\right\Vert _{\HH_{\vy}}^{2}\left(1-\frac{1}{200\cdot B^{14}}\right)^{2}\cdot\frac{1}{\left(1+\frac{1}{400B^{18}}\right)}\left(1+\frac{1}{200\cdot B^{15}}\right)^{2}\\
 & \leq\left\Vert \vr\right\Vert _{\HHt^{-1}}^{2}\cdot\left(1+\frac{1}{200\cdot B^{15}}\right)^{4}\\
 & -\frac{\left\Vert \HHt^{-1}\vr\right\Vert _{\HH_{\vy}}^{4}}{\left\Vert \HH_{\vy}\HHt^{-1}\vr\right\Vert _{\HHt^{-1}}^{2}} 
 \\
 &\cdot\left(2\left(1-\frac{1}{200\cdot B^{14}}\right)^{2}\cdot\frac{1}{\left(1+\frac{1}{400B^{18}}\right)}\left(1+\frac{1}{200\cdot B^{15}}\right)^{2}-\frac{\left(\frac{1}{1-\frac{1}{1000}}\right)^{8}}{\left(1-\frac{1}{400B^{18}}\right)^{4}}\cdot\left(1+\frac{1}{200\cdot B^{15}}\right)^{2}\right)\\
 & \leq\left\Vert \vr\right\Vert _{\HHt^{-1}}^{2}\cdot\left(1+\frac{1}{200\cdot B^{15}}\right)^{4}-\frac{\left\Vert \HHt^{-1}\vr\right\Vert _{\HH_{\vy}}^{4}}{\left\Vert \HH_{\vy}\HHt^{-1}\vr\right\Vert _{\HHt^{-1}}^{2}}\cdot\frac{9}{10}\,.
\end{align*}
Therefore, the failure condition 
\[
\frac{\left\Vert \vr\right\Vert _{\HHt^{-1}}^{2}-\left\Vert \vrp\right\Vert _{\HHt^{-1}}^{2}}{\left\Vert \vr\right\Vert _{\HHt^{-1}}^{2}}\leq\beta
\]
implies that
\[
1-\left(1+\frac{1}{200\cdot B^{15}}\right)^{4}+\frac{\left\Vert \HHt^{-1}\vr\right\Vert _{\HH_{\vy}}^{4}}{\left\Vert \HH_{\vy}\HHt^{-1}\vr\right\Vert _{\HHt^{-1}}^{2}\left\Vert \vr\right\Vert _{\HHt^{-1}}^{2}}\cdot\frac{9}{10}\leq1-\frac{\left\Vert \vrp\right\Vert _{\HHt^{-1}}^{2}}{\left\Vert \vr\right\Vert _{\HHt^{-1}}^{2}}\leq\beta
\]
and so
\[
\frac{\left\Vert \HHt^{-1}\vr\right\Vert _{\HH_{\vy}}^{4}}{\left\Vert \HH_{\vy}\HHt^{-1}\vr\right\Vert _{\HHt^{-1}}^{2}\left\Vert \vr\right\Vert _{\HHt^{-1}}^{2}}\leq\frac{10}{9}\left(\beta+\left(1+\frac{1}{200\cdot B^{15}}\right)^{4}-1\right)
\]
\[
\frac{\left\Vert \HHt^{-1}\vr\right\Vert _{\HH_{\vy}}^{4}}{\left\Vert \HH_{\vy}\HHt^{-1}\vr\right\Vert _{\HHt^{-1}}^{2}\left\Vert \vr\right\Vert _{\HHt^{-1}}^{2}}\leq\frac{20}{9}\beta\,,
\]
which just like in Lemma \ref{lem:prec-Richardson-progress-certificate}
shows that $\vr$ is an excentricity certificate, and concludes the
proof.
\end{proof}

\subsection{Proof of Lemma \ref{lem:precon-update-robust}\label{subsec:Proof-of-Lemma-precon-update-hess}}
\begin{proof}
It suffices to show that the rank-1 updates performed in Algorithm
\ref{alg:precon-update-linear-nohess} still suffice to decrease excentricity.
To do so, we analyze each of the two types of certificates. 

\paragraph{Type 1. }

Letting $\XX=\HH_{\vy}^{-1/2}\HHt\HH_{\vy}^{-1/2}$ and $\vu=\HH_{\vy}^{1/2}\HHt^{-1}\vr$,
we have $\vu^{\top}\XX\vu\geq\frac{1}{\sqrt{\frac{20}{9}\beta}}\vu^{\top}\vu$.
As the updated mandated by this certificate would require us, in order
to apply Lemma \ref{lem:excent-upd}, exact access to the norm $\left\Vert \HHt^{-1}\vr\right\Vert _{\HH_{\vy}}$
which we do not have available, we instead use a robust version (Lemma
\ref{lem:excent-upd-robust} in Appendix \ref{subsec:Robustness-Proofs}).
Since $\left\Vert \vu\right\Vert _{2}=\left\Vert \HHt^{-1}\vr\right\Vert _{\HH_{\vy}}$,
we have that per Lemma \ref{lem:hessian_norm_estimation}, $\left\Vert \vu\right\Vert _{2}\leq n_{\vy}\left(\HHt^{-1}\vr\right)\leq\left\Vert \vu\right\Vert _{2}\cdot\left(\frac{1}{1-\frac{1}{1000}}\right)^{2}$,
and thus using Lemma \ref{lem:excent-upd-robust}, we obtain that
changing the preconditioner to $\HHtp$ such that
\begin{align*}
&\HH^{-1/2}\HHtp\HH^{-1/2} 
\\
&=\HH^{-1/2}\HHt\HH^{-1/2}-\frac{\HH^{-1/2}\HHt\HH^{-1/2}\left(\HH^{1/2}\HHt^{-1}\vr\right)\left(\HH^{1/2}\HHt^{-1}\vr\right)^{\top}\HH^{-1/2}\HHt\HH^{-1/2}}{n_{\vy}\left(\HHt^{-1}\vr\right)^{2}+\left\Vert \HH^{1/2}\HHt^{-1}\vr\right\Vert _{\HH^{-1/2}\HHt\HH^{-1/2}}^{2}}\\
 & =\HH^{-1/2}\HHt\HH^{-1/2}-\frac{\HH^{-1/2}\vr\vr^{\top}\HH^{-1/2}}{n_{\vy}\left(\HHt^{-1}\vr\right)^{2}+\left\Vert \vr\right\Vert _{\HHt^{-1}}^{2}}\,,
\end{align*}
and equivalently setting
\[
\HHtp=\HHt-\frac{\vr\vr^{\top}}{n_{\vy}\left(\HHt^{-1}\vr\right)^{2}+\left\Vert \vr\right\Vert _{\HHt^{-1}}^{2}}\,,
\]
we obtain
\[
\exc\left(\HHtp\HH^{-1}\right)\leq\exc\left(\HHt\HH^{-1}\right)\cdot\frac{2}{\sqrt{1+\frac{1}{\left(\frac{1}{1-\frac{1}{1000}}\right)^{4}\sqrt{\frac{20}{9}\beta}}}}\leq\frac{2}{\sqrt{1+\frac{99}{100}\cdot\frac{1}{\sqrt{\frac{20}{9}\beta}}}}\,.
\]

\paragraph{Type 2. }

Letting $\XX=\HH_{\vy}^{-1/2}\HHt\HH_{\vy}^{-1/2}$ and $\vu=\HH_{\vy}^{1/2}\HHt^{-1}\vr$,
we have $\vu^{\top}\XX^{-1}\vu\geq\frac{1}{\sqrt{\frac{20}{9}\beta}}\vu^{\top}\vu$.
As the update mandated by this certificate would require us, in order
to apply Lemma \ref{lem:excent-upd}, exact access to the norm $\left\Vert \HHt^{-1}\vr\right\Vert _{\HH_{\vy}}$
and the vector $\HH\HHt^{-1}\vr$, which we do not have available.
Using Lemma using \ref{lem:hessian-vector-product-estimation}, we
obtain an approximation $p_{\vy}\left(\HHt^{-1}\vr\right)$ satisfying:
\[
\left\Vert p_{\vy}\left(\HHt^{-1}\vr\right)-\HH_{\vy}\HHt^{-1}\vr\right\Vert _{\HH_{\vy}^{-1}}\leq\frac{1}{400\cdot B^{20}}\left\Vert \HHt^{-1}\vr\right\Vert _{\HH_{\vy}}=\frac{1}{400\cdot B^{20}}\left\Vert \vu\right\Vert _{2}\,.
\]
Therefore letting $\vup=\HH_{\vy}^{-1/2}p_{\vy}\left(\HHt^{-1}\vr\right),$
we have that $\vup^{\top}\XX^{-1}\vup\geq\left(\frac{1-\frac{1}{100}}{1+\frac{1}{100}}\right)^{2}\cdot\vup^{\top}\vup.$
This is because
\begin{align*}
\frac{\vup^{\top}\XX^{-1}\vup}{\vup^{\top}\vup} & =\frac{\left\Vert \vup\right\Vert _{\XX^{-1}}^{2}}{\left\Vert \vup\right\Vert _{2}^{2}}=\frac{\left\Vert \vup-\vu+\vu\right\Vert _{\XX^{-1}}^{2}}{\left\Vert \vup-\vu+\vu\right\Vert _{2}^{2}}\geq\frac{\left(\left\Vert \vu\right\Vert _{\XX^{-1}}-\left\Vert \vup-\vu\right\Vert _{\XX^{-1}}\right)^{2}}{\left(\left\Vert \vu\right\Vert _{2}+\left\Vert \vup-\vu\right\Vert _{2}\right)^{2}}\\
 & =\frac{\left(\left\Vert \vu\right\Vert _{\XX^{-1}}-\left\Vert \HH_{\vy}^{-1/2}p_{\vy}\left(\HHt^{-1}\vr\right)-\HH_{\vy}^{-1/2}\HH_{\vy}\HHt^{-1}\vr\right\Vert _{\XX^{-1}}\right)^{2}}{\left(\left\Vert \vu\right\Vert _{2}+\left\Vert \HH_{\vy}^{-1/2}p_{\vy}\left(\HHt^{-1}\vr\right)-\HH_{\vy}^{-1/2}\HH_{\vy}\HHt^{-1}\vr\right\Vert _{2}\right)^{2}}\\
 & \geq\frac{\left(\left\Vert \vu\right\Vert _{\XX^{-1}}-\left\Vert \XX^{-1}\right\Vert \left\Vert p_{\vy}\left(\HHt^{-1}\vr\right)-\HH_{\vy}\HHt^{-1}\vr\right\Vert _{\HH_{\vy}^{-1}}\right)^{2}}{\left(\left\Vert \vu\right\Vert _{2}+\left\Vert p_{\vy}\left(\HHt^{-1}\nabla g\left(\vy\right)\right)-\HH_{\vy}\HHt^{-1}\vr\right\Vert _{\HH_{\vy}^{-1}}\right)^{2}}\\
 & \geq\frac{\left(\left\Vert \vu\right\Vert _{\XX^{-1}}-\frac{1}{400\cdot B^{20}}\left\Vert \vu\right\Vert _{2}\right)^{2}}{\left(\left\Vert \vu\right\Vert _{2}+\frac{1}{400\cdot B^{20}}\left\Vert \vu\right\Vert _{2}\right)^{2}}\\
 & \geq\frac{\left(\left\Vert \vu\right\Vert _{\XX^{-1}}-B^{2}\cdot\frac{1}{400\cdot B^{20}}\left\Vert \vu\right\Vert _{\XX^{-1}}\right)^{2}}{\left(\left\Vert \vu\right\Vert _{2}+\frac{1}{400\cdot B^{20}}\left\Vert \vu\right\Vert _{2}\right)^{2}}\\
 & \geq\left(1-\frac{1}{B^{17}}\right)\cdot\frac{\left\Vert \vu\right\Vert _{\XX^{-1}}^{2}}{\left\Vert \vu\right\Vert _{2}^{2}}\geq\left(1-\frac{1}{B^{17}}\right)\cdot\frac{1}{\sqrt{\frac{20}{9}\beta}}\,.
\end{align*}
Additionally, from Lemma \ref{lem:hessian_norm_estimation} we obtain
an approximation $\left\Vert \vu\right\Vert _{2}\leq n_{\vy}\left(\vu\right)\leq\left(\frac{1}{1-\frac{1}{1000}}\right)^{2}\left\Vert \vu\right\Vert _{2}$.
We now show that $n_{\vy}\left(\vu\right)$ is a good approximation
of $\left\Vert \vup\right\Vert _{2}$. Indeed following the same approach
as before, we have
\[
\left|\left\Vert \vup\right\Vert _{2}-\left\Vert \vu\right\Vert _{2}\right|\leq\left\Vert \vup-\vu\right\Vert _{2}\leq\left\Vert p_{\vy}\left(\HHt^{-1}\vr\right)-\HH_{\vy}\HHt^{-1}\vr\right\Vert _{\HH_{\vy}^{-1}}\leq\frac{1}{400\cdot B^{20}}\left\Vert \vu\right\Vert _{2}\,,
\]
and so, 
\[
\left\Vert \vu\right\Vert _{2}\left(1-\frac{1}{400\cdot B^{20}}\right)\leq\left\Vert \vup\right\Vert _{2}\leq\left\Vert \vu\right\Vert _{2}\left(1+\frac{1}{400\cdot B^{20}}\right)\,,
\]
which gives: 
\begin{align*}
\left\Vert \vup\right\Vert _{2}\leq\left(1+\frac{1}{400\cdot B^{20}}\right)n_{\vy}\left(\vu\right)\leq\left(\frac{1}{1-\frac{1}{1000}}\right)^{2}\cdot\frac{1+\frac{1}{400\cdot B^{20}}}{1-\frac{1}{400\cdot B^{20}}}\cdot\left\Vert \vup\right\Vert _{2}
\\
\leq
\left(\frac{1}{1-\frac{1}{1000}}\right)^{2}\cdot\left(1+\frac{1}{100\cdot B^{20}}\right)\cdot\left\Vert \vup\right\Vert _{2}\,.
\end{align*}
Thus per the first case of Lemma \ref{lem:excent-upd-robust}, updating
$\HHt$ to $\HHtp$ such that
\[
\XXp:=\HH_{\vy}^{-1/2}\HHtp\HH_{\vy}^{-1/2}=\HH_{\vy}^{-1/2}\HHt\HH_{\vy}^{-1/2}+\frac{\HH_{\vy}^{-1/2}p_{\vy}\left(\HHt^{-1}\vr\right)p_{\vy}\left(\HHt^{-1}\vr\right)^{\top}\HH_{\vy}^{-1/2}}{\left(1+\frac{1}{400\cdot B^{20}}\right)^{2}\cdot n_{\vy}\left(\HHt^{-1}\nabla g\left(\vy\right)\right)^{2}}
\]
and equivalently setting
\[
\HHtp=\HHt+\frac{p_{\vy}\left(\HHt^{-1}\vr\right)p_{\vy}\left(\HHt^{-1}\vr\right)^{\top}}{\left(1+\frac{1}{400\cdot B^{20}}\right)^{2}n_{\vy}\left(\HHt^{-1}\nabla g\left(\vy\right)\right)^{2}}
\]
 improves excentricity. in the sense that 
\[
\frac{\exc\left(\XXp\right)}{\exc\left(\XX\right)}\leq\frac{2}{\sqrt{1+\left(\frac{1}{\left(\frac{1}{1-\frac{1}{1000}}\right)^{2}\cdot\left(1+\frac{1}{100\cdot B^{20}}\right)}\right)^{2}\frac{1}{\sqrt{\frac{20}{9}\beta}}}}\leq\frac{2}{\sqrt{1+\frac{99}{100}\cdot\frac{1}{\sqrt{\frac{20}{9}\beta}}}}\,.
\]

Again, we note that the time to implement this update is dominated
by making a constant number of gradient queries, and performing a
constant number of matrix-vector multiplications involving $\HHt$
and $\HHt^{-1}$, which takes $O\left(n^{2}\right)$ time.

\paragraph{Change in preconditioner norm.}

For the purpose of this proof we only require slightly weaker guarantees
on the norm increases suffered by $\left\Vert \HHt\right\Vert $ or
$\left\Vert \HHt^{-1}\right\Vert $. Whenever the preconditioner norm
increases, it is because of a type 2 certificate, in which case we
have:
\[
\HHtp=\HHt+\frac{p_{\vy}\left(\HHt^{-1}\vr\right)p_{\vy}\left(\HHt^{-1}\vr\right)^{\top}}{\left(1+\frac{1}{400\cdot B^{20}}\right)^{2}n_{\vy}\left(\HHt^{-1}\vr\right)^{2}}\,,
\]
so
\[
\left\Vert \HHtp-\HHt\right\Vert \leq\frac{\left\Vert p_{\vy}\left(\HHt^{-1}\vr\right)\right\Vert _{2}^{2}}{n_{\vy}\left(\HHt^{-1}\vr\right)^{2}}
\]
and using Lemma \ref{lem:hessian-vector-product-estimation}, we know
that 
\begin{align*}
\left\Vert p_{\vy}\left(\HHt^{-1}\vr\right)-\HH_{\vy}\HHt^{-1}\vr\right\Vert _{\HH_{\vy}^{-1}} & \leq\frac{1}{400\cdot B^{20}}\left\Vert \HHt^{-1}\vr\right\Vert _{\HH_{\vy}}\,.,\\
\left\Vert p_{\vy}\left(\HHt^{-1}\vr\right)\right\Vert _{\HH_{\vy}^{-1}} & \leq\left\Vert \HHt^{-1}\vr\right\Vert _{\HH_{\vy}}\left(1+\frac{1}{400\cdot B^{20}}\right)\,,\\
\left\Vert p_{\vy}\left(\HHt^{-1}\vr\right)\right\Vert _{2} & \leq\left\Vert \HH_{\vy}^{-1}\right\Vert \cdot\left\Vert \HHt^{-1}\vr\right\Vert _{\HH_{\vy}}\left(1+\frac{1}{400\cdot B^{20}}\right)\,,
\end{align*}
Also since by Lemma \ref{lem:hessian_norm_estimation}, $\left\Vert \HHt^{-1}\vr\right\Vert _{\HH_{\vy}}\leq n_{\vy}\left(\HHt^{-1}\vr\right)\leq\left\Vert \vu\right\Vert _{2}\cdot\left(\frac{1}{1-\frac{1}{1000}}\right)^{2}$,
we conclude that 
\[
\left\Vert \HHtp-\HHt\right\Vert \leq\frac{\left(\left\Vert \HH_{\vy}^{-1}\right\Vert \cdot\left\Vert \HHt^{-1}\vr\right\Vert _{\HH_{\vy}}\left(1+\frac{1}{400\cdot B^{20}}\right)\right)^{2}}{\left\Vert \HHt^{-1}\vr\right\Vert _{\HH_{\vy}}^{2}}=\left\Vert \HH_{\vy}^{-1}\right\Vert ^{2}\left(1+\frac{1}{400\cdot B^{20}}\right)^{2}\,.
\]
Similarly, in the case where $\HHtp^{-1}$ increases, it is because
of a type 1 certificate, in which case we have:
\[
\HHtp^{-1}=\HHt^{-1}+\frac{\HHt^{-1}\vr\vr^{\top}\HHt^{-1}}{n_{\vy}\left(\HHt^{-1}\vr\right)^{2}}\,,
\]
and thus
\[
\left\Vert \HHtp^{-1}-\HHt^{-1}\right\Vert \leq\frac{\left\Vert \HHt^{-1}\vr\right\Vert _{2}^{2}}{n_{\vy}\left(\HHt^{-1}\vr\right)^{2}}\leq\frac{\left\Vert \HHt^{-1}\vr\right\Vert _{2}^{2}}{\left\Vert \HHt^{-1}\vr\right\Vert _{\HH_{\vy}}^{2}}=\frac{\left\Vert \HH_{\vy}^{1/2}\HHt^{-1}\vr\right\Vert _{\HH_{\vy}^{-1}}^{2}}{\left\Vert \HH_{\vy}^{1/2}\HHt^{-1}\vr\right\Vert _{2}^{2}}\leq\left\Vert \HH_{\vy}^{-1}\right\Vert \,.
\]
\end{proof}

\section{Proof from Section \ref{sec:path-following}}
\begin{lem}
\label{lem:ipm-prog}Let $g_{\mu}$ be self-concordant. Suppose that
$\left\Vert \nabla g_{\mu}\left(\vy\right)\right\Vert _{\HH_{\vy}^{-1}}\leq1/20$
and $\left\Vert \nabla g_{\mu}\left(\vy\right)-\HH_{\vy}\vx\right\Vert _{\HH_{\vy}^{-1}}\leq\epsilon\leq1/20$.
Then setting $\vyp=\vy-\vx$, 
\begin{align*}
\left\Vert \nabla g_{\mu}\left(\vyp\right)\right\Vert _{\HH_{\vyp}} & \leq2\epsilon+7\cdot\left\Vert \nabla g_{\mu}\left(\vy\right)\right\Vert _{\HH_{\vy}^{-1}}^{2}\,.
\end{align*}
Furthermore, 
\[
\left\Vert \vyp-\vy\right\Vert _{\HH_{\vy}}\leq\epsilon+\left\Vert \nabla g_{\mu}\left(\vy\right)\right\Vert _{\HH_{\vy}^{-1}}\,.
\]
\end{lem}

\begin{proof}
Let $\overline{\HH}=\int_{0}^{1}\HH_{\vy+\left(\vyp-\vy\right)t}dt$,
for which we know from self-concordance that 
\[
\max\left\{ \left\Vert \HH_{\vy}^{-1/2}\HH_{\vyp}\HH_{\vy}^{-1/2}\right\Vert ,\left\Vert \HH_{\vy}^{-1/2}\overline{\HH}\HH_{\vy}^{-1/2}\right\Vert \right\} \leq\left(\frac{1}{1-\left\Vert \vyp-\vy\right\Vert _{\HH_{\vy}}}\right)^{2}=\left(\frac{1}{1-\left\Vert \vx\right\Vert _{\HH_{\vy}}}\right)^{2}\,.
\]
Additionally we write 
\begin{equation}
\left\Vert \vx\right\Vert _{\HH_{\vy}}=\left\Vert \HH_{\vy}\vx\right\Vert _{\HH_{\vy}^{-1}}\leq\left\Vert \nabla g_{\mu}\left(\vy\right)-\HH_{\vy}\vx\right\Vert _{\HH_{\vy}^{-1}}+\left\Vert \nabla g_{\mu}\left(\vy\right)\right\Vert _{\HH_{\vy}^{-1}}\leq\epsilon+\left\Vert \nabla g_{\mu}\left(\vy\right)\right\Vert _{\HH_{\vy}^{-1}}\leq\frac{1}{10}\,.\label{eq:movement}
\end{equation}
Using these, we can bound the new gradient as:
\begin{align*}
\left\Vert \nabla g_{\mu}\left(\vyp\right)\right\Vert _{\HH_{\vy}^{-1}} & =\left\Vert \nabla g_{\mu}\left(\vy\right)+\overline{\HH}\left(\vyp-\vy\right)dt\right\Vert _{\HH_{\vy}^{-1}}\\
 & =\left\Vert \nabla g_{\mu}\left(\vy\right)-\overline{\HH}\vx\right\Vert _{\HH_{\vy}^{-1}}\\
 & =\left\Vert \nabla g_{\mu}\left(\vy\right)-\HH_{\vy}\vx+\left(\HH_{\vy}-\overline{\HH}\right)\vx\right\Vert _{\HH_{\vy}^{-1}}\\
 & \leq\epsilon+\left\Vert \left(\HH_{\vy}-\overline{\HH}\right)\vx\right\Vert _{\HH_{\vy}^{-1}}\\
 & =\epsilon+\left\Vert \left(\Id-\HH_{\vy}^{-1/2}\overline{\HH}\HH_{\vy}^{-1/2}\right)\HH_{\vy}^{1/2}\vx\right\Vert _{2}\\
 & \leq\epsilon+\left\Vert \Id-\HH_{\vy}^{-1/2}\overline{\HH}\HH_{\vy}^{-1/2}\right\Vert \left\Vert \HH_{\vy}^{1/2}\vx\right\Vert _{2}\\
 & \leq\epsilon+\left(1-\left(\frac{1}{1-\left\Vert \vx\right\Vert _{\HH_{\vy}}}\right)^{2}\right)\left\Vert \vx\right\Vert _{\HH_{\vy}}\\
 & \leq\epsilon+\frac{2\left\Vert \vx\right\Vert _{\HH_{\vy}}^{2}}{\left(1-\left\Vert \vx\right\Vert _{\HH_{\vy}}\right)^{2}}\\
 & \leq\epsilon+\left(\frac{10}{9}\right)^{2}\cdot2\left\Vert \vx\right\Vert _{\HH_{\vy}}^{2}\,.
\end{align*}
Using (\ref{eq:movement}) we obtain that
\begin{align*}
\left\Vert \nabla g_{\mu}\left(\vyp\right)\right\Vert _{\HH_{\vy}^{-1}} & \leq\epsilon+\left(\frac{10}{9}\right)^{2}\cdot2\left(\epsilon+\left\Vert \nabla g_{\mu}\left(\vy\right)\right\Vert _{\HH_{\vy}^{-1}}\right)^{2}\\
 & \leq\epsilon+\left(\frac{10}{9}\right)^{2}\cdot4\left(\epsilon^{2}+\left\Vert \nabla g_{\mu}\left(\vy\right)\right\Vert _{\HH_{\vy}^{-1}}^{2}\right)\\
 & \leq\epsilon\left(1+\left(\frac{10}{9}\right)^{2}\cdot\frac{1}{5}\right)+\left(\frac{10}{9}\right)^{2}\cdot4\cdot\left\Vert \nabla g_{\mu}\left(\vy\right)\right\Vert _{\HH_{\vy}^{-1}}^{2}\\
 & \leq\epsilon\left(1+\frac{1}{4}\right)+5\cdot\left\Vert \nabla g_{\mu}\left(\vy\right)\right\Vert _{\HH_{\vy}^{-1}}^{2}\,.
\end{align*}
Finally, changing the norm we have 
\begin{align*}
\left\Vert \nabla g_{\mu}\left(\vyp\right)\right\Vert _{\HH_{\vyp}^{-1}} & \leq\left\Vert \nabla g_{\mu}\left(\vyp\right)\right\Vert _{\HH_{\vy}^{-1}}\cdot\left\Vert \HH_{\vy}^{1/2}\HH_{\vyp}^{-1}\HH_{\vy}^{1/2}\right\Vert \\
 & \leq\left\Vert \nabla g_{\mu}\left(\vyp\right)\right\Vert _{\HH_{\vy}^{-1}}\cdot\left(\frac{1}{1-\left\Vert \vx\right\Vert _{\HH_{\vy}}}\right)^{2}\\
 & \leq\left(\epsilon\left(1+\frac{1}{4}\right)+5\cdot\left\Vert \nabla g_{\mu}\left(\vy\right)\right\Vert _{\HH_{\vy}^{-1}}^{2}\right)\cdot\frac{100}{81}\\
 & \leq2\epsilon+7\cdot\left\Vert \nabla g_{\mu}\left(\vy\right)\right\Vert _{\HH_{\vy}^{-1}}^{2}\,.
\end{align*}
\end{proof}
\begin{lem}
\label{lem:new-grad}Let $g_{\mu}\left(\vy\right)=\frac{\left\langle \vc,\vy\right\rangle }{\mu}+\phi\left(\vy\right)$.
Then $\nabla g_{\mu/\left(1-\delta\right)}\left(\vy\right)=\left(1+\delta\right)\nabla g_{\mu}\left(\vy\right)-\delta\nabla\phi\left(\vy\right)$.
\end{lem}

\begin{proof}
We prove the identity as follows:
\begin{align*}
\nabla g_{\mu\left(1-\delta\right)}\left(\vy\right) &=\frac{\vc}{\mu/\left(1+\delta\right)}+\nabla\phi\left(\vy\right)
=
\left(1+\delta\right)\left(\frac{\vc}{\mu}+\nabla\phi\left(\vy\right)\right)-\delta\nabla\phi\left(\vy\right)
\\
&=
\left(1+\delta\right)\nabla g_{\mu}\left(\vy\right)-\delta\nabla\phi\left(\vy\right)\,.
\end{align*}
\end{proof}

\bibliographystyle{alpha}
\bibliography{ref}

\end{document}